%% file: main.tex
\documentclass[letterpaper,twocolumn,10pt]{article}
\usepackage{usenix2019_v3}

\usepackage{amsmath,amssymb,amsfonts,amsthm}
\usepackage{centernot}
\usepackage{graphicx}
\usepackage{booktabs}
\usepackage{tabularx}
\usepackage{array}
\usepackage{float}

\makeatletter
\newcommand{\FloatBarrier}{%
  \ifx\@deferlist\@empty\else
    \clearpage
  \fi}
\makeatother

\usepackage{listings}
\usepackage{algorithm}
\usepackage{algpseudocode}
\usepackage{subcaption}
\usepackage{xspace}
\usepackage{xcolor}
\usepackage{tikz}
\usetikzlibrary{arrows.meta, backgrounds, positioning, calc, shapes.geometric, fit, shadows, matrix, patterns, decorations.pathreplacing}
\microtypesetup{spacing=false}
\AtBeginDocument{\DeclareMathAlphabet{\mathcal}{OMS}{cmsy}{m}{n}}

\newcolumntype{L}[1]{>{\raggedright\arraybackslash}p{#1}}
\newcolumntype{C}[1]{>{\centering\arraybackslash}p{#1}}
\newcolumntype{R}[1]{>{\raggedleft\arraybackslash}p{#1}}
\newcolumntype{Y}{>{\centering\arraybackslash}X}

\newcommand{\aas}{\mbox{AAS}\xspace}
\newcommand{\cac}{\mbox{CAC}\xspace}
\newcommand{\tct}{\mbox{TCT}\xspace}

\newcommand{\pdds}{\mbox{PDDS}\xspace}
\newcommand{\efd}{\mbox{EFD}\xspace}

\theoremstyle{definition}
\newtheorem{theorem}{Theorem}
\newtheorem{definition}{Definition}
\newtheorem{proposition}{Proposition}

\theoremstyle{definition}
\newtheorem{remark}{Remark}

\newtheorem{example}{Example}

\definecolor{slate}{RGB}{112,128,144}
\definecolor{emerald}{RGB}{16,163,127}
\definecolor{navy}{RGB}{24,49,83}
\definecolor{crimson}{RGB}{190,30,45}
\definecolor{amber}{RGB}{217,119,6}
\definecolor{cobalt}{RGB}{29,78,216}
\definecolor{teal}{RGB}{13,148,136}
\definecolor{purple}{RGB}{126,34,206}
\definecolor{softgray}{RGB}{247,249,252}
\definecolor{bordergray}{RGB}{218,224,233}
\definecolor{darkslate}{RGB}{60,72,88}
\definecolor{linegray}{RGB}{140,155,175}

\hypersetup{
  pdftitle={Before Agents Act: Assurance-Aware Semantic Scheduling for Evidence Acquisition in Distributed Systems},
  pdfauthor={Jun He; Deying Yu},
  pdfsubject={Scheduling evidence acquisition under risk, freshness, and epistemic dependencies},
  pdfkeywords={semantic scheduling, assurance obligations, cognitive admission control, epistemic fault domains, evidence acquisition, post-deterministic distributed systems}
}

\begin{document}
\raggedbottom

\title{\bf Before Agents Act: Assurance-Aware Semantic Scheduling\\for Evidence Acquisition in Distributed Systems}

\ifdefined\AASAnonymous
\author{Anonymous authors}
\hypersetup{pdfauthor={}}
\else
\author{
  {\rm Jun He}\\
  OpenKedge.io
  \and
  {\rm Deying Yu}\\
  OpenKedge.io
}
\fi

\maketitle

\begin{abstract}
Tool-using agents can initiate consequential infrastructure changes, yet evidence required for admission may expire while other checks run or depend on a shared fault domain. We formulate evidence acquisition as joint witness selection and scheduling under quorum, diversity, freshness, deadline, and resource constraints. \textbf{Assurance-Aware Semantic Scheduling (\aas)} combines integer-program selection, dispatch-aware temporal scheduling, bounded diagnostic expansion, and receipt-aware repair. Formal results state the assumptions needed for dispatch-time freshness and finite diagnostic expansion. In three generated infrastructure workloads, \aas produces 1,075/1,200 valid candidates versus 647/1,200 for constraint-aware forward scheduling; stale candidates fall from 440 to 12. Paired sensitivity studies reuse the same instances and operation latency draws across parameter settings. A corrected timeout intervention finds 18/20 admissions with repair or full resynthesis versus 0/20 for a static plan, with lower committed cost when receipts are reused. On 20 constructed cases requiring a certified decomposition cut, refinement recovers an oracle-matching feasible plan every time. These are controlled simulation results; the bounded oracle shares a temporal search component, and transfer to deployed systems remains untested.
\end{abstract}

\input{sections/01-introduction}
\input{sections/02-background-motivation}
\input{sections/03-problem-formulation}
\input{sections/04-temporal-freshness}
\input{sections/05-epistemic-scheduling}
\input{sections/06-consequential-acquisition}
\input{sections/07-adaptive-replanning}
\input{sections/08-reference-scheduler}
\input{sections/09-evaluation-design}
\input{sections/10-discussion-limitations}
\input{sections/11-related-work}
\input{sections/12-conclusion}

\smallskip
\noindent\textbf{AI-Use Disclosure.} OpenAI Codex and Google Antigravity assisted with draft structuring, language and notation editing, implementation debugging, experiment auditing, figure preparation, and reference formatting. The authors are responsible for verifying the code, proofs, sources, and reported results.

\bibliographystyle{unsrt}
\bibliography{refs}

\appendix
\input{appendices/a-formal-proofs}
\input{appendices/b-operation-catalogue}
\input{appendices/c-open-problems}

\end{document}

%% file: sections/01-introduction.tex
\section{Introduction}
\label{sec:introduction}

Agents that operate shared infrastructure may propose database failover, network isolation, access-control changes, or deployment rollouts. Authorization alone does not establish that the current system state supports a proposed action. Cognitive Admission Control (\cac)~\cite{he2026cac} expresses this additional requirement as risk-conditioned obligations $\Omega(q,s) = \{\omega_1, \ldots, \omega_k\}$ for action $q$ in modeled state $s$. The admission gateway discharges each obligation against typed evidence before allowing dispatch.

The gateway defines the evidence required for admission. It does not select verifiers or schedule their execution. The resulting systems question is:
\begin{center}
\emph{Given a set of assurance obligations, how should an autonomous system acquire, schedule, refresh, and compose the evidence needed to satisfy them under cost, latency, freshness, dependency, and fault-domain constraints?}
\end{center}

\begin{figure}[t]
\centering
\begin{tikzpicture}[
  box/.style={draw=navy,rounded corners=2pt,fill=softgray,align=center,text width=6.8cm,inner sep=4pt,font=\small},
  arr/.style={-{Latex[length=2.5mm]},thick,darkslate},
  annot/.style={font=\scriptsize,text=cobalt,align=center}]
  \node[box] (cac) {\textbf{Cognitive Admission Control (\cac)}\\\footnotesize\textit{``What evidence is required?''}\\\scriptsize Maps risk $\rho(q,s)$ to obligations $\Omega(q,s)$};
  \node[box,below=0.7cm of cac] (aas) {\textbf{Assurance-Aware Semantic Scheduling (\aas)}\\\footnotesize\textit{``How is that evidence acquired in time?''}\\\scriptsize Selects and schedules operations in plan $\pi$};
  \node[box,below=0.7cm of aas] (tct) {\textbf{Transactional Context Tracking \& Execution (\tct)}\\\footnotesize\textit{``When can the action commit?''}\\\scriptsize Enforces atomic dispatch and single-use admission};
  \draw[arr] (cac) -- node[midway,right=2pt,annot] {Obligations $\Omega(q,s)$} (aas);
  \draw[arr] (aas) -- node[midway,right=2pt,annot] {Witness Manifest $\mathcal{W}_q$ \& Cert $\mathcal{C}_q$} (tct);
\end{tikzpicture}
\caption{The Post-Deterministic Distributed Systems (\pdds) control stack. \aas bridges the architectural gap between declarative admission specification (\cac) and transactional execution (\tct).}
\label{fig:pdds-progression}
\end{figure}
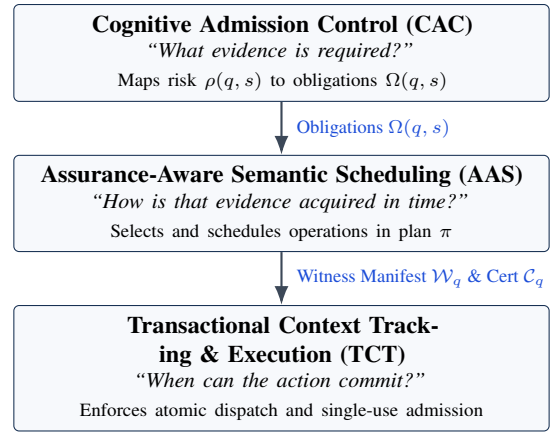

Consider a concrete operational scenario: an autonomous agent proposes a failover of a primary PostgreSQL cluster. \cac dictates that admission requires:
(i)~a verified cluster topology snapshot,
(ii)~an independent confirmation of replication lag $\le 100$\,ms,
(iii)~a cryptographically signed node fencing receipt,
(iv)~an Epistemic Fault Domain (\efd) diversity cut $\kappa_E \ge 2$, and
(v)~evidence freshness $\Delta t \le 5$\,s.
Although the admission requirements are explicit, the execution runtime faces a complex web of operational dilemmas:
\begin{itemize}
    \item \textbf{Operation Selection \& Multi-Obligation Coverage:} Should the runtime query low-level telemetry, invoke a heavyweight formal verification tool, run an in-memory simulation, or query a redundant model verifier? An attestation from a consensus coordinator might satisfy both topology and fencing obligations simultaneously.
    \item \textbf{Temporal Freshness Intersection:} If fencing verification requires 6\,s, while replication lag evidence expires after 5\,s, a serial schedule $A \to B \to C$ will arrive at dispatch with stale replication evidence. The scheduler must coordinate operations such that their validity intervals overlap at dispatch time: $\bigcap_{i} \operatorname{ValidityInterval}(e_i) \neq \emptyset$.
    \item \textbf{Epistemic Dependencies vs. Cost:} An optimizer selecting the cheapest verifiers may select two distinct services that secretly depend on the same faulty upstream telemetry broker, violating $\kappa_E \ge 2$ and inducing common-mode admission failure.
    \item \textbf{Consequential Diagnostic Actions:} To verify rollback feasibility, the system might need to execute an active database probe or snapshot lock. But that diagnostic probe is itself a mutating, consequential operation that requires its own admission certificate, risking recursive deadlocks.
    \item \textbf{Adaptive Replanning Under Shifted Risk:} If early evidence reveals that the candidate replica is partially degraded, the risk profile shifts from $\rho_0$ to $\rho_1$, altering the required obligation set $\Omega_0 \to \Omega_1$ and requiring dynamic schedule adaptation.
\end{itemize}

The timing problem matters because agent tool use increasingly crosses an authorization boundary at execution time~\cite{toolguardian2026}. Additional reasoning cannot replace a current observation of a partitioned switch, and an early observation can expire before a slow independent check completes. Verification-portfolio work, including VP-CONTROL~\cite{vpcontrol2026}, studies verifier selection, evidence-source diversity, and commit-time guards; \aas adds an explicit schedule for expiring operational evidence and for diagnostics that themselves require admission. Classical task scheduling supplies precedence and capacity constraints, but the admissibility of an evidence plan also depends on witness independence and validity at dispatch.

\subsection{Contributions}
\label{sec:contributions}

We study \textbf{Assurance-Aware Semantic Scheduling (\aas)} through five contributions:
\begin{enumerate}
    \item \textbf{Problem formulation.} ASP selects and schedules evidence-producing operations under hard admission, budget, and deadline constraints. It separates prospective feasibility from the gateway's decision on realized receipts (Section~\ref{sec:problem-formulation}).
    \item \textbf{Temporal scheduling.} A dispatch-time freshness criterion and backward scheduler place short-lived observations after slower checks when precedence permits. The guarantee is conditional on modeled latency and clock bounds; runtime receipts are checked again at the gateway (Section~\ref{sec:temporal-freshness}).
    \item \textbf{Fault-domain-aware selection.} An integer program assigns witnesses to obligations while enforcing modeled \efd cuts and charging shared operations once. A decomposed solver coordinates selection with temporal feasibility; cost optimality requires exact subproblem search and full cost modeling (Section~\ref{sec:epistemic-scheduling}).
    \item \textbf{Consequential diagnostics and repair.} Risk stratification bounds recursive sub-admission under explicit timeout assumptions. A generation-aware controller replans from the live clock, excludes failed operations, and reuses receipts only while their original observations remain valid (Sections~\ref{sec:consequential-acquisition} and~\ref{sec:adaptive-replanning}).
    \item \textbf{Controlled evaluation.} Generated infrastructure workloads and paired sensitivity studies measure admission, freshness, diversity, and cost. Selected-operation timeout trials isolate receipt reuse from resynthesis and a static control; constructed cut-producing cases test both supported refinement certificates against a bounded exact oracle (Section~\ref{sec:evaluation-design}).
\end{enumerate}

Section~\ref{sec:reference-scheduler} gives the reference controller and algorithms; Section~\ref{sec:evaluation-design} describes the evaluation.

%% file: sections/02-background-motivation.tex
\section{Background and Motivation}
\label{sec:background-motivation}

\subsection{Cognitive Admission Control (\cac) Recap}
\label{sec:cac-recap}

Cognitive Admission Control (\cac)~\cite{he2026cac} defines an admission boundary for agentic systems operating in high-consequence environments. When an agent proposes a mutating action $q$ at modeled state $s$, the control plane evaluates a multidimensional risk vector:
\begin{equation}
\rho(q,s) = \langle C, B_{\mathrm{blast}}, I, U, D_{\mathrm{exp}}, P \rangle \in \mathcal{R},
\label{eq:risk-vector}
\end{equation}
capturing consequence severity ($C$), blast radius ($B_{\mathrm{blast}}$), irreversibility ($I$), observational uncertainty ($U$), dependency exposure ($D_{\mathrm{exp}}$), and adverse plausibility ($P$). Based on $\rho(q,s)$, the admission policy $\Pi$ resolves a set of mandatory assurance obligations:
\begin{equation}
\Omega_{\Pi}(q,s) = F_{\Pi}(\rho(q,s), q, s) = \{\omega_1, \omega_2, \ldots, \omega_k\}.
\label{eq:obligation-derivation}
\end{equation}
Each obligation $\omega \in \Omega$ defines a formal predicate $\phi_\omega$, acceptable evidence classes $\mathcal{E}_\omega$, structural quorum constraints, validity duration $\Delta t_\omega$, and epistemic diversity requirements $\kappa_E(\omega) \ge k$.

To satisfy an obligation $\omega$, the system must provide a set of cryptographic evidence receipts $W \subseteq E$. The evaluator computes a deterministic ternary discharge:
\begin{equation}
\begin{aligned}
\operatorname{Discharge}(\omega, E, s, \chi) \in \{&\textsc{Satisfied}(W),\\
&\textsc{Violated}(W), \textsc{Unknown}(c)\},
\end{aligned}
\label{eq:ternary-discharge}
\end{equation}
where $\chi$ denotes the evaluation context. If all required obligations return $\textsc{Satisfied}$, the controller mints a single-use admission certificate $\mathcal{C}_q$ containing an immutable proposal digest, the witness manifest $\mathcal{W}_q = \{(\omega, W_\omega)\}$, state version guards $G_q$, and an expiry bound $t_{\mathrm{exp}}$. The execution gateway validates $\mathcal{C}_q$ and performs an atomic compare-and-swap on the certificate nonce before forwarding the action to external infrastructure.

\subsection{Reasoning vs. Evidence Acquisition}
\label{sec:reasoning-vs-evidence}

A critical theoretical finding in \cac is the non-fungibility of cognitive resources and empirical evidence. Let $R$ denote a cognitive resource budget encompassing model reasoning tokens $T$, context window $C_{\mathrm{ctx}}$, tool invocations $A_{\mathrm{tool}}$, and verification compute $V_{\mathrm{cpu}}$. The epistemic lifecycle follows a strict two-stage transformation:
\begin{equation}
R \xrightarrow{\text{acquisition}} E \xrightarrow{\text{discharge under }\Pi} \text{Admission Verdict}.
\label{eq:two-stage}
\end{equation}
Allocating additional cognitive compute (e.g., test-time reasoning tokens~\cite{snell2024scaling} or self-reflection~\cite{shinn2023reflexion}) can synthesize more sophisticated proof strategies or optimize query parameters. However, \emph{compute cannot conjure external truth}: no amount of internal reasoning can prove that an external replication stream is unpartitioned without acquiring an authentic, current receipt from an external observation channel.

Consequently, cognitive resource planning must not be treated as an isolated prompt-engineering problem. Instead, cognitive compute is simply one resource category within a general systems scheduling problem that allocates both computational reasoning and empirical observation bandwidth.

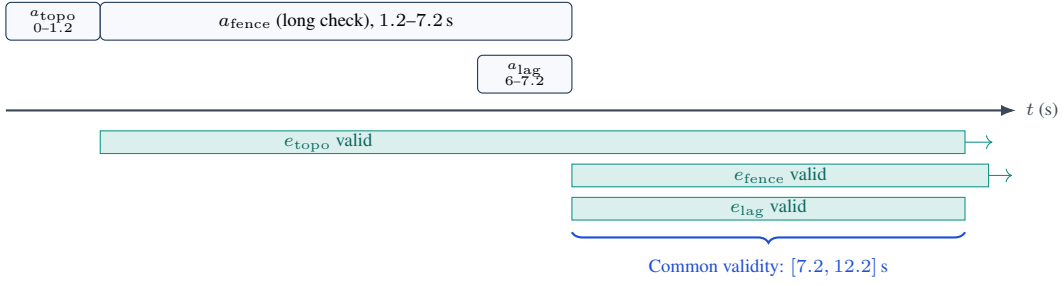
\begin{figure*}[t]
\centering
\begin{tikzpicture}[
  x=1.04cm,y=1cm,
  timeaxis/.style={-{Latex[length=2mm]},thick,darkslate},
  task/.style={draw=navy,fill=softgray,rounded corners=2pt},
  validity/.style={fill=teal!15,draw=teal},
  expired/.style={fill=crimson!12,draw=crimson}]

  \node[font=\bfseries\small,anchor=west] at (0,6.65) {A. Serial forward schedule: lag evidence expires before dispatch};
  \draw[timeaxis] (0,5.52) -- (12.85,5.52) node[right,font=\scriptsize] {$t$ (s)};
  \draw[task] (0,5.72) rectangle (1.2,6.22);
  \draw[task] (1.2,5.72) rectangle (2.4,6.22);
  \draw[task] (2.4,5.72) rectangle (8.4,6.22);
  \node[font=\tiny,align=center] at (0.6,5.97) {$a_{\mathrm{topo}}$\\$0$--$1.2$};
  \node[font=\tiny,align=center] at (1.8,5.97) {$a_{\mathrm{lag}}$\\$1.2$--$2.4$};
  \node[font=\scriptsize] at (5.4,5.97) {$a_{\mathrm{fence}}$ (fencing verification), $2.4$--$8.4$\,s};
  \draw[expired] (2.4,4.92) rectangle (7.4,5.22);
  \node[font=\scriptsize,text=crimson!85!black] at (4.9,5.07) {$e_{\mathrm{lag}}$ valid: $[2.4,7.4]$\,s};
  \draw[validity] (8.4,4.47) rectangle (12.4,4.77);
  \draw[teal,->] (12.4,4.62) -- (12.7,4.62);
  \node[font=\scriptsize,text=teal!70!black] at (10.4,4.62) {$e_{\mathrm{fence}}$ valid after $8.4$\,s};
  \draw[crimson,thick,dashed] (7.4,4.86) -- (7.4,5.45);
  \draw[darkslate,thick,dotted] (8.4,4.43) -- (8.4,5.45);
  \node[font=\scriptsize,text=crimson,anchor=west] at (2.4,4.27) {$e_{\mathrm{lag}}$ expires at $7.4$\,s $<$ dispatch at $8.4$\,s};

  \node[font=\bfseries\small,anchor=west] at (0,3.57) {B. Assurance-aware schedule: all receipts are fresh at dispatch};
  \draw[timeaxis] (0,1.69) -- (12.85,1.69) node[right,font=\scriptsize] {$t$ (s)};
  \draw[task] (0,2.62) rectangle (1.2,3.12);
  \draw[task] (1.2,2.62) rectangle (7.2,3.12);
  \draw[task] (6.0,1.92) rectangle (7.2,2.42);
  \node[font=\tiny,align=center] at (0.6,2.87) {$a_{\mathrm{topo}}$\\$0$--$1.2$};
  \node[font=\scriptsize] at (4.2,2.87) {$a_{\mathrm{fence}}$ (long check), $1.2$--$7.2$\,s};
  \node[font=\tiny,align=center] at (6.6,2.17) {$a_{\mathrm{lag}}$\\$6$--$7.2$};
  \draw[validity] (1.2,1.12) rectangle (12.2,1.42);
  \draw[teal,->] (12.2,1.27) -- (12.55,1.27);
  \node[font=\scriptsize,text=teal!70!black] at (4.1,1.27) {$e_{\mathrm{topo}}$ valid};
  \draw[validity] (7.2,0.68) rectangle (12.5,0.98);
  \draw[teal,->] (12.5,0.83) -- (12.8,0.83);
  \node[font=\scriptsize,text=teal!70!black] at (9.85,0.83) {$e_{\mathrm{fence}}$ valid};
  \draw[validity] (7.2,0.24) rectangle (12.2,0.54);
  \node[font=\scriptsize,text=teal!70!black] at (9.7,0.39) {$e_{\mathrm{lag}}$ valid};
  \draw[decorate,decoration={brace,amplitude=4pt,mirror},thick,cobalt] (7.2,0.08) -- (12.2,0.08)
    node[midway,below=6pt,font=\scriptsize,text=cobalt] {Common validity: $[7.2,12.2]$\,s};

\end{tikzpicture}
\caption{Motivating failure under naive scheduling versus Assurance-Aware Semantic Scheduling. In Case~A, serial forward execution causes short-lived evidence ($e_{\mathrm{lag}}$, validity 5\,s, finished at $t=2.4$\,s) to expire at $t=7.4$\,s while long-latency fencing verification ($a_{\mathrm{fence}}$, latency 6\,s) completes at $t=8.4$\,s, aborting dispatch. In Case~B, \aas pre-fetches long-lived topology evidence, schedules long-latency fencing early, and fires $a_{\mathrm{lag}}$ Just-in-Time on $[6.0, 7.2]$\,s, guaranteeing a valid dispatch freshness intersection $[7.2, 12.2]$\,s.}
\label{fig:motivating-timeline}
\end{figure*}

\subsection{Motivating Failures in Naive Schedulers}
\label{sec:motivating-failures}

To demonstrate why existing schedulers fail to satisfy assurance obligations, consider three concrete pathologies:

\paragraph{Pathology 1: Dispatch-Time Freshness Decay.}
Suppose a database failover requires three obligations: $\omega_{\mathrm{topo}}$ (cluster topology, valid for 30\,s), $\omega_{\mathrm{lag}}$ (replication lag under 100\,ms, valid for 5\,s), and $\omega_{\mathrm{fence}}$ (independent confirmation that the dead primary is isolated, valid for 10\,s). Measuring topology takes 1.2\,s, measuring lag takes 1.2\,s, while executing fencing verification requires 6.0\,s. A standard serial workflow engine executes tasks in declared sequence ($\omega_{\mathrm{topo}} \to \omega_{\mathrm{lag}} \to \omega_{\mathrm{fence}}$), as illustrated in Figure~\ref{fig:motivating-timeline}(Case~A). Under serial dispatch, $a_{\mathrm{topo}}$ executes on $[0, 1.2]$\,s, $a_{\mathrm{lag}}$ on $[1.2, 2.4]$\,s (yielding receipt $e_{\mathrm{lag}}$ valid on $[2.4, 7.4]$\,s), and $a_{\mathrm{fence}}$ on $[2.4, 8.4]$\,s. By the time fencing verification concludes at $t=8.4$\,s, the replication lag receipt issued at $t=2.4$\,s has already expired at $t=7.4$\,s ($7.4 < 8.4$). At the execution gateway, the certificate is rejected for stale evidence. In contrast, \aas pre-fetches long-lived topology evidence, executes long-latency fencing on $[1.2, 7.2]$\,s, and launches $a_{\mathrm{lag}}$ Just-in-Time on $[6.0, 7.2]$\,s, ensuring all validity intervals overlap at dispatch $t_{\mathrm{dispatch}} = 7.2$\,s (Case~B).

\paragraph{Pathology 2: Common-Mode Epistemic Correlation.}
Assume $\omega_{\mathrm{lag}}$ requires an Epistemic Fault Domain cut $\kappa_E \ge 2$, mandating two structurally independent witnesses. The scheduler can query three verifiers:
\begin{itemize}
    \item Verifier $V_A$: cost \$0.01, latency 400\,ms, queries Prometheus node-exporter $S_1$.
    \item Verifier $V_B$: cost \$0.01, latency 450\,ms, queries CloudWatch agent $S_2$.
    \item Verifier $V_C$: cost \$0.10, latency 2000\,ms, executes a direct Postgres replication query $S_3$.
\end{itemize}
A cost-minimizing scheduler (such as classic knapsack or greedy cost solvers) selects $\{V_A, V_B\}$ for a total cost of \$0.02. However, unbeknownst to the scheduler, both $S_1$ and $S_2$ pull their metrics from a shared telemetry proxy that has stalled due to memory pressure. Both verifiers return identical, stale lag reports. Because they share an underlying epistemic fault domain ($\mathcal{X}(V_A) \cap \mathcal{X}(V_B) \neq \emptyset$), their structural diversity cut is $\kappa_E = 1$. The admission controller rejects the pair. A semantic scheduler must understand epistemic dependency topology: it must recognize that $\{V_A, V_C\}$ (\$0.11) is the minimal \emph{admissible} set, whereas $\{V_A, V_B\}$ is an inadmissible waste of resources.

\paragraph{Pathology 3: Remediation Deadlock from Consequential Diagnostics.}
To satisfy an obligation regarding storage rollback feasibility, the scheduler selects a diagnostic operation $a_{\mathrm{diag}}$ that executes an ephemeral filesystem freeze and snapshot probe. However, freezing production storage is itself a consequential operation that carries high risk ($\rho(a_{\mathrm{diag}}, s) \ge \tau_{\Pi}$). Under \cac, $a_{\mathrm{diag}}$ cannot execute without its own admission certificate, which demands evidence of cluster quiescence! If the scheduler spawns recursive diagnostics without well-foundedness constraints, the system enters an infinite admission cycle or deadlocks under shared budget depletion.

%% file: sections/03-problem-formulation.tex
\section{The Assurance Scheduling Problem}
\label{sec:problem-formulation}

The \textbf{Assurance Scheduling Problem (ASP)} maps declarative \cac obligations to a feasible, cost-aware evidence-acquisition plan. Optimality applies only under the exact-search and full-model conditions stated in Section~\ref{sec:joint-cost-efd}.

\subsection{Formal Problem Inputs}
\label{sec:formal-inputs}

An instance of the Assurance Scheduling Problem is defined by the 5-tuple:
\begin{equation}
\mathcal{P} = \langle \Omega, \mathcal{A}, B, \mathcal{D}, \mathcal{T} \rangle,
\label{eq:asp-tuple}
\end{equation}
where each component is specified as follows:

\begin{definition}[Assurance Obligations $\Omega$]
$\Omega = \{\omega_1, \ldots, \omega_k\}$ is the finite set of required obligations emitted by the \cac policy resolver $\Omega_{\Pi}(q,s) = F_\Pi(\rho(q,s), q, s)$. Each obligation $\omega = \langle \phi_\omega, \mathcal{E}_\omega, \Delta t_\omega, \kappa_E^{\min}(\omega), \mathcal{Q}_\omega \rangle$ specifies:
\begin{itemize}
    \item a semantic predicate $\phi_\omega: \mathcal{S}^* \times \mathbb{E}_{\mathrm{inst}}^* \to \{\text{true}, \text{false}\}$,
    \item acceptable evidence classes $\mathcal{E}_\omega \subseteq \mathbb{E}$ (e.g., attestation, telemetry receipt, proof trace),
    \item a maximum allowable evidence age $\Delta t_\omega \in \mathbb{R}^+$,
    \item a minimum epistemic diversity cut threshold $\kappa_E^{\min}(\omega) \in \mathbb{N}_{\ge 1}$,
    \item a quorum rule $\mathcal{Q}_\omega = \langle k_\omega, \mathrm{pol}_\omega \rangle$ specifying required witness counts ($k_\omega \in \mathbb{N}_{\ge 1}$) and authorization polarities.
\end{itemize}
\end{definition}

\begin{definition}[Assurance Operations $\mathcal{A}$]
$\mathcal{A} = \{a_1, \ldots, a_m\}$ represents the catalogue of available evidence-producing operations. Each operation $a \in \mathcal{A}$ is a typed systems action characterized by:
\begin{equation}
\begin{aligned}
a = \langle &c(a), \ell(a), \hat{\ell}(a), \operatorname{tok}(a), \operatorname{Cov}(a), \\
            &\Delta t(a), \operatorname{Pre}(a), \mathcal{X}(a), \rho(a) \rangle,
\end{aligned}
\label{eq:assurance-op-signature}
\end{equation}
where:
\begin{itemize}
    \item $c(a) \in \mathbb{R}_{\ge 0}$ is financial or monetary cost (e.g., API charges, egress fees),
    \item $\ell(a) \in \mathbb{R}^+$ is nominal execution latency, and $\hat{\ell}(a) \ge \ell(a)$ is a conservative upper bound accounting for network jitter and tail variance,
    \item $\operatorname{tok}(a) \in \mathbb{N}_0$ is cognitive token consumption (for LLM verifiers or reasoning-based checkers),
    \item $\operatorname{Cov}(a) \subseteq \Omega \times \mathbb{E}$ denotes the modeled potential coverage: pairs $(\omega, \varepsilon)$ indicating that $a$ can produce evidence of class $\varepsilon \in \mathcal{E}_\omega$ relevant to evaluating obligation $\omega$,
    \item $\Delta t(a) \in \mathbb{R}^+$ is the intrinsic physical lifespan of evidence emitted by $a$,
    \item $\operatorname{Pre}(a) \subseteq \mathcal{A}$ is the set of operational prerequisites that must successfully terminate before $a$ can execute,
    \item $\mathcal{X}(a) \subseteq \mathbb{F}$ is the epistemic exposure set, enumerating all underlying shared fault domains (e.g., host kernels, DNS resolvers, metrics daemons, model weights) in the fault domain universe $\mathbb{F} = \{f_1, \ldots, f_p\}$,
    \item $\rho(a) \in \mathcal{R}$ is the operational risk vector of $a$, determining whether $a$ is non-consequential ($\rho(a) \prec \tau_{\Pi}$) or consequential ($\rho(a) \succeq \tau_{\Pi}$).
\end{itemize}
When operation $a$ executes, it produces a non-empty set of concrete evidence receipts $\operatorname{Emit}(a) \subseteq \mathbb{E}_{\mathrm{inst}}$. A single operation $a$ may emit multiple receipts, and a single receipt $e \in \operatorname{Emit}(a)$ may supply witness evidence supporting multiple obligations $\omega \in \Omega$. Each receipt $e$ records an observation timestamp $t_{\mathrm{obs}}(e) \le C(a)$, where $C(a)$ is the completion timestamp of $a$.
\end{definition}

\begin{definition}[Resource Budgets $B$]
$B = \langle B_{\$}, B_{\mathrm{lat}}, B_{\mathrm{tok}}, B_{\mathrm{risk}}, K_{\max} \rangle$ bounds total allowable dollar expenditure, end-to-end wall-clock latency, cognitive reasoning tokens, cumulative operational risk perturbation, and the maximum number of concurrent active operations $K_{\max} \in \mathbb{N}_{\ge 1} \cup \{\infty\}$.
\end{definition}

\begin{definition}[Assurance Dependencies $\mathcal{D}$]
$\mathcal{D} = \langle D_{\mathrm{op}}, D_{\mathrm{epi}} \rangle$ encodes:
\begin{itemize}
    \item Operational precedence constraints $D_{\mathrm{op}} \subseteq \mathcal{A} \times \mathcal{A}$, where $(a_i, a_j) \in D_{\mathrm{op}}$ implies $a_j$ cannot start until $a_i$ successfully completes ($\bigcup_{a \in \mathcal{A}} \{(u, a) \mid u \in \operatorname{Pre}(a)\} \subseteq D_{\mathrm{op}}$),
    \item Epistemic dependency topology $D_{\mathrm{epi}} = (\mathcal{A} \cup \mathbb{F}, \mathcal{E}_{\mathrm{dep}})$, defining the bipartite graph linking operations to shared failure domains ($\mathcal{E}_{\mathrm{dep}} = \{(a, f) \mid a \in \mathcal{A}, f \in \mathcal{X}(a)\}$).
\end{itemize}
\end{definition}

\begin{definition}[Temporal Constraints $\mathcal{T}$]
$\mathcal{T} = \langle t_0, T_{\mathrm{dead}}, \delta_{\mathrm{exec}}, \delta_{\mathrm{verify}}, \delta_{\mathrm{mint}}, \delta_{\mathrm{disp}}, \epsilon_{\mathrm{skew}} \rangle$ specifies the scheduling epoch, proposal deadline, protected execution window, receipt-verification time, certificate-minting time, gateway dispatch delay, and clock uncertainty bound, respectively. All delays are nonnegative. Define the post-completion readiness delay once as $\delta_{\mathrm{ready}}=\delta_{\mathrm{verify}}+\delta_{\mathrm{mint}}+\delta_{\mathrm{disp}}$.
\end{definition}

\subsection{Assurance Plan Representation}
\label{sec:plan-representation}

An assurance plan $\pi$ is a scheduled directed acyclic graph:
\begin{equation}
\pi = \langle V_\pi, E_\pi, \sigma_\pi, \mathcal{M}_\pi \rangle,
\label{eq:plan-dag}
\end{equation}
where:
\begin{itemize}
    \item $V_\pi \subseteq \mathcal{A}$ is the subset of scheduled operations,
    \item $E_\pi \subseteq V_\pi \times V_\pi$ is the set of precedence edges satisfying $D_{\mathrm{op}} \cap (V_\pi \times V_\pi) \subseteq E_\pi$,
    \item $\sigma_\pi: V_\pi \to [t_0, T_{\mathrm{dead}}]$ assigns a scheduled start time to each operation,
    \item $\mathcal{M}_\pi: \Omega \to 2^{V_\pi}$ assigns candidate witness operations $\mathcal{M}_\pi(\omega) \subseteq V_\pi$ to each obligation $\omega \in \Omega$.
\end{itemize}
Under conservative latency bounds $\hat{\ell}(a)$, the modeled completion time of operation $a \in V_\pi$ is $\hat{C}_\pi(a) = \sigma_\pi(a) + \hat{\ell}(a)$. The nominal completion time is $C_\pi(a) = \sigma_\pi(a) + \ell(a)$. The proposed target dispatch time of the primary action $q$ is:
\begin{equation}
t_{\mathrm{dispatch}}(\pi) \ge \max_{a \in V_\pi} \hat{C}_\pi(a) + \delta_{\mathrm{ready}}.
\label{eq:dispatch-time}
\end{equation}

\subsection{Prospective Feasibility vs. Realized CAC Discharge}
\label{sec:prospective-vs-realized}

\textbf{Pre-execution plan feasibility does not establish actual obligation discharge}:
\begin{enumerate}
    \item \textbf{Prospective Plan Feasibility ($\pi \models_{\mathrm{pros}} \mathcal{P}$):} Evaluated at compile time by \aas. A plan is prospectively feasible if, under modeled capabilities $\operatorname{Cov}(a)$, latency bounds $\hat{\ell}(a)$, and exposure sets $\mathcal{X}(a)$, its schedule satisfies precedence, budgets, deadline windows, quorum cardinalities, and epistemic diversity cuts. Actual execution can still exceed the latency bound or fail to emit affirmative evidence.
    \item \textbf{Realized Obligation Discharge ($\mathcal{W}_q \models_{\mathrm{CAC}} \Omega$):} Evaluated at runtime by \cac. When operations execute, they emit concrete evidence receipts $E_{\mathrm{store}} = \bigcup_{a \in V_\pi} \operatorname{Emit}(a)$. \cac evaluates the formal predicate $\operatorname{Discharge}(\omega, E_{\mathrm{store}}, s, \chi)$ for each $\omega \in \Omega$. An operation $a$ with $(\omega, \varepsilon) \in \operatorname{Cov}(a)$ may fail, timeout, or return evidence refuting the predicate ($\textsc{Violated}$). Only when all obligations return $\textsc{Satisfied}(W_\omega)$ does \cac mint certificate $\mathcal{C}_q$, granting dispatch eligibility.
\end{enumerate}

\begin{definition}[Prospectively Admissible Assurance Plan]
\label{def:prospective-admissibility}
A plan $\pi = \langle V_\pi, E_\pi, \sigma_\pi, \mathcal{M}_\pi \rangle$ is prospectively admissible for $\mathcal{P}$, written $\pi \models_{\mathrm{pros}} \mathcal{P}$, if and only if it strictly satisfies all of the following hard invariants:
\begin{enumerate}
    \item \textbf{Precedence Feasibility:}
    \begin{equation}
    \forall (u, v) \in E_\pi: \quad \sigma_\pi(u) + \hat{\ell}(u) \le \sigma_\pi(v).
    \label{eq:precedence-constraint}
    \end{equation}
    \item \textbf{Deadline and Execution Window Feasibility:}
    \begin{equation}
    t_{\mathrm{dispatch}}(\pi) + \delta_{\mathrm{exec}} \le T_{\mathrm{dead}}.
    \label{eq:deadline-constraint}
    \end{equation}
    \item \textbf{Resource Budget Invariants:}
    \begin{align}
    \operatorname{Cost}_{\mathrm{total}}(\pi) &= \sum_{a \in V_\pi} c(a) \le B_{\$}, \label{eq:cost-budget} \\
    \operatorname{Latency}(\pi) &= t_{\mathrm{dispatch}}(\pi) - t_0 \le B_{\mathrm{lat}}, \label{eq:latency-budget} \\
    \operatorname{Tokens}(\pi) &= \sum_{a \in V_\pi} \operatorname{tok}(a) \le B_{\mathrm{tok}}, \label{eq:token-budget} \\
    \operatorname{Exposure}(\pi) &= \sum_{a \in V_\pi} \|\rho(a)\| \le B_{\mathrm{risk}}, \label{eq:risk-budget}
    \end{align}
    where $V_\pi$ includes each operation of any admitted diagnostic sub-plan exactly once. Shared prerequisites are likewise charged once.
    If a concurrency limit $K_{\max} < \infty$ is specified, then for all $t \in [t_0, T_{\mathrm{dead}}]$:
    \begin{equation}
    \left| \{a \in V_\pi \mid \sigma_\pi(a) \le t < \sigma_\pi(a) + \hat{\ell}(a)\} \right| \le K_{\max}.
    \label{eq:concurrency-constraint}
    \end{equation}
    \item \textbf{Prospective Witness Quorum Coverage:}
    For each obligation $\omega \in \Omega$, the assigned operation set $\mathcal{M}_\pi(\omega) \subseteq V_\pi$ satisfies:
    \begin{equation}
    \forall a \in \mathcal{M}_\pi(\omega): \; \exists \varepsilon \in \mathcal{E}_\omega \text{ s.t. } (\omega, \varepsilon) \in \operatorname{Cov}(a),
    \end{equation}
    and the candidate witness cardinality meets the quorum rule:
    \begin{equation}
    |\mathcal{M}_\pi(\omega)| \ge k_\omega.
    \label{eq:quorum-constraint}
    \end{equation}
    \item \textbf{Prospective Dispatch Freshness Overlap:}
    For every obligation $\omega \in \Omega$ and assigned operation $a \in \mathcal{M}_\pi(\omega)$, let the effective lifespan be $\Delta t_{\mathrm{eff}}(a, \omega) = \min(\Delta t_\omega, \Delta t(a))$. The target dispatch time $t_{\mathrm{dispatch}}(\pi)$ and execution window $\delta_{\mathrm{exec}}$ must be covered by the prospective validity window:
    \begin{equation}
    t_{\mathrm{dispatch}}(\pi) + \delta_{\mathrm{exec}} \le \sigma_\pi(a) + \Delta t_{\mathrm{eff}}(a, \omega) - \epsilon_{\mathrm{skew}}.
    \label{eq:freshness-invariant}
    \end{equation}
    \item \textbf{Epistemic Fault Domain Diversity Cut:}
    For every obligation $\omega \in \Omega$, candidate witnesses $\mathcal{M}_\pi(\omega)$ satisfy:
    \begin{equation}
    \kappa_E(\mathcal{M}_\pi(\omega),\Gamma_\omega) \ge \kappa_E^{\min}(\omega),
    \label{eq:efd-invariant}
    \end{equation}
    where $\Gamma_\omega$ is the positive $k_\omega$-of-$|W|$ approval rule and $\kappa_E(W,\Gamma_\omega)$ is the decision-rule-relative cut in Definition~\ref{def:efd-cut}.
    \item \textbf{Diagnostic Admission Invariance:}
    For every consequential diagnostic operation $a \in V_\pi$ (where $\rho(a) \succeq \tau_{\Pi}$), $V_\pi$ contains an admitted prerequisite sub-plan $\pi_a \prec a$ prospectively satisfying $\Omega_{\Pi}(a, s)$.
\end{enumerate}
\end{definition}

Let $\Pi_{\mathrm{adm}}(\mathcal{P}) = \{\pi \mid \pi \models_{\mathrm{pros}} \mathcal{P}\}$ denote the set of all prospectively admissible plans. Under our canonical optimization contract (\textbf{Contract C: Constrained Cost Minimization}), the primary objective of \aas is to synthesize an admissible plan $\pi^*$ that minimizes total financial expenditure subject to hard bounds on latency ($B_{\mathrm{lat}}$), cognitive reasoning tokens ($B_{\mathrm{tok}}$), operational risk perturbation ($B_{\mathrm{risk}}$), and CAC quorums and EFD cuts:
\begin{align}
\pi^* &= \arg\min_{\pi \in \Pi_{\mathrm{adm}}(\mathcal{P})} \operatorname{Cost}_{\mathrm{total}}(\pi) \nonumber \\
&= \arg\min_{\pi \in \Pi_{\mathrm{adm}}(\mathcal{P})} \sum_{a \in V_\pi} c(a).
\label{eq:optimization-objective}
\end{align}
When multiple candidate plans achieve identical minimum financial expenditure, \aas resolves ties lexicographically by minimizing dispatch latency $\operatorname{Latency}(\pi) = t_{\mathrm{dispatch}}(\pi) - t_0$ and cumulative operational risk $\operatorname{Exposure}(\pi) = \sum_{a \in V_\pi} \|\rho(a)\|$.

If $\Pi_{\mathrm{adm}}(\mathcal{P}) = \emptyset$, \aas returns $\textsc{Infeasible}$, triggering deterministic safe refusal at the admission boundary.

%% file: sections/04-temporal-freshness.tex
\section{Temporal Freshness and Validity Intersections}
\label{sec:temporal-freshness}

A distinguishing characteristic separating assurance scheduling from standard DAG scheduling (e.g., job-shop or workflow scheduling) is that \textbf{outputs expire}. In traditional systems, once an intermediate artifact is computed, it remains valid indefinitely unless explicitly evicted. In epistemic systems, an evidence receipt $e$ reflects a dynamic distributed state whose fidelity decays over time.

\subsection{Validity Intervals, Latency, and Readiness Lifecycle}
\label{sec:validity-intervals}

When an assurance operation $a \in \mathcal{A}$ executes, it inspects external system state at an \emph{observation timestamp} $t_{\mathrm{obs}}(e)$ and completes processing at completion timestamp $C_\pi(a) \ge t_{\mathrm{obs}}(e)$. Distinguishing $t_{\mathrm{obs}}(e)$ from $C_\pi(a)$ is vital: a distributed telemetry probe, log aggregator, or formal model checker may sample state at start, yet finish computation seconds or minutes later.

To establish sound scheduling, we track the complete evidence and dispatch readiness lifecycle:
\begin{enumerate}
    \item \textbf{Observation Time $t_{\mathrm{obs}}(e)$:} The physical moment at which system state is sampled. For an operation scheduled at start time $\sigma_\pi(a)$ with conservative upper-bound latency $\hat{\ell}(a)$ and actual latency $\ell_{\mathrm{act}}(a) \le \hat{\ell}(a)$, observation occurs at $\sigma_\pi(a) \le t_{\mathrm{obs}}(e) \le C_\pi(a)$. In the conservative worst case for expiration, state is assumed observed at invocation: $t_{\mathrm{obs}}(e) \ge \sigma_\pi(a)$.
    \item \textbf{Operation Completion $C_\pi(a)$:} The timestamp when computation finishes: $C_\pi(a) = \sigma_\pi(a) + \ell_{\mathrm{act}}(a) \le \sigma_\pi(a) + \hat{\ell}(a) = \hat{C}_\pi(a)$, where $\hat{C}_\pi(a)$ is conservative predicted completion. Post-observation computation duration is $\ell_{\mathrm{post}}(a) = C_\pi(a) - t_{\mathrm{obs}}(e) \ge 0$.
    \item \textbf{Receipt Availability \& Verification $t_{\mathrm{avail}}(e)$:} Emitting, cryptographically signing, transmitting, and verifying the receipt requires non-zero duration $\delta_{\mathrm{verify}} \ge 0$. The receipt is available in the receipt store only at $t_{\mathrm{avail}}(e) \ge C_\pi(a) + \delta_{\mathrm{verify}}$.
    \item \textbf{Certificate Minting $t_{\mathrm{mint}}$:} \cac evaluates the composite witness manifest $\mathcal{W}_q$ and mints admission certificate $\mathcal{C}_q$. Minting requires that all required witness receipts have completed verification:
    \begin{equation}
    t_{\mathrm{mint}} \ge \max_{a \in V_\pi} C_\pi(a) + \delta_{\mathrm{verify}} + \delta_{\mathrm{mint}}.
    \end{equation}
    \item \textbf{Gateway Dispatch $t_{\mathrm{dispatch}}$:} Proposal $q$ and certificate $\mathcal{C}_q$ arrive at the execution gateway and complete mediation after network transit delay $\delta_{\mathrm{disp}}$:
    \begin{equation}
    t_{\mathrm{dispatch}} \ge t_{\mathrm{mint}} + \delta_{\mathrm{disp}} \ge \max_{a \in V_\pi} C_\pi(a) + \delta_{\mathrm{ready}},
    \label{eq:earliest-legal-dispatch}
    \end{equation}
    where $\delta_{\mathrm{ready}}$ is the post-completion lead time defined in Section~\ref{sec:formal-inputs}.
    \item \textbf{Protected Execution Window:} The action executes during $[t_{\mathrm{dispatch}}, \; t_{\mathrm{dispatch}} + \delta_{\mathrm{exec}}]$.
\end{enumerate}

\paragraph{Clock Uncertainty Model.}
Each observer node records $t_{\mathrm{obs}}(e)$ on its local clock $C_i$, while gateway dispatch is evaluated against the gateway clock $C_{\mathrm{gw}}$. We model clock synchronization uncertainty under bounded skew: $|C_i(t) - C_{\mathrm{gw}}(t)| \le \epsilon_{\mathrm{skew}}$ for all physical times $t$.
Hence, when an observer records $t_{\mathrm{obs}}(e)$, the physical observation time expressed in the gateway's reference frame satisfies $t_{\mathrm{sample}}^{\mathrm{gw}} \in [t_{\mathrm{obs}}(e) - \epsilon_{\mathrm{skew}}, \; t_{\mathrm{obs}}(e) + \epsilon_{\mathrm{skew}}]$. In the worst case, physical observation occurred at the lower bound $t_{\mathrm{obs}}(e) - \epsilon_{\mathrm{skew}}$.
When two independent observer nodes $i$ and $j$ compare raw timestamps directly without gateway mediation, the relative clock skew is at most $|C_i(t) - C_j(t)| \le 2\epsilon_{\mathrm{skew}}$ (reconciling Section~\ref{sec:discussion-limitations}).

Each receipt $e \in \operatorname{Emit}(a)$ carries an intrinsic validity lifespan $\Delta t(e)$. For obligation $\omega \in \Omega$ with policy freshness bound $\Delta t_\omega$, the effective lifespan is $\Delta t_{\mathrm{eff}}(e, \omega) = \min(\Delta t_\omega, \Delta t(e))$. If $e$ supports multiple obligations $\Omega_e \subseteq \Omega$, its effective lifespan across all supported obligations is:
\begin{equation}
\Delta t_{\mathrm{eff}}(e) = \min_{\omega \in \Omega_e} \Delta t_{\mathrm{eff}}(e, \omega).
\label{eq:multi-obligation-lifespan}
\end{equation}
Evaluated at the gateway, receipt $e$ expires at $t_{\mathrm{exp}}^{\mathrm{gw}}(e) = t_{\mathrm{obs}}(e) + \Delta t_{\mathrm{eff}}(e) - \epsilon_{\mathrm{skew}}$, yielding the conservative validity interval:
\begin{equation}
I(e) = \left[ t_{\mathrm{obs}}(e) - \epsilon_{\mathrm{skew}}, \; t_{\mathrm{obs}}(e) + \Delta t_{\mathrm{eff}}(e) - \epsilon_{\mathrm{skew}} \right].
\label{eq:validity-interval}
\end{equation}

\paragraph{Freshness, Admissibility, and Physical Truth.}
Sound systems engineering requires rigorously distinguishing three distinct properties:
\begin{enumerate}
    \item \textbf{Receipt Freshness:} Receipt $e$ is fresh at gateway time $t$ if $t \in I(e)$, i.e., $t \le t_{\mathrm{obs}}(e) + \Delta t_{\mathrm{eff}}(e) - \epsilon_{\mathrm{skew}}$.
    \item \textbf{Certificate Admissibility:} Certificate $\mathcal{C}_q$ is admissible at dispatch time $t_{\mathrm{dispatch}}$ if all state version guards hold ($G_q = \{\nu_s = \nu_{\mathrm{live}}\}$), all obligation quorums and \efd cuts are satisfied, and all constituent witnesses remain continuously valid throughout the protected execution window:
    \begin{equation}
    [t_{\mathrm{dispatch}}, \; t_{\mathrm{dispatch}} + \delta_{\mathrm{exec}}] \subseteq \bigcap_{e \in W_{\mathrm{all}}} I(e),
    \label{eq:execution-window-validity}
    \end{equation}
    where $W_{\mathrm{all}} = \bigcup_{(\omega, W_\omega) \in \mathcal{W}_q} W_\omega$.
    \item \textbf{Physical Predicate Truth:} The underlying physical safeguard $\phi_\omega(s(t))$ holds continuously in the environment: $\forall t \in [t_{\mathrm{dispatch}}, t_{\mathrm{dispatch}} + \delta_{\mathrm{exec}}], \phi_\omega(s(t)) = \mathbf{true}$.
\end{enumerate}
\begin{remark}[Temporal Freshness vs. Physical Truth]
Temporal freshness and certificate admissibility (Properties 1 and 2) are necessary epistemic proxies, but \textbf{do not guarantee physical predicate truth (Property 3)}. If the environment experiences unmodeled out-of-band state mutations (e.g., a physical power cut or external manual override), or if policy freshness $\Delta t_\omega$ is set more permissively than the physical state drift rate $\tau_{\mathrm{drift}}(\phi_\omega)$, a temporally fresh certificate will attest to a condition that has already ceased to hold physically. \aas and \cac guarantee rigorous epistemic admissibility under the modeled policy; ensuring physical truth requires that policy authors configure $\Delta t_\omega \le \tau_{\mathrm{drift}}(\phi_\omega)$.
\end{remark}

\begin{theorem}[Completion-Aware Dispatch Freshness Intersection]
\label{thm:freshness-intersection}
Fix an executed assurance plan $\pi$ with completed operations $V_\pi$, actual completions $\{C_\pi(a)\}$, observation timestamps $\{t_{\mathrm{obs}}(e)\}$, and composite witness manifest $W_{\mathrm{all}}$. Let $V_{\mathrm{req}} \subseteq V_\pi$ denote the set of all operations whose completion is mandatory prior to proposal dispatch, including both witness-producing operations $V_{\mathrm{manifest}} = \{a \in V_\pi \mid \exists e \in W_{\mathrm{all}}, e \in \operatorname{Emit}(a)\}$ and prerequisite operational dependencies $\operatorname{Anc}_{D_{\mathrm{op}}}(V_{\mathrm{manifest}})$.

For a fixed completed plan and its receipts, a dispatch timestamp satisfying readiness and freshness throughout the protected execution window exists if and only if:
\begin{equation}
\begin{aligned}
\max_{a \in V_{\mathrm{req}}} C_\pi(a) &+ \delta_{\mathrm{ready}} + \delta_{\mathrm{exec}} + \epsilon_{\mathrm{skew}} \\
&\le \min_{e \in W_{\mathrm{all}}} \left( t_{\mathrm{obs}}(e) + \Delta t_{\mathrm{eff}}(e) \right).
\end{aligned}
\label{eq:freshness-intersection-condition}
\end{equation}
This condition concerns the freshness interval. Full admission additionally requires $t_{\mathrm{dispatch}}+\delta_{\mathrm{exec}}\le T_{\mathrm{dead}}$ and the latency budget; equivalently, their upper bounds must also exceed the readiness lower bound.
Equivalently, the condition holds if and only if for every required operation $a \in V_{\mathrm{req}}$ and every witness receipt $e_i \in W_{\mathrm{all}}$:
\begin{equation}
C_\pi(a) - t_{\mathrm{obs}}(e_i) \le \Delta t_{\mathrm{eff}}(e_i) - \delta_{\mathrm{ready}} - \delta_{\mathrm{exec}} - \epsilon_{\mathrm{skew}}.
\label{eq:pairwise-freshness-skew}
\end{equation}
Furthermore, for any witness-producing operation $a \in V_{\mathrm{manifest}}$ emitting receipt $e_j \in W_{\mathrm{all}}$ with observation time $t_{\mathrm{obs}}(e_j)$ and post-observation latency $\ell_{\mathrm{post}}(a) = C_\pi(a) - t_{\mathrm{obs}}(e_j) \ge 0$:
\begin{equation}
\begin{aligned}
t_{\mathrm{obs}}(e_j) - t_{\mathrm{obs}}(e_i) \le{} &\Delta t_{\mathrm{eff}}(e_i) - \ell_{\mathrm{post}}(a) \\
&- \delta_{\mathrm{ready}} - \delta_{\mathrm{exec}} - \epsilon_{\mathrm{skew}}.
\end{aligned}
\label{eq:pairwise-obs-skew}
\end{equation}
Non-witness prerequisite operations $u \in V_{\mathrm{req}} \setminus V_{\mathrm{manifest}}$ emit no admission receipts and therefore have no $t_{\mathrm{obs}}$ terms, but their completions $C_\pi(u)$ participate in the readiness bound~\eqref{eq:pairwise-freshness-skew} against all $e_i \in W_{\mathrm{all}}$.
\end{theorem}
\begin{proof}
By Equation~\eqref{eq:earliest-legal-dispatch}, dispatch cannot occur until all mandatory operations in $V_{\mathrm{req}}$ have completed and candidate receipts have been verified and minted: $t_{\mathrm{dispatch}} \ge \max_{a \in V_{\mathrm{req}}} C_\pi(a) + \delta_{\mathrm{ready}}$.
Simultaneously, condition~\eqref{eq:execution-window-validity} requires that for every witness receipt $e \in W_{\mathrm{all}}$, the execution window end does not exceed receipt expiration on the gateway clock:
$t_{\mathrm{dispatch}} + \delta_{\mathrm{exec}} \le \min_{e \in W_{\mathrm{all}}} (t_{\mathrm{obs}}(e) + \Delta t_{\mathrm{eff}}(e) - \epsilon_{\mathrm{skew}})$.
Thus, admissible dispatch timestamps form the closed interval $[t_{\mathrm{lower}}, t_{\mathrm{upper}}]$, where:
\begin{align*}
t_{\mathrm{lower}} &= \max_{a \in V_{\mathrm{req}}} C_\pi(a) + \delta_{\mathrm{ready}}, \\
t_{\mathrm{upper}} &= \min_{e \in W_{\mathrm{all}}} \left( t_{\mathrm{obs}}(e) + \Delta t_{\mathrm{eff}}(e) - \epsilon_{\mathrm{skew}} \right) - \delta_{\mathrm{exec}}.
\end{align*}
This interval is non-empty if and only if $t_{\mathrm{lower}} \le t_{\mathrm{upper}}$, which rearranges directly to Equation~\eqref{eq:freshness-intersection-condition}.
Condition~\eqref{eq:freshness-intersection-condition} holds if and only if each term in the maximum is bounded by each term in the minimum, directly establishing pairwise equivalence~\eqref{eq:pairwise-freshness-skew} for all $a \in V_{\mathrm{req}}$ and $e_i \in W_{\mathrm{all}}$.
For witness-producing operations $a \in V_{\mathrm{manifest}}$, substituting $C_\pi(a) = t_{\mathrm{obs}}(e_j) + \ell_{\mathrm{post}}(a)$ yields Equation~\eqref{eq:pairwise-obs-skew}.
\end{proof}

\begin{example}[Counterexample: Failure of Observation-Only Dispatch Bounds]
Consider an operation $a_{\mathrm{verify}}$ that samples database state at $t_{\mathrm{obs}}(e) = 0$, but executes complex consistency checking requiring latency $\hat{\ell}(a) = 10$\,s, completing at $C(a) = 10$\,s. Let effective evidence lifespan be $\Delta t_{\mathrm{eff}}(e) = 8$\,s, with lead times $\delta_{\mathrm{ready}} = 0.2$\,s, execution window $\delta_{\mathrm{exec}} = 1.0$\,s, and $\epsilon_{\mathrm{skew}} = 0$.
If dispatch feasibility were naively bounded by observation time alone ($\max t_{\mathrm{obs}} + \delta_{\mathrm{ready}} + \delta_{\mathrm{exec}} \le \min(t_{\mathrm{obs}} + \Delta t)$), the inequality would evaluate to $0 + 0.2 + 1.0 = 1.2 \le 8.0$. An observation-only scheduler would falsely declare this schedule feasible, planning dispatch at $t = 1.2$\,s.
In physical reality, receipt $e$ does not exist at $t = 1.2$\,s because $a_{\mathrm{verify}}$ is still running. Earliest legal dispatch cannot occur before $t_{\mathrm{ready}} = 10 + 0.2 = 10.2$\,s. At $t = 10.2$\,s, the evidence collected at $t = 0$ has already expired ($10.2 + 1.0 = 11.2 > 8.0$). Theorem~\ref{thm:freshness-intersection} correctly flags this as infeasible because $\max C(a) + \delta_{\mathrm{ready}} + \delta_{\mathrm{exec}} = 11.2 > 8.0$.
\end{example}

\subsection{Why Forward Scheduling Fails}
\label{sec:forward-scheduling-fails}

Traditional workflow engines apply \emph{earliest-deadline-first} (EDF) or \emph{as-soon-as-possible} (ASAP) forward scheduling: operations dispatch immediately once precedence dependencies are met.

Consider three operations from Section~\ref{sec:motivating-failures}:
\begin{itemize}
    \item $a_{\mathrm{topo}}$: latency 1.2\,s, lifespan $\Delta t = 30$\,s,
    \item $a_{\mathrm{lag}}$: latency 1.2\,s, lifespan $\Delta t = 5$\,s,
    \item $a_{\mathrm{fence}}$: latency 6.0\,s, lifespan $\Delta t = 10$\,s.
\end{itemize}
Under forward scheduling starting at $t=0$, $a_{\mathrm{topo}}$ finishes at $t=1.2$\,s (valid on $[1.2, 31.2]$\,s). Next, $a_{\mathrm{lag}}$ completes at $t=2.4$\,s, producing receipt $e_{\mathrm{lag}}$ whose validity window is $[2.4, 7.4]$\,s. Finally, long-latency fencing $a_{\mathrm{fence}}$ executes from $t=2.4$\,s to $t=8.4$\,s.
Here:
\begin{equation}
t_{\mathrm{lower}} = 8.4\,\text{s} + \delta_{\mathrm{ready}} > 7.4\,\text{s} = t_{\mathrm{upper}},
\end{equation}
yielding an empty validity intersection! ASAP forward execution guarantees dispatch failure whenever short-lived telemetry is gathered prior to long-latency operations.

\subsection{Backward Just-in-Time Scheduling}
\label{sec:backward-scheduling}

To satisfy Theorem~\ref{thm:freshness-intersection} without delaying action execution unnecessarily, \aas implements \textbf{Backward Just-in-Time (JIT) Scheduling}.

\paragraph{Analysis of Target Dispatch Feasibility: Continuous Monotonicity vs. Heuristic Artifacts.}
A critical theoretical question is whether the set of feasible target dispatch timestamps $\mathcal{T}_{\mathrm{feas}} \subseteq [t_{\min}, t_{\max}]$ is continuous and monotonic, or whether it can fragment into disconnected sub-intervals.
\begin{proposition}[Time-Translation Invariance of Continuous ASP Feasibility]
\label{prop:time-translation}
For any static ASP instance in continuous time with release epoch $t_0$, deadline $T_{\mathrm{dead}}$, execution window $\delta_{\mathrm{exec}}$, and latency budget $B_{\mathrm{lat}}$, define the maximum legal dispatch horizon:
\begin{equation}
T_{\max} = \min(T_{\mathrm{dead}} - \delta_{\mathrm{exec}}, \; t_0 + B_{\mathrm{lat}}).
\label{eq:max-dispatch-horizon}
\end{equation}
If $(\sigma_\pi, t_{\mathrm{dispatch}})$ is a prospectively feasible schedule, then for all forward shifts $\Delta \in [0, \; T_{\max} - t_{\mathrm{dispatch}}]$, the uniform time-translated schedule defined by $\sigma'_\pi(u) = \sigma_\pi(u) + \Delta$ and $t'_{\mathrm{dispatch}} = t_{\mathrm{dispatch}} + \Delta$ is also prospectively feasible. Consequently, $\mathcal{T}_{\mathrm{feas}}$ is either empty or a single connected interval $[t^*_{\min}, \; T_{\max}]$; it is never a disconnected union of sub-intervals.
\end{proposition}
\begin{proof}
Let $(\sigma_\pi, t_{\mathrm{dispatch}})$ be feasible. Under translation by $\Delta \ge 0$: (i) release constraints hold since $\sigma'_\pi(u) = \sigma_\pi(u) + \Delta \ge t_0 + \Delta \ge t_0$; (ii) precedence constraints $\sigma'_\pi(u) + \hat{\ell}(u) \le \sigma'_\pi(v)$ hold identically; (iii) dispatch readiness $t'_{\mathrm{dispatch}} \ge \max_u (\sigma'_\pi(u) + \hat{\ell}(u)) + \delta_{\mathrm{ready}}$ is preserved since $\Delta$ appears on both sides; (iv) receipt freshness at dispatch satisfies $t'_{\mathrm{dispatch}} + \delta_{\mathrm{exec}} - (\sigma'_\pi(u) + \ell(u)) = t_{\mathrm{dispatch}} + \delta_{\mathrm{exec}} - (\sigma_\pi(u) + \ell(u)) \le \Delta t_{\mathrm{eff}}(u) - \epsilon_{\mathrm{skew}}$, since $\Delta$ cancels identically; (v) instantaneous concurrency $|\{u \mid \sigma'_\pi(u) \le t < \sigma'_\pi(u) + \hat{\ell}(u)\}|$ at time $t$ equals concurrency at $t - \Delta$, preserving $K_{\max}$; (vi) the deadline condition $t'_{\mathrm{dispatch}} + \delta_{\mathrm{exec}} \le T_{\mathrm{dead}}$ holds because $t'_{\mathrm{dispatch}} \le T_{\max} \le T_{\mathrm{dead}} - \delta_{\mathrm{exec}}$; and (vii) the latency budget condition $\operatorname{Latency}(\pi') = t'_{\mathrm{dispatch}} - t_0 = (t_{\mathrm{dispatch}} - t_0) + \Delta \le B_{\mathrm{lat}}$ holds because $t'_{\mathrm{dispatch}} \le T_{\max} \le t_0 + B_{\mathrm{lat}}$.
\end{proof}

\noindent
Why, then, do practical temporal schedulers encounter non-monotonic behavior across candidate targets? \textbf{Non-monotonicity is strictly an artifact of heuristic placement algorithms}, rather than a property of the underlying continuous scheduling problem:
\begin{enumerate}
    \item \textbf{Epoch Clamping:} Backward propagation anchors completions to $t_{\mathrm{target}} - \delta_{\mathrm{ready}}$ and pushes starts backwards. When $t_{\mathrm{target}}$ is small, unconstrained starts fall before $t_0$. Heuristically clamping starts to $\sigma_\pi(u) \ge t_0$ alters inter-operation spacing; as $t_{\mathrm{target}}$ is increased without uniformly translating the earliest operations, the elapsed duration $t_{\mathrm{target}} - \sigma_\pi(u)$ grows, causing clamped receipts to expire before dispatch.
    \item \textbf{Concurrency Collisions ($K_{\max} < \infty$):} Under finite concurrency, greedy placement heuristics serialize overlapping operations using local priority rules. Shifting $t_{\mathrm{target}}$ changes which operations collide, producing discrete reorganizations of the schedule where a greedy heuristic may succeed at $t_1$, fail at $t_2 > t_1$, and succeed again at $t_3 > t_2$.
\end{enumerate}
Accordingly, universal binary search over $[t_{\min}, t_{\max}]$ is unsound for backward heuristics. \aas evaluates candidates across a discrete grid with step $\tau_{\mathrm{grid}}$. If grid search exhausts without finding a placement, it returns $\textsc{NotFound}$; provable $\textsc{Infeasible}$ is reserved strictly for certified structural violations.

\paragraph{Backward Propagation with Safety Margins.}
Given candidate target dispatch $t_{\mathrm{target}}$, operations are scheduled in reverse topological order. For each operation $u \in V_\pi$, the latest permissible completion time $T_{\mathrm{latest}}(u)$ is bounded by successors in $E_\pi$ and target dispatch:
\begin{align}
T_{\mathrm{latest}}(u) &\le \min_{(u, v) \in E_\pi} \sigma_\pi(v), \label{eq:succ-bound} \\
T_{\mathrm{latest}}(u) &\le t_{\mathrm{target}} - \delta_{\mathrm{ready}}, \label{eq:dispatch-bound} \\
\sigma_\pi(u) &= T_{\mathrm{latest}}(u) - \hat{\ell}(u) - \epsilon_{\mathrm{guard}}, \label{eq:jit-start}
\end{align}
where $\hat{\ell}(u)$ is the conservative upper-bound latency and $\epsilon_{\mathrm{guard}}$ is an operational safety buffer absorbing network round-trip jitter.

\begin{algorithm}[t]
\caption{Two-Tier Backward Just-in-Time Temporal Scheduler}
\label{alg:backward-scheduler}
\begin{algorithmic}[1]
\footnotesize
\linespread{0.92}\selectfont
\Require Plan DAG $(V_\pi, E_\pi)$, manifest assignments $\mathcal{M}_\pi$, conservative latencies $\hat{\ell}$, lifespans $\Delta t_{\mathrm{eff}}$, epoch $t_0$, deadline $T_{\mathrm{dead}}$, latency budget $B_{\mathrm{lat}}$, concurrency $K_{\max}$, guards $\delta_{\mathrm{exec}}, \delta_{\mathrm{ready}}, \epsilon_{\mathrm{guard}}, \epsilon_{\mathrm{skew}}$, grid step $\tau_{\mathrm{grid}}$
\Ensure Feasible schedule $(\sigma_\pi, t_{\mathrm{dispatch}})$, $\textsc{Infeasible}$ (certified structural failure), or $\textsc{NotFound}$ (heuristic exhaustion)
\State $t_{\mathrm{CP}} \gets \operatorname{ComputeCriticalPath}(V_\pi, E_\pi, \hat{\ell} + \epsilon_{\mathrm{guard}})$
\State $t_{\min} \gets t_0 + t_{\mathrm{CP}} + \delta_{\mathrm{ready}}$; $t_{\max} \gets \min(T_{\mathrm{dead}} - \delta_{\mathrm{exec}}, \; t_0 + B_{\mathrm{lat}})$
\If{$t_{\min} > t_{\max}$}
  \State \Return $\textsc{Infeasible}$ \Comment{Structural critical path exceeds deadline or latency budget}
\EndIf
\State $V_{\mathrm{manifest}} \gets \{a \in V_\pi \mid \exists \omega, a \in \mathcal{M}_\pi(\omega)\}$
\Comment{\textbf{Tier 1: Fast Parallel Backward-JIT Search}}
\For{$t_{\mathrm{target}} \gets t_{\min}$ \textbf{to} $t_{\max}$ \textbf{step} $\tau_{\mathrm{grid}}$}
  \State $\mathrm{feasible} \gets \mathbf{true}$
  \State Initialize $T_{\mathrm{latest}}(u) \gets t_{\mathrm{target}} - \delta_{\mathrm{ready}}$ for all $u \in V_\pi$
  \For{each $u \in V_\pi$ in reverse topological order}
    \State $T_{\mathrm{latest}}(u) \gets \min \Big( T_{\mathrm{latest}}(u), \; \min_{(u,v) \in E_\pi} \sigma_\pi(v) \Big)$
    \State $\sigma_\pi(u) \gets T_{\mathrm{latest}}(u) - \hat{\ell}(u) - \epsilon_{\mathrm{guard}}$
    \If{$\sigma_\pi(u) < t_0$}
      \State $\mathrm{feasible} \gets \mathbf{false}$; \textbf{break} \Comment{Insufficient lead time from epoch}
    \EndIf
    \If{$u \in V_{\mathrm{manifest}}$}
      \State $\Delta t_{\min}(u) \gets \min_{\omega: u \in \mathcal{M}_\pi(\omega)} \min(\Delta t_\omega, \Delta t(u))$
      \If{$t_{\mathrm{target}} + \delta_{\mathrm{exec}} > \sigma_\pi(u) + \Delta t_{\min}(u) - \epsilon_{\mathrm{skew}}$}
        \State $\mathrm{feasible} \gets \mathbf{false}$; \textbf{break} \Comment{Witness receipt expires before window ends}
      \EndIf
    \EndIf
  \EndFor
  \If{$\mathrm{feasible}$}
    \If{$K_{\max} = \infty \lor \max_{t} |\{u \in V_\pi \mid \sigma_\pi(u) \le t < \sigma_\pi(u) + \hat{\ell}(u)\}| \le K_{\max}$}
      \State \Return $(\sigma_\pi, t_{\mathrm{target}})$ \Comment{Fast parallel JIT placement succeeds}
    \EndIf
  \EndIf
\EndFor
\Comment{\textbf{Tier 2: Authoritative Slack-Aware Temporal Solver}}
\If{$K_{\max} < \infty$}
  \State $\mathrm{res} \gets \textsc{AuthTemporal}(\pi, K_{\max}, t_{\min}, t_{\max}, \tau_{\mathrm{grid}})$
  \If{$\mathrm{res} \neq \textsc{NotFound}$}
    \State \Return $\mathrm{res}$ \Comment{Slack-aware serialization succeeds}
  \EndIf
\EndIf
\State \Return $\textsc{NotFound}$ \Comment{Grid search exhausted without valid placement}
\end{algorithmic}
\end{algorithm}

\paragraph{Two-Tier Architecture, Completeness, and Limits.}
Tier~1 evaluates candidate target dispatch timestamps over $K_{\mathrm{grid}} = \lceil (t_{\max} - t_{\min})/\tau_{\mathrm{grid}} \rceil$ discrete breakpoints in reverse topological order, requiring $O(K_{\mathrm{grid}} \cdot (|V_\pi| + |E_\pi|))$ time.
When concurrency is constrained ($K_{\max} < \infty$) and parallel placement collides, Tier~1 fails conservatively. To prevent heuristic parallel failure from being mistaken for problem infeasibility, Tier~2 activates the Authoritative Temporal Solver ($\operatorname{SolveAuthoritativeTemporal}$). For each candidate $t_{\mathrm{target}}$, Tier~2 establishes the admissible release-due interval $[\underline{s}(u), \bar{s}(u)]$ for each operation $u$:
\begin{align}
\underline{s}(u) &= \max\big(t_0, \; t_{\mathrm{target}} + \delta_{\mathrm{exec}} + \epsilon_{\mathrm{skew}} - \Delta t_{\min}(u)\big), \label{eq:earliest-start} \\
\bar{s}(u) &= t_{\mathrm{target}} - \delta_{\mathrm{ready}} - \hat{\ell}(u) - \epsilon_{\mathrm{guard}}. \label{eq:latest-start}
\end{align}
Tier~2 then executes an event-driven forward simulation under $K_{\max}$, prioritizing ready operations by minimum slack ($\bar{s}(u)$). If slack-aware serialization finds a valid placement, it returns the sound schedule $(\sigma_\pi, t_{\mathrm{target}})$.
We formally distinguish the levels of completeness and their cut semantics:
\begin{enumerate}
    \item \textbf{Soundness:} Any schedule returned by Tier~1 or Tier~2 strictly satisfies all precedence, deadline, concurrency ($K_{\max}$), and freshness window constraints.
    \item \textbf{Certified Critical-Path Infeasibility:} If the directed critical path of prerequisites exceeds available lead time ($t_{\mathrm{CP}} + \delta_{\mathrm{ready}} > T_{\mathrm{dead}} - t_0 - \delta_{\mathrm{exec}}$), the instance is provably infeasible in continuous time, generating a sound structural Benders cut~\eqref{eq:ilp-chain-cut}.
    \item \textbf{Finite-Grid Exhaustion vs. Certified Infeasibility:} When candidate evaluation over a finite grid $\tau_{\mathrm{grid}}$ terminates without finding a placement, Algorithm~\ref{alg:backward-scheduler} returns $\textsc{NotFound}$, \textbf{not} certified $\textsc{Infeasible}$. Under continuous time, feasibility is a single connected interval (Proposition~\ref{prop:time-translation}); failure of finite-grid heuristic placement cannot prove continuous infeasibility. Crucially, \textbf{no heuristic failure ($\textsc{NotFound}$) generates a globally valid Benders cut} in the Master ILP. Exact certificates of infeasibility under general concurrency bounds require exhaustive Simple Temporal Network permutation enumeration, as implemented by our Exact Oracle (Section~\ref{sec:eval-optimality}). Non-manifest operations are excluded from freshness checks, preventing artificial restriction of the dispatch window.
\end{enumerate}

\subsection{Robustness to Tail Latency and Staging Protocol}
\label{sec:latency-jitter}

Real distributed environments exhibit heavy-tailed execution latencies~\cite{dean2013tail}. Let actual latency be a random variable $\ell_{\mathrm{act}}(a) \sim \mathcal{D}_a$ with conservative quantile bound $\hat{\ell}(a) = \inf \{ \ell \mid \mathbb{P}(\ell_{\mathrm{act}}(a) \le \ell) \ge 1 - \alpha_{\mathrm{tail}} \}$.

If an unpredicted stall occurs in a long-latency operation $a_{\mathrm{slow}}$, pushing completion beyond $\hat{\ell}(a_{\mathrm{slow}})$, any short-lived receipts collected earlier may expire before dispatch. To decouple short-lived validity from long-latency variance, \aas employs a \textbf{Dual-Stage Staging Protocol}:
\begin{enumerate}
    \item \textbf{Stage 1 (Pre-fetch \& Long-Latency Execution):} Operations with long validity horizons ($\Delta t_{\mathrm{eff}}(a) \gg \hat{\ell}(a)$) and long latencies are dispatched early. Stochastic completion variance is absorbed while short-lived operations remain dormant.
    \item \textbf{Stage 2 (Synchronization Barrier \& JIT Burst):} As soon as all Stage~1 operations finish and their receipts are verified, the scheduler locks the exact completion timestamp $t_{\mathrm{sync}} = \max_{a \in V_{\mathrm{stage1}}} C_\pi(a) + \delta_{\mathrm{verify}}$. It recomputes target dispatch $t_{\mathrm{target}} = t_{\mathrm{sync}} + \max_{u \in V_{\mathrm{jit}}} (\hat{\ell}(u) + \epsilon_{\mathrm{guard}}) + \delta_{\mathrm{ready}}$ and fires all short-lived operations in a tightly synchronized parallel burst.
\end{enumerate}

\begin{theorem}[Temporal Freshness Invariant under Bounded Latency]
\label{thm:temporal-correctness}
Suppose execution latencies are bounded by $\ell_{\mathrm{act}}(a) \le \hat{\ell}(a)$ for all $a \in V_\pi$. If Algorithm~\ref{alg:backward-scheduler} returns a schedule $(\sigma_\pi, t_{\mathrm{dispatch}})$, then $(\sigma_\pi, t_{\mathrm{dispatch}})$ satisfies the common post-schedule validation conditions across both Tier~1 and Tier~2:
\begin{enumerate}
    \item Release validity: $\sigma_\pi(u) \ge t_0$ for all $u \in V_\pi$.
    \item Topological precedence: $\sigma_\pi(u) + \hat{\ell}(u) \le \sigma_\pi(v)$ for all $(u, v) \in E_\pi$.
    \item Lead-time readiness: $t_{\mathrm{dispatch}} \ge \sigma_\pi(u) + \hat{\ell}(u) + \delta_{\mathrm{ready}}$ for all $u \in V_\pi$.
    \item Witness freshness: $t_{\mathrm{dispatch}} + \delta_{\mathrm{exec}} \le \sigma_\pi(u) + \Delta t_{\min}(u) - \epsilon_{\mathrm{skew}}$ for all manifest operations $u \in V_{\mathrm{manifest}}$.
    \item Capacity bound: $|\{u \in V_\pi \mid \sigma_\pi(u) \le t < \sigma_\pi(u) + \hat{\ell}(u)\}| \le K_{\max}$ for all $t$.
    \item Deadline feasibility: $t_{\mathrm{dispatch}} + \delta_{\mathrm{exec}} \le T_{\mathrm{dead}}$.
\end{enumerate}
Consequently, executing $\pi$ strictly preserves completion-aware dispatch readiness and witness freshness during target execution:
\begin{align*}
t_{\mathrm{dispatch}} &\ge \max_{a \in V_\pi} C_\pi(a) + \delta_{\mathrm{ready}}, \quad \text{and} \\
[t_{\mathrm{dispatch}}, \; t_{\mathrm{dispatch}} + \delta_{\mathrm{exec}}] &\subseteq \bigcap_{e \in W_{\mathrm{all}}} I(e).
\end{align*}
\end{theorem}
\begin{proof}
Let $(\sigma_\pi, t_{\mathrm{dispatch}})$ be any schedule returned by Algorithm~\ref{alg:backward-scheduler}. Both Tier~1 and Tier~2 explicitly gate schedule emission upon satisfying validation conditions (1)--(6).
For every operation $u \in V_\pi$, actual completion satisfies $C_\pi(u) = \sigma_\pi(u) + \ell_{\mathrm{act}}(u) \le \sigma_\pi(u) + \hat{\ell}(u)$ because execution latency is bounded by $\hat{\ell}(u)$. By validation condition (3), $\sigma_\pi(u) + \hat{\ell}(u) \le t_{\mathrm{dispatch}} - \delta_{\mathrm{ready}}$, directly implying $t_{\mathrm{dispatch}} \ge \max_{u \in V_\pi} C_\pi(u) + \delta_{\mathrm{ready}}$.
Next, for each admitted witness receipt $e \in W_{\mathrm{all}}$ emitted by manifest operation $u \in V_{\mathrm{manifest}}$, physical observation occurs at or after start: $t_{\mathrm{obs}}(e) \ge \sigma_\pi(u)$. The gateway receipt validity interval is $I(e) = [t_{\mathrm{obs}}(e) - \epsilon_{\mathrm{skew}}, \; t_{\mathrm{obs}}(e) + \Delta t_{\mathrm{eff}}(e) - \epsilon_{\mathrm{skew}}]$.
By condition (3), $t_{\mathrm{dispatch}} \ge \sigma_\pi(u) + \hat{\ell}(u) \ge t_{\mathrm{obs}}(e) \ge t_{\mathrm{obs}}(e) - \epsilon_{\mathrm{skew}}$ since $\epsilon_{\mathrm{skew}} \ge 0$. By condition (4), $t_{\mathrm{dispatch}} + \delta_{\mathrm{exec}} \le \sigma_\pi(u) + \Delta t_{\min}(u) - \epsilon_{\mathrm{skew}} \le t_{\mathrm{obs}}(e) + \Delta t_{\mathrm{eff}}(e) - \epsilon_{\mathrm{skew}}$.
Therefore, $[t_{\mathrm{dispatch}}, t_{\mathrm{dispatch}} + \delta_{\mathrm{exec}}] \subseteq I(e)$ holds for every witness receipt $e \in W_{\mathrm{all}}$, preserving the freshness invariant throughout proposal execution.
\end{proof}

\begin{remark}[Stochastic Latency Safety Invariant]
If realized latency exceeds $\hat{\ell}(a)$, the runtime admission boundary preserves fail-closed freshness under the stated clock and receipt assumptions: it rejects a manifest containing any receipt that expires before $t_{\mathrm{dispatch}}+\delta_{\mathrm{exec}}$, triggering replanning (Section~\ref{sec:adaptive-replanning}) or refusal. Tail latency can reduce admission availability. Freshness alone does not establish that the underlying physical predicate remains true.
\end{remark}

%% file: sections/05-epistemic-scheduling.tex
\section{Epistemic-Dependency-Aware Scheduling}
\label{sec:epistemic-scheduling}

\cac specifies when high-consequence operations require independent witnesses. \aas encodes their declared Epistemic Fault Domain (\efd) dependencies as constraints on witness selection.

\subsection{Epistemic Fault Domains and Structural Cuts}
\label{sec:efd-formalism}

Let $\mathbb{F} = \{f_1, f_2, \ldots, f_p\}$ denote the set of modeled epistemic fault domains. A fault domain $f \in \mathbb{F}$ represents an unobserved, shared point of failure---such as an operating system kernel, a hypervisor, a local network switch, a metrics scraping daemon, an external SaaS API, or a base foundation model.

Each assurance operation $a \in \mathcal{A}$ has an \emph{exposure set} $\mathcal{X}(a) \subseteq \mathbb{F}$ enumerating modeled roots on which its correctness depends. The completed fault basis includes a distinct local root $f_a\in\mathcal{X}(a)$ for each witness operation; shared roots represent correlated dependencies. Its receipts inherit at least this exposure. The gateway unions additional receipt-declared roots with the catalogued set.

\paragraph{Safe Default for Unknown Dependencies.}
If an assurance operation $a$ has an unmapped or unknown dependency structure ($\mathcal{X}(a) = \bot$ or unspecified), it must \textbf{never default to the empty set} $\emptyset$. Setting $\mathcal{X}(a) = \emptyset$ would imply that $a$ is completely immune to all known failure domains ($\forall C \subseteq \mathbb{F}, \mathcal{X}(a) \cap C = \emptyset$), falsely granting the unmapped operation artificial epistemic independence from all other verifiers.
Instead, \aas enforces the \emph{Conservative Exposure Principle}: any operation with unknown shared dependencies is assigned the universal infrastructure exposure set $\mathcal{X}(a) = \mathbb{F}_{\mathrm{infra}} \subseteq \mathbb{F}$ (or $\mathbb{F}$ if infrastructure partitions are undefined), in addition to its local root. It cannot establish independence from another witness sharing standard infrastructure.

\begin{definition}[Epistemic Diversity Cut $\kappa_E$]
\label{def:efd-cut}
For candidate witness set $W$ and obligation $\omega$, let $\Gamma_\omega$ be its positive $k_\omega$-of-$|W|$ approval rule. Write $\mathcal{W}_{\min}=\{W_d\subseteq W:|W_d|=k_\omega\}$ for its minimal decisive coalitions. A root $f$ exposes $D_W(f)=\{a\in W:f\in\mathcal{X}(a)\}$. Following decision-rule-relative \efd theory~\cite{he2026efd}, define
\begin{equation}
\begin{aligned}
\kappa_E(W,\Gamma_\omega)=\min\{\,|C|:\;&C\subseteq\mathbb{F},\\
&\exists W_d\in\mathcal{W}_{\min},\\
&W_d\subseteq\bigcup_{f\in C}D_W(f)\,\}.
\end{aligned}
\label{eq:efd-cut}
\end{equation}
Set the cut to zero if $|W|<k_\omega$. With local roots present, $1\le\kappa_E(W,\Gamma_\omega)\le k_\omega$ once a quorum is assigned. Refuting evidence vetoes admission before this positive approval rule is evaluated; arbitrary nonmonotone rules are outside the implemented model.
\end{definition}

If $\kappa_E(W,\Gamma_\omega)=1$, one root exposes a decisive approval coalition, even if other signatures remain independent. For three independent witnesses the cut is 2 under a 2-of-3 rule and 3 under unanimity. \cac requires $\kappa_E(W,\Gamma_\omega)\ge\kappa_E^{\min}(\omega)$. This structural guarantee requires a complete exposure map and a sound causal account of what exposure permits; it does not establish physical predicate truth.

\subsection{The Cost-Minimization Vulnerability}
\label{sec:cost-minimization-trap}

Standard portfolio optimization and LLM routing algorithms~\cite{chen2023frugalgpt,vpcontrol2026} rank verifiers primarily by financial cost and marginal accuracy. In distributed systems, this creates a catastrophic vulnerability:

\begin{proposition}[Correlated Selection Under Cost Minimization]
\label{prop:correlated-selection}
Let an obligation $\omega$ require $k$ affirmative witnesses. Suppose available verifiers partition into two sets:
\begin{itemize}
    \item Correlated cluster $\mathcal{V}_{\mathrm{corr}} = \{v_1, \ldots, v_m\}$ ($m \ge k$) sharing a single common telemetry exporter $f_{\mathrm{proxy}} \in \mathbb{F}$, with per-operation cost bounded by $c(v_i) \le c_{\max}$.
    \item Independent verifier $v_{\mathrm{indep}}$ with disjoint exposure ($\mathcal{X}(v_{\mathrm{indep}}) \cap \mathcal{X}(v_i) = \emptyset$) but higher financial cost satisfying $c(v_{\mathrm{indep}}) > k \cdot c_{\max}$.
\end{itemize}
Any cost-minimizing scheduler that optimizes cost $\sum_{i} c(v_i)$ without enforcing $\kappa_E(W,\Gamma_\omega) \ge 2$ will select $k$ verifiers exclusively from $\mathcal{V}_{\mathrm{corr}}$. The resulting witness set has $\kappa_E(W,\Gamma_\omega) = 1$ and will be rejected by the admission controller.
\end{proposition}
\begin{proof}
Any candidate selection $W \subseteq \mathcal{V}_{\mathrm{corr}}$ of size $k$ incurs total cost $\sum_{v \in W} c(v) \le k \cdot c_{\max}$. Conversely, any candidate selection $W' \subset \mathcal{V}_{\mathrm{corr}} \cup \{v_{\mathrm{indep}}\}$ of size $k$ containing $v_{\mathrm{indep}}$ must include $v_{\mathrm{indep}}$ and $k-1$ other verifiers; since costs are non-negative ($c(v) \ge 0$), $\operatorname{Cost}(W') = c(v_{\mathrm{indep}}) + \sum_{v \in W' \setminus \{v_{\mathrm{indep}}\}} c(v) \ge c(v_{\mathrm{indep}}) > k \cdot c_{\max} \ge \operatorname{Cost}(W)$. Therefore, any cost-minimizing selection strictly chooses $k$ verifiers exclusively from $\mathcal{V}_{\mathrm{corr}}$. However, $\{f_{\mathrm{proxy}}\} \cap \mathcal{X}(v_i) \neq \emptyset$ for all $v_i \in \mathcal{V}_{\mathrm{corr}}$. Thus $C = \{f_{\mathrm{proxy}}\}$ exposes a decisive coalition, so $\kappa_E(W,\Gamma_\omega) = 1 < 2$, causing deterministic gateway rejection.
\end{proof}

\subsection{Optimization Architecture: Decomposed Master-Subproblem Solver}
\label{sec:joint-cost-efd}

The joint assurance scheduling problem (Equation~\eqref{eq:optimization-objective}) requires simultaneously choosing a witness portfolio satisfying quorums, EFD diversity, and resource budgets, while constructing a feasible temporal schedule $(\sigma_\pi, t_{\mathrm{dispatch}})$ satisfying precedence $D_{\mathrm{op}}$ and evidence freshness intervals.

Because joint mixed-integer non-linear optimization over continuous time and combinatorial fault cuts is computationally intractable in real-time control loops, \aas adopts a \textbf{Decomposed Optimization Architecture} based on Logic-Based Benders Decomposition (LBBD):
\begin{enumerate}
    \item \textbf{Master Selection Problem (0-1 ILP):} Solves the combinatorial witness selection and obligation assignment problem, minimizing financial cost subject to quorums, \efd cuts, and hard resource budgets ($B_{\$}, B_{\mathrm{tok}}, B_{\mathrm{risk}}$).
    \item \textbf{Temporal Scheduling Subproblem (Algorithm~\ref{alg:backward-scheduler}):} Solves the continuous backward JIT scheduling pass over the candidate witness assignments $\mathbf{y}^*$ and operational DAG $G[V^*]$ under precedence $D_{\mathrm{op}}$ and freshness intervals.
    \item \textbf{Sound Conflict Generation \& Refinement:} If the subproblem detects temporal infeasibility, it distinguishes assignment-induced freshness conflicts from structural precedence bottlenecks, generating provably sound Benders cuts.
\end{enumerate}

\paragraph{Operation Roles and Distinctions.}
To guarantee soundness in cut generation, \aas strictly distinguishes five categories of operations:
\begin{enumerate}
    \item \textbf{Selected Operations ($V^* = \{a \mid x_a = 1\}$):} Operations committed for execution, which consume financial cost $c(a)$, tokens $\operatorname{tok}(a)$, and risk budget $\|\rho(a)\|$.
    \item \textbf{Witness Assignments ($\mathbf{y}^* = \{ (a, \omega) \mid y_{a, \omega} = 1 \}$):} Operations whose emitted receipts are actively assigned to discharge obligation $\omega$. Freshness constraints apply \textbf{only} to assigned pairs.
    \item \textbf{Manifest Contributors ($V_{\mathrm{manifest}} = \{a \in V^* \mid \exists \omega, y_{a, \omega} = 1\}$):} The subset of selected operations that produce receipts for the final admission certificate.
    \item \textbf{Precedence Prerequisites ($V_{\mathrm{prec}} = \operatorname{Anc}_{D_{\mathrm{op}}}(V_{\mathrm{manifest}}) \setminus V_{\mathrm{manifest}}$):} Operations selected through prerequisite closure that must complete before a manifest contributor can launch, even if they emit no admission receipts. Thus $V_{\mathrm{prec}}\subseteq V^*$.
    \item \textbf{Nonmanifest Operations ($V_{\mathrm{nonmanifest}} = V^* \setminus V_{\mathrm{manifest}}$):} Selected operations whose receipts are unassigned or superseded, including mandatory prerequisites. Their receipt freshness does not constrain dispatch, but their completion and resource use still do.
\end{enumerate}

\paragraph{Counterexample: Unsoundness of Operation-Set Conflict Cuts.}
A naive Benders formulation adds the conflict cut $\sum_{a \in V_{\mathrm{conflict}}} x_a \le |V_{\mathrm{conflict}}| - 1$ whenever portfolio $V_{\mathrm{conflict}}$ fails temporal scheduling. \textbf{This cut is unsound}:
\begin{example}[Alternative Witness Relief]
\label{ex:benders-counterexample}
Suppose Master ILP selects $V^* = \{a_1, a_2\}$ with assignments $y_{a_1, \omega_1} = 1$ and $y_{a_2, \omega_2} = 1$. Operation $a_1$ has long latency $\hat{\ell}(a_1) = 8$\,s, but obligation $\omega_1$ has a short freshness lifespan $\Delta t_{\omega_1} = 3$\,s.
The temporal subproblem determines that assigning $a_1$ to $\omega_1$ is infeasible ($[t_{\min}, t_{\max}] = \emptyset$).
If the scheduler appends an operation-set cut $x_{a_1} + x_{a_2} \le 1$, it permanently forbids selecting both $a_1$ and $a_2$ in any future solution.

Now suppose there exists an alternative operation $a_3$ with short latency $\hat{\ell}(a_3) = 1$\,s that also covers $\omega_1$, while operation $a_1$ is also eligible to cover a third obligation $\omega_3$ with a generous lifespan $\Delta t_{\omega_3} = 45$\,s.
Consider the expanded portfolio $V' = \{a_1, a_2, a_3\}$ with reassigned witnesses $y_{a_3, \omega_1} = 1$, $y_{a_2, \omega_2} = 1$, and $y_{a_1, \omega_3} = 1$.
This solution is completely feasible and may be cost-optimal! However, the naive cut $x_{a_1} + x_{a_2} \le 1$ incorrectly excludes $V'$, pruning a valid, optimal solution.
\end{example}

\paragraph{The Sound Master Selection ILP.}
For each operation $a \in \mathcal{A}$ and obligation $\omega \in \Omega$, let $\operatorname{Eligible}(a, \omega) \iff \exists \varepsilon \in \mathcal{E}_\omega \text{ s.t. } (\omega, \varepsilon) \in \operatorname{Cov}(a)$.
We formulate the Master ILP under Contract C:
\begin{align}
\min_{\mathbf{x}, \mathbf{y}} \quad & \sum_{a \in \mathcal{A}} c(a) \cdot x_a \label{eq:ilp-obj} \\
\text{subject to} \quad & y_{a, \omega} \le x_a \quad \forall a \in \mathcal{A}, \; \omega \in \Omega, \label{eq:ilp-link} \\
& x_a \le x_u \quad \forall a \in \mathcal{A}, \; \forall u \in \operatorname{Prereq}(a), \label{eq:ilp-prereq} \\
& y_{a, \omega} = 0 \quad \forall a, \omega \text{ with } \neg\operatorname{Eligible}(a, \omega), \label{eq:ilp-elig} \\
& \sum_{a \in \mathcal{A}} y_{a, \omega} \ge k_\omega \quad \forall \omega \in \Omega, \label{eq:ilp-quorum} \\
& \sum_{a: \mathcal{X}(a) \cap C \neq \emptyset} y_{a, \omega} \le k_\omega-1 \label{eq:ilp-cut} \\
& \qquad \forall \omega \in \Omega, \; \forall C \subset \mathbb{F} \text{ s.t. } |C| < \kappa_E^{\min}(\omega), \nonumber \\
& \sum_{a \in \mathcal{A}} c(a) \cdot x_a \le B_{\$}, \label{eq:ilp-cost} \\
& \sum_{a \in \mathcal{A}} \operatorname{tok}(a) \cdot x_a \le B_{\mathrm{tok}}, \label{eq:ilp-tok} \\
& \sum_{a \in \mathcal{A}} \|\rho(a)\| \cdot x_a \le B_{\mathrm{risk}}, \label{eq:ilp-risk} \\
& \sum_{(a, \omega) \in A_{\mathrm{conflict}}} y_{a, \omega} \le |A_{\mathrm{conflict}}| - 1 \label{eq:ilp-assign-cut} \\
& \qquad \forall A_{\mathrm{conflict}} \in \mathcal{K}_{\mathrm{assign}}, \nonumber \\
& \sum_{a \in V_{\mathrm{chain}}} x_a \le |V_{\mathrm{chain}}| - 1 \label{eq:ilp-chain-cut} \\
& \qquad \forall V_{\mathrm{chain}} \in \mathcal{K}_{\mathrm{chain}}, \nonumber \\
& x_a \in \{0, 1\}, \; y_{a, \omega} \in \{0, 1\} \quad \forall a \in \mathcal{A}, \; \omega \in \Omega. \label{eq:ilp-binary}
\end{align}

The diversity row forbids any fault set smaller than the required cut from exposing $k_\omega$ assigned witnesses. The implementation synthesizes local roots for exact validation. In the master ILP it enumerates shared-root rows: a local root can be replaced by a declared shared root of the same witness without shrinking the exposed set, while local roots alone cannot expose a quorum with fewer than $k_\omega$ failures. Thus the omitted local rows are redundant when the required threshold is at most $k_\omega$.

\paragraph{Formulation and Cut Semantics.}
\begin{itemize}
    \item \textbf{Exact Financial and Prerequisite Accounting:} Constraint~\eqref{eq:ilp-link} links witness assignments to operation selection, while Constraint~\eqref{eq:ilp-prereq} enforces transitive operational prerequisite closure ($x_a \le x_u$). This guarantees that selecting an operation automatically forces the selection of all required predecessor operations, charging each operation $c(a)$ exactly once in objective~\eqref{eq:ilp-obj} and budget~\eqref{eq:ilp-cost}, regardless of multi-obligation witness coverage ($\sum_\omega y_{a, \omega} \ge 2$).
    \item \textbf{Assignment-Aware No-Good Cuts ($\mathcal{K}_{\mathrm{assign}}$):} When temporal scheduling fails due to evidence expiration or lack of freshness intersection, the authoritative subproblem isolates the certified minimal conflicting assignments $A_{\mathrm{conflict}} = \{(a, \omega) \mid y_{a, \omega}^* = 1 \text{ active in conflict}\}$. Constraint~\eqref{eq:ilp-assign-cut} eliminates this specific assignment without forbidding operations in $A_{\mathrm{conflict}}$ from being selected or reassigned to other obligations. When heuristic placement exhausts without certified proof ($\textsc{NotFound}$), the solver applies a safe full-candidate combinatorial no-good cut ($\sum_{a: x_a^*=1}(1-x_a) + \sum_{a: x_a^*=0} x_a \ge 1$) to avoid over-pruning alternative witness assignments.
    \item \textbf{Precedence-Chain Structural Cuts ($\mathcal{K}_{\mathrm{chain}}$):} When failure is caused by an unavoidable directed precedence chain $V_{\mathrm{chain}} \subseteq \mathcal{A}$ in $D_{\mathrm{op}}$ whose cumulative conservative latency exceeds available deadline:
    \begin{equation}
    \sum_{u \in V_{\mathrm{chain}}} (\hat{\ell}(u) + \epsilon_{\mathrm{guard}}) + \delta_{\mathrm{ready}} + \delta_{\mathrm{exec}} > T_{\mathrm{dead}} - t_0,
    \end{equation}
    latency monotonicity guarantees that any portfolio containing $V_{\mathrm{chain}}$ is temporally infeasible. Only in this structural case is the operation-level cut~\eqref{eq:ilp-chain-cut} valid.
\end{itemize}

\paragraph{Financial Optimization Scope and Consequential Coordination.}
The Master ILP objective~\eqref{eq:ilp-obj} minimizes the direct financial cost $\operatorname{Cost}(\pi) = \sum_{a \in \mathcal{A}} c(a) x_a$. When all assurance operations in $\mathcal{A}$ are non-consequential ($\rho(a) \prec \tau_\Pi$), direct cost equals total cost: $\operatorname{Cost}_{\mathrm{total}}(\pi) = \operatorname{Cost}(\pi)$ (Section~\ref{sec:problem-formulation}). When consequential operations require diagnostic sub-plans (Section~\ref{sec:consequential-acquisition}), cost optimization has two cases:
\begin{enumerate}
    \item \textbf{Statically Modeled / Flattened Instances:} If candidate diagnostic trees are fully expanded and flattened into the catalog $\mathcal{A}$ with their prerequisites, risks, and shared costs accounted for, the Master ILP jointly minimizes $\operatorname{Cost}_{\mathrm{total}}$.
    \item \textbf{Runtime Consequential Coordination:} In dynamic systems where diagnostic sub-plans are synthesized recursively at runtime via local risk stratification (Algorithm~\ref{alg:complete-aas}), the static Master ILP optimizes direct operational cost, while recursive admission coordinates sub-plan synthesis and budget reservation dynamically.
\end{enumerate}
We note that direct-cost minimization does not guarantee joint consequential optimality without flattening:
\begin{example}[Cheap Verifier with Expensive Diagnostics]
\label{ex:cheap-verifier-counterexample}
Suppose obligation $\omega$ can be satisfied by either $a_1$ ($c(a_1) = \$1.00$, but consequential with diagnostic sub-plan $\pi_{a_1}$ costing $c(\pi_{a_1}) = \$100.00$) or $a_2$ ($c(a_2) = \$5.00$, non-consequential with $c(\pi_{a_2}) = \$0.00$). A static optimizer minimizing only direct cost strictly selects $a_1$ (\$1.00 vs.\ \$5.00). Upon runtime admission, $a_1$ admits its diagnostic sub-plan, yielding total cost $\operatorname{Cost}_{\mathrm{total}} = \$101.00 \gg \$5.00$. While shadow reservation ledgers prevent budget overruns, global cost-optimality $\operatorname{Cost}_{\mathrm{total}}$ requires pre-flattening diagnostic expansions into $\mathcal{A}$.
\end{example}

\paragraph{Decomposition Interface, Convergence, and Optimality.}
The master-subproblem solver operates in an iterative loop:
\begin{enumerate}
    \item Solve Master ILP~\eqref{eq:ilp-obj}--\eqref{eq:ilp-binary}. If infeasible, terminate with safe refusal ($\Pi_{\mathrm{adm}}(\mathcal{P}) = \emptyset$).
    \item Pass candidate assignments $\mathbf{y}^*$ and operational DAG $G[V^* \cup \operatorname{Anc}_{D_{\mathrm{op}}}(V^*)]$ to the Authoritative Temporal Solver (Algorithm~\ref{alg:backward-scheduler}).
    \item If the temporal solver returns a valid schedule $(\sigma_\pi, t_{\mathrm{dispatch}})$, terminate. Plan $\pi$ is prospectively admissible.
    \item If the temporal solver proves $\textsc{Infeasible}$:
    \begin{itemize}
        \item If the critical path of an induced precedence chain exceeds $T_{\mathrm{dead}} - t_0 - \delta_{\mathrm{exec}}$, append structural cut~\eqref{eq:ilp-chain-cut} to $\mathcal{K}_{\mathrm{chain}}$.
        \item Otherwise, when certified by the exact oracle, isolate minimal conflicting witness assignments $A_{\mathrm{conflict}} \subseteq \mathbf{y}^*$ and append assignment cut~\eqref{eq:ilp-assign-cut} to $\mathcal{K}_{\mathrm{assign}}$. If heuristic search exhausts without certified proof ($\textsc{NotFound}$), append a safe full-candidate combinatorial no-good cut to avoid unsound pruning of unvisited assignments.
    \end{itemize}
    Return to Step~1.
\end{enumerate}

\begin{theorem}[Convergence and Conditional Global Optimality]
\label{thm:benders-optimality}
The decomposed LBBD algorithm terminates in a finite number of iterations. Furthermore, for any fully modeled or flattened problem $\mathcal{P}$, if the master solver executes to completion without timeout ($T_{\mathrm{ctrl}} = \infty$) and the subproblem is evaluated authoritatively via an exact temporal oracle (Section~\ref{sec:eval-optimality}), the returned plan $\pi^*$ is \textbf{globally cost-optimal} over all prospectively admissible plans:
\begin{equation}
\operatorname{Cost}(\pi^*) = \min_{\pi \in \Pi_{\mathrm{adm}}(\mathcal{P})} \sum_{a \in V_\pi} c(a).
\end{equation}
When evaluated under bounded concurrency $K_{\max} < \infty$ with the polynomial-time slack-priority heuristic, or upon solver timeout $T_{\mathrm{ctrl}} < \infty$, the returned plan is \textbf{approximately optimized and prospectively feasible}, with no claim of global optimality. Dynamic consequential sub-plans admitted at runtime are coordinated outside this static optimality guarantee.
\end{theorem}
\begin{proof}
The assignment space $\{0, 1\}^{|\mathcal{A}| \times |\Omega|}$ and portfolio space $2^{|\mathcal{A}|}$ are finite. Each iteration generates either an assignment cut~\eqref{eq:ilp-assign-cut} (or safe full-candidate cut) that prunes at least one integer assignment vector $\mathbf{y}^*$ or a structural cut~\eqref{eq:ilp-chain-cut} that prunes at least one operation subset. Both cut families are sound: structural cuts prune operation subsets whose topological critical path latency strictly exceeds the deadline window ($t_{\mathrm{CP}} + \delta_{\mathrm{ready}} > T_{\mathrm{dead}} - t_0 - \delta_{\mathrm{exec}}$), while assignment cuts prune certified conflicting witness assignments. When subproblem feasibility is determined via exact temporal search, no prospectively admissible plan $\pi \in \Pi_{\mathrm{adm}}(\mathcal{P})$ is ever falsely pruned. Because the Master ILP optimizes the exact linear cost objective $\sum_a c(a) x_a$ over an admissible outer relaxation, the first verified feasible candidate achieves global cost optimality over the flattened instance $\mathcal{P}$.
\end{proof}

\subsection{Complexity and EFD Cut Separation}
\label{sec:asp-complexity}

To understand the computational tractability of \aas, we rigorously distinguish four related problem variants:
\begin{enumerate}
    \item \textbf{Checking Diversity of a Fixed Witness Set:} Given $W \subseteq \mathcal{A}$ and threshold $h$, determining whether $\kappa_E(W,\Gamma_\omega) \ge h$ requires checking that each $C\subseteq\mathbb F$ with $|C|<h$ exposes fewer than $k_\omega$ assigned witnesses. For fixed $h$, there are $O(|\mathbb F|^{h-1})$ candidates and each check costs $O(|W|h)$.
    \item \textbf{Separating Violated ILP Diversity Cuts:} Given candidate assignments $(\mathbf{x}, \mathbf{y})$, finding a violated cut $C$ for Constraint~\eqref{eq:ilp-cut} requires evaluating the $O(|\mathbb{F}|^{k-1})$ candidate fault sets. For $k=2$, $|C|=1$, so there are only $|\mathbb{F}|$ cuts per obligation, which can be instantiated statically!
    \item \textbf{Unbounded Cut Computation:} Finding a minimum root coalition exposing a decisive quorum generalizes set-cover variants and is NP-hard for arbitrary thresholds~\cite{karp1972reducibility}. Small fixed thresholds ($h\in\{2,3\}$) permit polynomial-time checking.
    \item \textbf{The ASP Decision Problem:} Selecting a minimum-cost witness set satisfying multi-obligation coverage and quorums is NP-hard.
\end{enumerate}

\begin{theorem}[NP-Hardness of ASP Decision Problem]
\label{thm:np-hard}
The decision version of the Assurance Scheduling Problem (\textsc{ASP-DEC})---determining whether there exists a prospectively admissible plan $\pi \models_{\mathrm{pros}} \mathcal{P}$ with financial cost $\operatorname{Cost}(\pi) \le K$---is NP-hard, even in the absence of precedence constraints ($D_{\mathrm{op}} = \emptyset$) and with unit operation latencies.
\end{theorem}
\begin{proof}[Proof Sketch]
We establish a polynomial-time reduction from \textsc{Minimum Weight Set Cover} to \textsc{ASP-DEC}. Given universe $\mathcal{U} = \{u_1, \ldots, u_n\}$ and candidate subsets $\mathcal{S} = \{S_1, \ldots, S_m\}$ with weights $w(S_j)$, we map each universe element $u_i$ to an obligation $\omega_i \in \Omega$ and each subset $S_j$ to an assurance operation $a_j \in \mathcal{A}$ with cost $c(a_j) = w(S_j)$ and coverage $\operatorname{Cov}(a_j) = \{(\omega_i, \varepsilon) \mid u_i \in S_j\}$. Setting unit quorums, a set cover of weight at most $K$ exists if and only if there exists an admissible plan of cost at most $K$. Full bidirectional reduction and equivalence proofs are given in Appendix~\ref{app:proofs}.
\end{proof}

\subsection{Progress-Based Greedy Fallback Heuristic}
\label{sec:greedy-efd}

When the controller computation budget is constrained ($T_{\mathrm{ctrl}} < 5$\,ms), \aas executes an efficient \textbf{Progress-Based Marginal Score Heuristic}:

For an active partial plan with selected operations $V_\pi$ and assigned witnesses $W_\omega \subseteq V_\pi$, define the provisional cut $\widetilde\kappa_E(W,\omega)=\kappa_E(W,\Gamma_{\min(k_\omega,|W|)})$ for nonempty $W$, and zero for empty $W$. This provisional score rewards distinct roots while a quorum is assembled; admission always checks the fixed rule $\Gamma_\omega$.
\begin{enumerate}
    \item \textbf{Quorum Progress:} $\Delta \mathcal{Q}(a, \omega) = 1$ if $\operatorname{Eligible}(a, \omega)$ and $|W_\omega| < k_\omega$; 0 otherwise.
    \item \textbf{Diversity Progress:} $\Delta \kappa(a, \omega) = \max(0, \; \min(\widetilde\kappa_E(W_\omega \cup \{a\},\omega), \kappa_E^{\min}(\omega)) - \widetilde\kappa_E(W_\omega,\omega))$.
    \item \textbf{Marginal Cost:} If $a \in V_\pi$, $\tilde{c}(a) = 0$; if $a \notin V_\pi$, $\tilde{c}(a) = c(a)$.
    \item \textbf{Temporal Feasibility Guard:} $\mathbf{1}_{\mathrm{temp}}(a) = 1$ if adding $a$ to $V_\pi$ preserves a feasible JIT completion window within remaining latency and deadline bounds; 0 otherwise.
\end{enumerate}

At each iteration, the heuristic selects:
\begin{equation}
a^* = \arg\max_{a : \mathbf{1}_{\mathrm{temp}}(a) = 1} \frac{\sum_{\omega \in \Omega} \left[ \Delta \mathcal{Q}(a, \omega) + \lambda_{\kappa} \Delta \kappa(a, \omega) \right]}{\tilde{c}(a) + \epsilon},
\label{eq:greedy-efd-heuristic}
\end{equation}
where $\lambda_{\kappa} > 0$ weights fault-cut diversity relative to quorum count.
The reference heuristic computes the exact provisional cut by enumerating root coalitions up to $\min(k_\omega,|W_\omega|+1)$; with local roots, its worst-case enumeration is $O((|\mathbb F|+|W_\omega|)^{k_\omega})$ for bounded quorum size. It may fail after choosing an unfavorable quorum; it is not a completeness guarantee.

\begin{example}[Greedy Heuristic Execution Trajectories]
To illustrate heuristic behavior under trade-offs:
\begin{itemize}
    \item \textbf{Quorum vs. EFD Diversity:} Suppose obligation $\omega_1$ requires quorum $k_{\omega_1} = 2$ and diversity $\kappa_E^{\min} = 2$. Let current witness set be $W = \{a_1\}$ with exposure $\mathcal{X}(a_1) = \{f_1\}$. Candidate $a_{\mathrm{cheap}}$ ($c=\$0.01$) has exposure $\{f_1\}$, while candidate $a_{\mathrm{indep}}$ ($c=\$0.03$) has exposure $\{f_2\}$.
    Candidate $a_{\mathrm{cheap}}$ yields $\Delta \mathcal{Q} = 1$, but $\widetilde\kappa_E(\{a_1, a_{\mathrm{cheap}}\},\omega_1) = 1$ (since $\{f_1\}$ hits both), so $\Delta \kappa = 0$. Its score is $1 / (0.01 + \epsilon) \approx 100$.
    Candidate $a_{\mathrm{indep}}$ yields $\Delta \mathcal{Q} = 1$ and $\Delta \kappa = 1$. With $\lambda_\kappa = 3$, its score is $(1 + 3) / (0.03 + \epsilon) \approx 133$.
    The heuristic strictly prefers the independent verifier $a_{\mathrm{indep}}$, successfully steering away from the correlated trap.
    \item \textbf{Diversity vs. Temporal Deadlines:} Suppose obligation $\omega_2$ requires diversity $\kappa_E^{\min} = 2$. Candidate $a_{\mathrm{solver}}$ provides perfect epistemic diversity ($\Delta \kappa = 1$), but is a heavy formal model checker requiring conservative latency $\hat{\ell}(a_{\mathrm{solver}}) = 8$\,s. If the remaining deadline budget is $T_{\mathrm{dead}} - t_{\mathrm{now}} = 5$\,s, Backward JIT scheduling detects that $a_{\mathrm{solver}}$ cannot complete before dispatch. The temporal guard evaluates $\mathbf{1}_{\mathrm{temp}}(a_{\mathrm{solver}}) = 0$, immediately pruning $a_{\mathrm{solver}}$ from consideration and preventing a dead-end plan commitment.
\end{itemize}
\end{example}

%% file: sections/06-consequential-acquisition.tex
\section{Recursive and Consequential Acquisition}
\label{sec:consequential-acquisition}

The \cac paper~\cite{he2026cac} identifies a remediation-liveness problem: a diagnostic action may require another diagnostic action to be admitted first. Budget bounds alone do not ensure progress. We give graph and risk conditions for finite diagnostic expansion and state the additional assumptions needed for operational liveness.

\subsection{When Evidence Acquisition is Consequential}
\label{sec:consequential-diagnostics}

Not all evidence acquisition is passive telemetry scraping. In real-world distributed architectures, establishing that a complex predicate holds often requires \emph{active systems probing}:
\begin{itemize}
    \item \textbf{Storage Rollback Validation:} Freezing a database tablespace and mounting a copy-on-write snapshot to verify point-in-time recovery.
    \item \textbf{Network Isolation Verification:} Injecting synthetic probe traffic or testing border gateway BGP route filtering.
    \item \textbf{Failover Readiness:} Initiating an ephemeral dry-run transition on an auxiliary standby replica.
\end{itemize}

Under \cac, \emph{any} action whose modeled operational risk equals or exceeds the policy threshold $\tau_\Pi$ is classified as consequential:
\begin{equation}
\rho(a) \succeq \tau_\Pi \implies a \text{ is consequential}.
\end{equation}
Because the execution gateway completely mediates all infrastructure interactions, a consequential diagnostic operation $a \in \mathcal{A}$ \textbf{cannot bypass the admission boundary}. The gateway will reject $a$ unless presented with a valid admission certificate $\mathcal{C}_a$ satisfying its own obligations $\Omega_\Pi(a, s)$.

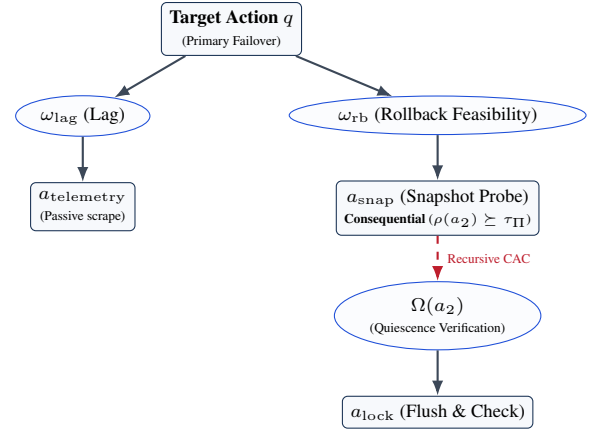
\begin{figure}[t]
\centering
\begin{tikzpicture}[
  actionnode/.style={draw=navy,rounded corners=2pt,fill=softgray,font=\scriptsize,align=center,inner sep=3pt},
  obnode/.style={draw=cobalt,ellipse,fill=softgray!50,font=\scriptsize,align=center,inner sep=2pt},
  arr/.style={-{Latex[length=2mm]},thick,darkslate},
  deparr/.style={-{Latex[length=2mm]},thick,dashed,crimson}]

  \node[actionnode] (q) {\textbf{Target Action $q$}\\\tiny (Primary Failover)};
  \node[obnode,below left=0.6cm and 0.4cm of q] (w1) {$\omega_{\mathrm{lag}}$ (Lag)};
  \node[obnode,below right=0.6cm and 0.4cm of q] (w2) {$\omega_{\mathrm{rb}}$ (Rollback Feasibility)};

  \node[actionnode,below=0.6cm of w1] (a1) {$a_{\mathrm{telemetry}}$\\\tiny (Passive scrape)};
  \node[actionnode,below=0.6cm of w2] (a2) {$a_{\mathrm{snap}}$ (Snapshot Probe)\\\tiny \textbf{Consequential} ($\rho(a_2) \succeq \tau_{\Pi}$)};

  \node[obnode,below=0.6cm of a2] (w3) {$\Omega(a_2)$\\\tiny (Quiescence Verification)};
  \node[actionnode,below=0.6cm of w3] (a3) {$a_{\mathrm{lock}}$ (Flush \& Check)};

  \draw[arr] (q) -- (w1);
  \draw[arr] (q) -- (w2);
  \draw[arr] (w1) -- (a1);
  \draw[arr] (w2) -- (a2);
  \draw[deparr] (a2) -- node[right,font=\tiny,text=crimson] {Recursive CAC} (w3);
  \draw[arr] (w3) -- (a3);
\end{tikzpicture}
\caption{The Assurance Dependency Graph (ADG). Consequential diagnostic action $a_{\mathrm{snap}}$ requires its own admission obligations $\Omega(a_{\mathrm{snap}})$, spawning recursive sub-scheduling.}
\label{fig:adg-tree}
\end{figure}

\subsection{The Assurance Dependency Graph (ADG)}
\label{sec:adg-formulation}

This recursive dependency induces a directed bipartite graph of alternating actions and obligations (Figure~\ref{fig:adg-tree}):

\begin{definition}[Assurance Dependency Graph (ADG)]
The Assurance Dependency Graph $\mathcal{G}_{\mathrm{ADG}} = (\mathcal{V}_A \cup \mathcal{V}_\Omega, \mathcal{E})$ is a directed bipartite graph where:
\begin{itemize}
    \item $\mathcal{V}_A \subseteq \mathcal{A} \cup \{q\}$ are action nodes,
    \item $\mathcal{V}_\Omega \subseteq \bigcup_{u \in \mathcal{V}_A} \Omega(u)$ are obligation nodes,
    \item Edge $(u, \omega) \in \mathcal{E}$ indicates that action $u$ requires obligation $\omega$,
    \item Edge $(\omega, a) \in \mathcal{E}$ indicates that operation $a$ is scheduled to produce a witness discharging $\omega$.
\end{itemize}
\end{definition}

Without formal safeguards, recursive acquisition risks two structural pathologies:
\begin{enumerate}
    \item \textbf{Cyclic Admission Deadlock:} Operation $a_1$ requires evidence from $a_2$, while $a_2$ requires evidence from $a_1$ (e.g., verifying replication requires snapshot lock, while snapshot lock requires replication sync).
    \item \textbf{Infinite Remediation Regress:} Each diagnostic probe requires a further diagnostic probe of equal or greater consequence, depleting budgets without reaching base evidence.
\end{enumerate}

\subsection{Structural Acyclicity and Bounded Expansion}
\label{sec:well-foundedness}

To prevent cyclic dependencies and infinite regress during plan synthesis, \aas enforces a \emph{Risk Stratification Discipline}.

\begin{definition}[Strict Risk Stratification of Finite Height]
\label{def:risk-stratification}
Let $\prec_{\mathcal{R}}$ be a strict partial order on the risk space $\mathcal{R}$ with \textbf{finite height} $H_{\mathcal{R}} = \operatorname{height}(\mathcal{R}) < \infty$ (i.e., every strictly descending chain has length at most $H_{\mathcal{R}}$). An admission policy is \emph{strictly risk-stratified} if, for every consequential action $u$, any diagnostic action $a$ invoked to discharge $\omega \in \Omega(u)$ satisfies:
\begin{equation}
\rho(a) \prec_{\mathcal{R}} \rho(u).
\label{eq:risk-stratification}
\end{equation}
Furthermore, non-consequential operations form a base stratum $\mathcal{R}_{\mathrm{base}} = \{a \in \mathcal{A} \mid \rho(a) \prec \tau_\Pi\}$ whose obligations are empty by policy definition: $\forall a \in \mathcal{R}_{\mathrm{base}}, \Omega(a) = \emptyset$.
\end{definition}

\begin{proposition}[Structural Acyclicity of the ADG]
\label{prop:adg-acyclic}
If an admission policy $\Pi$ is strictly risk-stratified, then the Assurance Dependency Graph $\mathcal{G}_{\mathrm{ADG}}$ is a directed acyclic graph (DAG).
\end{proposition}
\begin{proof}
Suppose for contradiction that $\mathcal{G}_{\mathrm{ADG}}$ contains a directed cycle: $u_1 \to \omega_1 \to u_2 \to \omega_2 \to \cdots \to u_m \to \omega_m \to u_1$. By Definition~\ref{def:risk-stratification}, every composite step $u_i \to \omega_i \to u_{i+1}$ implies $\rho(u_{i+1}) \prec_{\mathcal{R}} \rho(u_i)$. By transitivity of the strict partial order, $\rho(u_1) \prec_{\mathcal{R}} \rho(u_1)$, contradicting irreflexivity. Thus $\mathcal{G}_{\mathrm{ADG}}$ contains no directed cycles.
\end{proof}

\begin{proposition}[Bounded Recursive Expansion of Bipartite ADG]
\label{prop:bounded-expansion}
Assume admission policy $\Pi$ is strictly risk-stratified of finite height $H_{\mathcal{R}} = \operatorname{height}(\mathcal{R}) < \infty$. Let $K_\Omega = \max_{u \in \mathcal{A} \cup \{q\}} |\Omega(u)| < \infty$ be the maximum number of obligations per action, and let each obligation require at most $K_{\mathrm{inst}} < \infty$ concrete action instances dynamically instantiated from capability templates $\mathcal{T}_{\mathrm{cap}}$ ($K_{\mathrm{inst}} \le |\mathcal{T}_{\mathrm{cap}}| \cdot K_{\mathrm{inst}}^{\mathrm{per}}$). Then:
\begin{enumerate}
    \item The maximum action-recursion depth of $\mathcal{G}_{\mathrm{ADG}}$ is strictly bounded by the ordering height: $\operatorname{depth}_A(\mathcal{G}_{\mathrm{ADG}}) \le H_{\mathcal{R}} < \infty$.
    \item The effective action branching factor is finite: $B_{\mathrm{act}} \le K_\Omega \cdot K_{\mathrm{inst}} < \infty$.
    \item The total number of concrete action nodes $\mathcal{V}_A$ in the expanded ADG is strictly bounded by the sum:
    \begin{equation}
    M_A = |\mathcal{V}_A| \le \sum_{d=0}^{H_{\mathcal{R}}} B_{\mathrm{act}}^d =
    \begin{cases}
    1 & \text{if } B_{\mathrm{act}} = 0, \\
    H_{\mathcal{R}} + 1 & \text{if } B_{\mathrm{act}} = 1, \\
    \frac{B_{\mathrm{act}}^{H_{\mathcal{R}}+1} - 1}{B_{\mathrm{act}} - 1} & \text{if } B_{\mathrm{act}} > 1,
    \end{cases}
    \label{eq:max-action-nodes}
    \end{equation}
    where $M_A < \infty$ in all cases. The total number of obligation nodes satisfies $|\mathcal{V}_\Omega| \le K_\Omega \cdot M_A < \infty$.
\end{enumerate}
\end{proposition}
\begin{proof}
Every directed path of action nodes in bipartite graph $\mathcal{G}_{\mathrm{ADG}}$ corresponds to a strictly descending sequence of risk values $\rho(u_0) \succ_{\mathcal{R}} \rho(u_1) \succ_{\mathcal{R}} \cdots \succ_{\mathcal{R}} \rho(u_d)$ alternating through obligation nodes $\omega_i \in \Omega(u_{i-1})$. Because $(\mathcal{R}, \prec_{\mathcal{R}})$ has finite height $H_{\mathcal{R}}$, any such chain contains at most $H_{\mathcal{R}}$ action transitions. Each action node expands into at most $K_\Omega$ obligation nodes, each of which links to at most $K_{\mathrm{inst}}$ candidate diagnostic actions, establishing branching factor $B_{\mathrm{act}} \le K_\Omega \cdot K_{\mathrm{inst}}$. Summing action nodes across depths $0 \le d \le H_{\mathcal{R}}$ yields the geometric sum~\eqref{eq:max-action-nodes}. Leaves terminate in non-consequential base actions $\mathcal{R}_{\mathrm{base}}$ with $\Omega(a) = \emptyset$, terminating expansion in finite steps.
\end{proof}

\subsection{Operational Diagnostic Liveness}
\label{sec:diagnostic-liveness}

We emphasize an essential systems distinction: \textbf{structural acyclicity and bounded expansion do not automatically imply operational termination during physical execution}. A structurally acyclic plan may still hang at runtime due to unhandled controller delays, unbounded retries, deadlock on shared physical locks, or event-loop stalls.

To establish operational liveness, \aas couples risk stratification with six explicit systems execution bounds:

\begin{theorem}[Conditional Operational Diagnostic Liveness]
\label{thm:diagnostic-liveness}
Assume an admission policy is strictly risk-stratified of finite height $H_{\mathcal{R}} < \infty$ with bounded concrete instances $B_{\mathrm{act}} < \infty$, yielding an expanded ADG of at most $M_A$ actions. Diagnostic evidence acquisition is guaranteed to terminate deterministically in finite wall-clock elapsed time---either minting certificate $\mathcal{C}_q$ or safely refusing proposal $q$---provided the distributed runtime satisfies:
\begin{enumerate}
    \item \textbf{Bounded Operation Latency:} Every executed operation $a \in \mathcal{A}$ has an enforced wall-clock timeout $\hat{\ell}(a) \le \hat{\ell}_{\max} < \infty$.
    \item \textbf{Bounded Cancellation and Cleanup Delay:} Transmitting cancellation signals and terminating an in-flight worker thread requires at most $\delta_{\mathrm{cancel}} < \infty$ time, and controller state cleanup requires at most $\delta_{\mathrm{cleanup}} < \infty$ time.
    \item \textbf{Bounded Controller Computation:} Every invocation of the plan synthesizer, ILP solver, or replanning monitor terminates within controller computation timeout $T_{\mathrm{ctrl}} < \infty$ (falling back to greedy heuristic or refusal upon timeout).
    \item \textbf{Finite Retry and Generation Limits:} The controller permits at most $N_{\mathrm{retry}}^{\max} < \infty$ adaptive repair or retry attempts across the entire lifecycle of proposal $q$.
    \item \textbf{Deadlock-Free Resource and RPC Execution:} Mutating locks or physical hardware mutexes required by concurrent diagnostic operations are acquired in a canonical total order and released within bounded time $\delta_{\mathrm{rel}} < \infty$ (or executed in copy-on-write sandboxes). All remote RPCs and network verifiers enforce non-blocking socket timeouts $\hat{\ell}_{\max}$, preventing unbounded waiting on external queues.
    \item \textbf{Non-Blocking Event Loop \& Epoch-Relative Deadline Abort:} The controller event loop never blocks indefinitely awaiting external messages; if physical wall-clock time reaches proposal deadline $T_{\mathrm{dead}}$, all active workers are cancelled, state is reconciled, and the workflow deterministically aborts via safe refusal.
\end{enumerate}
Under these conditions, the total elapsed wall-clock time $T_{\mathrm{term}}$ from scheduling epoch $t_0$ until complete system quiescence (all worker threads terminated and resources released) is strictly bounded by:
\begin{equation}
\begin{aligned}
T_{\mathrm{term}} \le \min \Big( &(T_{\mathrm{dead}} - t_0) + \delta_{\mathrm{cancel}} + \delta_{\mathrm{cleanup}}, \\
&(N_{\mathrm{retry}}^{\max} + 1) \cdot T_{\mathrm{ctrl}} \\
&\quad + (N_{\mathrm{retry}}^{\max} + 1) \cdot M_A \cdot (\hat{\ell}_{\max} + \delta_{\mathrm{rel}}) \\
&\quad + N_{\mathrm{retry}}^{\max} \cdot (\delta_{\mathrm{cancel}} + \delta_{\mathrm{cleanup}}) \Big) < \infty.
\end{aligned}
\label{eq:liveness-bound}
\end{equation}
\end{theorem}
\begin{proof}
By Proposition~\ref{prop:bounded-expansion}, the number of actions in any synthesized plan DAG is bounded by $M_A < \infty$. Under condition (5), diagnostic operations cannot deadlock on shared physical mutexes or block indefinitely on remote RPCs. Under condition (1), every launched operation either finishes or times out in at most $\hat{\ell}_{\max}$. If an operation fails, times out, or refutes a predicate, condition (2) guarantees that cancellation halts all in-flight workers within $\delta_{\mathrm{cancel}}$ and controller cleanup finishes within $\delta_{\mathrm{cleanup}}$. Under condition (3), each scheduling pass completes in at most $T_{\mathrm{ctrl}}$ wall-clock time. Because condition (4) limits the number of repair attempts to $N_{\mathrm{retry}}^{\max}$, the controller can execute at most $N_{\mathrm{retry}}^{\max} + 1$ planning phases (1 initial synthesis plus at most $N_{\mathrm{retry}}^{\max}$ repairs) and at most $N_{\mathrm{retry}}^{\max}$ cancellation/cleanup cycles before aborting. Finally, condition (6) enforces an absolute hard stop when wall-clock time reaches $T_{\mathrm{dead}}$, with in-flight worker cancellation and cleanup completing within $\delta_{\mathrm{cancel}} + \delta_{\mathrm{cleanup}}$ after the deadline. Hence, the system cannot livelock, deadlock, or stall indefinitely, guaranteeing deterministic termination within bound~\eqref{eq:liveness-bound}.
\end{proof}

\subsection{Bounded Diagnostic Capabilities}
\label{sec:diagnostic-caps}

In high-throughput environments, deep recursive admission introduces unacceptable latency overhead. \aas supports \emph{Bounded Diagnostic Capabilities}: pre-authorized, ephemeral execution tokens with strictly constrained blast radiuses (e.g., read-only filesystem snapshots, dedicated diagnostic tenant namespaces, or sandboxed eBPF probes).

When an operation executes using a certified diagnostic capability, its effective operational risk is pre-attenuated below the policy threshold: $\rho_{\mathrm{eff}}(a) \prec \tau_\Pi$. This places $a$ immediately in the base stratum $\mathcal{R}_{\mathrm{base}}$, collapsing the ADG depth to 1 and eliminating recursive overhead.

%% file: sections/07-adaptive-replanning.tex
\section{Adaptive Replanning Under Partial Information}
\label{sec:adaptive-replanning}

Static workflow planners assume that the environment remains stationary while an evidence plan executes. In reality, evidence acquisition is an active discovery process: \textbf{acquiring evidence updates the controller's knowledge of the world}, which can fundamentally alter both the perceived risk of the action and the obligations required to admit it.

\begin{figure}[t]
\centering
\begin{tikzpicture}[
  box/.style={draw=navy,rounded corners=2pt,fill=softgray,align=center,font=\scriptsize,inner sep=3pt,text width=5.6cm},
  action/.style={draw=cobalt,rounded corners=2pt,fill=softgray!80,font=\scriptsize,align=center,inner sep=3pt,text width=5.6cm},
  arr/.style={-{Latex[length=2mm]},thick,darkslate},
  replan/.style={-{Latex[length=2mm]},thick,crimson,dashed}]

  \node[box] (cac0) {\textbf{Initial CAC Evaluation}\\\scriptsize $\rho_0 = \rho(q, s_0) \implies \Omega_0 = F_\Pi(\rho_0, q, s_0)$};
  \node[box,below=0.45cm of cac0] (sched) {\textbf{\aas Dynamic Scheduler}\\\scriptsize Synthesize \& dispatch plan $\pi_0$};
  \node[action,below=0.45cm of sched] (exec) {\textbf{Execute Assurance Op $a_1$}\\\scriptsize Returns receipt $e_1$};
  \node[box,below=0.45cm of exec] (state) {\textbf{State Update \& Risk Shift}\\\scriptsize $s_1 = \operatorname{Update}(s_0, e_1) \implies \rho_1 = \rho(q, s_1)$};
  \node[box,below=0.45cm of state] (recalc) {\textbf{CAC Obligation Escalation}\\\scriptsize $\Omega_1 = F_\Pi(\rho_1, q, s_1)$ with $\Delta \Omega^+ = \Omega_1 \setminus \Omega_0$};

  \draw[arr] (cac0) -- (sched);
  \draw[arr] (sched) -- (exec);
  \draw[arr] (exec) -- (state);
  \draw[arr] (state) -- (recalc);
  \draw[replan] (recalc.west) -- ++(-0.35,0) |- node[near start,left,font=\tiny,text=crimson,align=right] {Dynamic Replan:\\Splice $\Delta \pi$} (sched.west);
\end{tikzpicture}
\caption{The closed-loop adaptive assurance cycle. Incoming evidence updates system state $s$, triggering risk escalation $\rho_0 \to \rho_1$ and online schedule replanning.}
\label{fig:adaptive-loop}
\end{figure}
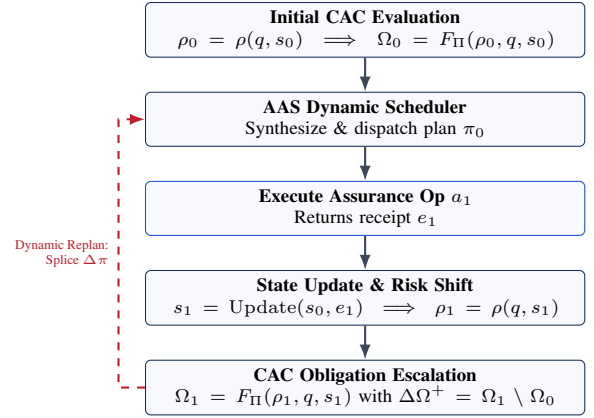

\subsection{Risk Dynamics and Obligation Escalation}
\label{sec:risk-dynamics}

Consider the execution timeline illustrated in Figure~\ref{fig:adaptive-loop}:
\begin{enumerate}
    \item At $t_0$, the controller observes state snapshot $s_0$ and evaluates baseline risk $\rho_0 = \rho(q, s_0)$. The policy emits obligations $\Omega_0 = F_\Pi(\rho_0, q, s_0)$.
    \item The scheduler constructs plan $\pi_0$ and dispatches an initial operation $a_1$ (e.g., query replica replication stream).
    \item Operation $a_1$ completes, returning evidence receipt $e_1$. The payload of $e_1$ reveals that the secondary replica has accumulated 45\,GB of unapplied Write-Ahead Logs (WAL) and disk I/O throughput is severely degraded.
    \item The world state updates to $s_1 = \operatorname{Update}(s_0, e_1)$. Under this degraded condition, observational uncertainty and adverse consequence plausibility increase, shifting the risk profile:
    \begin{equation}
    \rho_1 = \rho(q, s_1) \quad \text{with} \quad \rho_1 \not\preceq \rho_0.
    \label{eq:risk-escalation}
    \end{equation}
    \item The admission engine re-evaluates obligations under updated risk $\rho_1$:
    \begin{equation}
    \Omega_1 = F_\Pi(\rho_1, q, s_1).
    \label{eq:obligation-escalation}
    \end{equation}
    Because risk increased, $\Omega_1$ contains new mandatory obligations:
    \begin{equation}
    \Delta \Omega^+ = \Omega_1 \setminus \Omega_0 \neq \emptyset,
    \end{equation}
    such as requiring an attestation of secondary memory headroom or escalating to human dual-control.
\end{enumerate}

If the scheduler executed $\pi_0$ blind to environment updates, it would arrive at the gateway with an evidence manifest satisfying $\Omega_0$ but lacking $\Delta \Omega^+$, resulting in an immediate rejection and wasted work.

\subsection{Authoritative Runtime State Machine \texorpdfstring{$\Sigma$}{Sigma}}
\label{sec:runtime-state-machine}

To ensure sound execution across asynchronous events, \aas defines a single authoritative runtime state $\Sigma$:
\begin{equation}
\begin{aligned}
\Sigma = \big\langle &q, s, \rho, \Omega, g, \pi, B_{\mathrm{rem}}, B_{\mathrm{commit}}, B_{\mathrm{resv}}, \\
&\mathcal{J}_{\mathrm{run}}, V_{\mathrm{done}}, E_{\mathrm{store}}, \mathcal{W}_q, t_{\mathrm{dispatch}}, T_{\mathrm{dead}} \big\rangle,
\end{aligned}
\label{eq:runtime-state-tuple}
\end{equation}
where:
\begin{itemize}
    \item $q$ is the active action proposal,
    \item $s$ is the modeled environment state (with version vector $\nu_s$),
    \item $\rho = \rho(q, s)$ is current evaluated operational risk,
    \item $\Omega = F_\Pi(\rho, q, s)$ is the active mandatory obligation set,
    \item $g \in \mathbb{N}$ is the monotonically increasing plan generation tag,
    \item $\pi = \langle V_\pi, E_\pi, \sigma_\pi, \mathcal{M}_\pi \rangle$ is the active plan DAG stamped with generation $g$,
    \item $B_{\mathrm{rem}}, B_{\mathrm{commit}}, B_{\mathrm{resv}}$ partition total resource budgets,
    \item $\mathcal{J}_{\mathrm{run}}$ is the set of currently executing worker jobs $\langle a, g_a, t_{\mathrm{launch}}, \hat{\ell}(a) \rangle$,
    \item $V_{\mathrm{done}}$ is the set of terminated operations,
    \item $E_{\mathrm{store}}$ is the pool of verified evidence receipts,
    \item $\mathcal{W}_q$ is the candidate witness manifest,
    \item $t_{\mathrm{dispatch}}$ is the prospective gateway dispatch target timestamp,
    \item $T_{\mathrm{dead}}$ is the proposal hard wall-clock deadline.
\end{itemize}

\subsection{Ten-State Operation Lifecycle Across Plan Generations}
\label{sec:operation-lifecycle}

Every scheduled operation $a \in \mathcal{A}$ transitions through an explicit 10-state lifecycle:
\begin{equation}
\begin{aligned}
\texttt{Reserved} &\longrightarrow \texttt{Ready} \longrightarrow \texttt{Running} \\
&\longrightarrow \{\texttt{Completed}, \texttt{Failed}, \texttt{TimedOut}\} \\
&\longrightarrow \{\texttt{CancelReq} \to \texttt{CancelConf}\} \\
&\longrightarrow \{\texttt{Superseded}, \texttt{Reconciled}\}
\end{aligned}
\label{eq:op-lifecycle}
\end{equation}
\begin{enumerate}
    \item \texttt{Reserved}: Budget is set aside ($B_{\mathrm{rem}} \to B_{\mathrm{resv}}$); operational prerequisites are executing.
    \item \texttt{Ready}: Prerequisites in $D_{\mathrm{op}}$ have terminated; awaiting scheduled start time $\sigma_\pi(a)$.
    \item \texttt{Running}: Launched into an asynchronous worker thread; tagged with current generation $g$; financial cost committed ($B_{\mathrm{resv}} \to B_{\mathrm{commit}}$) exactly once.
    \item \texttt{Completed}: Finished execution within timeout $\hat{\ell}(a)$; emitted receipts ingested into $E_{\mathrm{store}}$.
    \item \texttt{Failed}: Threw an unhandled software exception or network connection failure.
    \item \texttt{TimedOut}: Wall-clock elapsed time exceeded $\hat{\ell}(a)$ without completion; worker is sent an abort signal.
    \item \texttt{CancelRequested}: Cancellation signal dispatched due to plan replacement or affirmative refutation.
    \item \texttt{CancelConfirmed}: Remote worker confirmed thread termination and resource cleanup.
    \item \texttt{Superseded}: The active plan generation advanced ($g' > g$) before or during execution; the operation's witness role in $\mathcal{W}_q$ is revoked, but its physical state side-effects are preserved.
    \item \texttt{Reconciled}: Terminal accounting complete; any genuine unspent reservations returned to $B_{\mathrm{rem}}$.
\end{enumerate}

\paragraph{Handling Seven Asynchronous Runtime Scenarios.}
The controller state machine resolves race conditions via strict semantic rules:
\begin{enumerate}
    \item \textbf{Diagnostic Completing Post-Supersession:} If an operation launched in generation $g$ completes when the active generation is $g' > g$, its receipts are \textbf{not} automatically assigned to the new manifest $\mathcal{W}_q$. However, any physical telemetry or state mutation in its payload is incorporated into environment state $s$.
    \item \textbf{Multi-Receipt Operations:} An operation producing multiple receipts emits them all into $E_{\mathrm{store}}$, but financial cost $c(a)$ is charged strictly once upon launch into $B_{\mathrm{commit}}$.
    \item \textbf{Cancellation Racing with Completion:} If completion arrives before cancellation confirmation, the operation is marked \texttt{Completed} and receipts are ingested. If cancellation confirms first, the operation enters \texttt{CancelConfirmed}, receipts are discarded, and unspent reservations are refunded.
    \item \textbf{Timed-Out Operations Continuing Remotely:} An operation marked \texttt{TimedOut} cannot resurrect its witness role if late output arrives; its results are treated as discardable.
    \item \textbf{Non-Rollback Mutations:} If a consequential diagnostic mutates external state and subsequent admission fails, the state change is retained in $s$, and compensation actions are scheduled via \tct.
    \item \textbf{Earlier-Generation Evidence:} Receipts acquired in earlier generations remain in $E_{\mathrm{store}}$ and are eligible for reuse if they satisfy version match $G(e) = s|_{G(e)}$ and freshness against the repaired dispatch target.
    \item \textbf{New Obligations at Minting Boundary:} Handled via atomic Compare-And-Swap (CAS) on state version $\nu_s$; if state changed while minting, the certificate is discarded and replanning is re-entered.
\end{enumerate}

\subsection{Atomic Incremental Plan Repair Transaction}
\label{sec:incremental-repair}

Algorithm~\ref{alg:adaptive-repair} formalizes how state $\Sigma$ transitions upon an asynchronous event without corrupting plan generation invariants or masking real-world state mutations. To maintain physical and algorithmic integrity, \aas strictly distinguishes:
\begin{enumerate}
    \item \textbf{Irreversible Physical Effects:} Actual mutations in external systems, hardware state, or non-refundable financial/token expenditure. These effects are authoritatively reflected in $\Sigma.s$ and $\Sigma.B_{\mathrm{commit}}$ and are \textbf{never rolled back} merely because replacement plan synthesis fails.
    \item \textbf{Authoritative Environment State ($\Sigma.s$):} The true, known state of the physical world.
    \item \textbf{Tentative Repair State:} Shadow evaluations of risk $\rho_{\mathrm{cand}}$ and active obligations $\Omega_{\mathrm{live}}$.
    \item \textbf{Candidate Schedule ($\Delta \pi, t'_{\mathrm{dispatch}}$):} The prospective replacement sub-plan.
    \item \textbf{Candidate Reservations ($B_{\mathrm{resv}}^{\mathrm{cand}}, B_{\mathrm{rem}}^{\mathrm{cand}}$):} Tentative resource ledger allocations.
\end{enumerate}
Plan repair executes as a structured \textbf{PREPARE--VALIDATE--COMMIT} transaction:

\begin{algorithm}[t]
\caption{Atomic Adaptive Plan Repair Transaction}
\label{alg:adaptive-repair}
\begin{algorithmic}[1]
\small
\Require Runtime state $\Sigma$, triggering event $\mathcal{E}$, wall-clock time $t_{\mathrm{now}}$
\Ensure Updated state $\Sigma$ or $\textsc{Abort}$
\State \Comment{\textbf{Phase I: Authoritative Physical State Reconciliation}}
\If{$\mathcal{E}$ carries irreversible physical mutations}
  \State $\Sigma.s \gets \operatorname{ApplyPhysicalEffects}(\Sigma.s, \mathcal{E})$ \Comment{Authoritative physical state updated}
\EndIf
\If{$\mathcal{E}$ is Affirmative Refutation}
  \State $\operatorname{CancelAll}(\Sigma.\mathcal{J}_{\mathrm{run}})$; refund unspent $B_{\mathrm{resv}} \to B_{\mathrm{rem}}$
  \State \Return $\textsc{Abort}(\text{Predicate Refuted})$
\EndIf
\State \Comment{\textbf{Phase II: PREPARE (Shadow State \& Tentative Ledgers)}}
\State $\rho_{\mathrm{cand}} \gets \rho(\Sigma.q, \Sigma.s)$; $\Omega_{\mathrm{live}} \gets F_\Pi(\rho_{\mathrm{cand}}, \Sigma.q, \Sigma.s)$
\State Let $\Delta \Omega^+ \gets \Omega_{\mathrm{live}} \setminus \Sigma.\Omega$ and $\Delta \Omega^- \gets \Sigma.\Omega \setminus \Omega_{\mathrm{live}}$
\State $V_{\mathrm{cand}} \gets \Sigma.V_\pi$; $B_{\mathrm{rem}}^{\mathrm{cand}} \gets \Sigma.B_{\mathrm{rem}}$; $B_{\mathrm{resv}}^{\mathrm{cand}} \gets \Sigma.B_{\mathrm{resv}}$
\For{each unlaunched $u \in \Sigma.V_\pi$ witnessing only $\Delta \Omega^-$}
  \State $V_{\mathrm{cand}} \gets V_{\mathrm{cand}} \setminus \{u\}$; $B_{\mathrm{rem}}^{\mathrm{cand}} \gets B_{\mathrm{rem}}^{\mathrm{cand}} + c(u)$; $B_{\mathrm{resv}}^{\mathrm{cand}} \gets B_{\mathrm{resv}}^{\mathrm{cand}} - c(u)$
\EndFor
\State $E_{\mathrm{reusable}} \gets \{e \in \Sigma.E_{\mathrm{store}} \mid G(e) = \Sigma.s|_{G(e)} \land \neg e.\text{refuted}\}$
\State Determine $\Omega_{\mathrm{unresolved}} \subseteq \Omega_{\mathrm{live}}$ lacking quorums/cuts under $E_{\mathrm{reusable}}$
\If{$\Omega_{\mathrm{unresolved}} = \emptyset$}
  \State $t_{\mathrm{disp}}' \gets \max(t_{\mathrm{now}}, \; \Sigma.t_{\mathrm{dispatch}})$
  \If{$t_{\mathrm{disp}}' + \delta_{\mathrm{exec}} \le \min_{e \in E_{\mathrm{reusable}}} (t_{\mathrm{obs}}(e) + \Delta t_{\mathrm{eff}}(e) - \epsilon_{\mathrm{skew}}) \land t_{\mathrm{disp}}' + \delta_{\mathrm{exec}} \le \Sigma.T_{\mathrm{dead}}$}
    \State $\Sigma.\rho \gets \rho_{\mathrm{cand}}$; $\Sigma.\Omega \gets \Omega_{\mathrm{live}}$; $\Sigma.V_\pi \gets V_{\mathrm{cand}}$
    \State $\Sigma.B_{\mathrm{rem}} \gets B_{\mathrm{rem}}^{\mathrm{cand}}$; $\Sigma.B_{\mathrm{resv}} \gets B_{\mathrm{resv}}^{\mathrm{cand}}$; $\Sigma.t_{\mathrm{dispatch}} \gets t_{\mathrm{disp}}'$; $\Sigma.g \gets \Sigma.g + 1$; \Return $\Sigma$
  \Else
    \State Mark expired obligations in $\Omega_{\mathrm{live}}$ as unresolved $\to \Omega_{\mathrm{unresolved}}$
  \EndIf
\EndIf
\State \Comment{\textbf{Phase III: SYNTHESIZE \& VALIDATE (Transaction Guard)}}
\State $(\Delta \pi, t_{\mathrm{disp}}') \gets \textsc{SynthPlan}(\Omega_{\mathrm{unresolved}}, B_{\mathrm{cand}}, \Sigma.T_{\mathrm{dead}})$
\If{$\Delta \pi = \textsc{Infeasible} \lor t_{\mathrm{disp}}' + \delta_{\mathrm{exec}} > \Sigma.T_{\mathrm{dead}} \lor \operatorname{Cost}(\Delta \pi) > B_{\mathrm{cand}}$}
  \State \Comment{\textbf{ABORT:} Retain authoritative physical state $\Sigma.s$; cancel unlaunched candidate}
  \State $\operatorname{CancelAll}(\Sigma.\mathcal{J}_{\mathrm{run}})$; $\Sigma.B_{\mathrm{rem}} \gets B_{\mathrm{rem}}^{\mathrm{cand}} + B_{\mathrm{resv}}^{\mathrm{cand}}$; $\Sigma.B_{\mathrm{resv}} \gets \mathbf{0}$
  \State \Return $\textsc{Abort}(\text{Infeasible within Budget/Deadline})$
\EndIf
\State \Comment{\textbf{Phase IV: COMMIT (Publish Schedule \& Advance Generation)}}
\State $\Sigma.g \gets \Sigma.g + 1$ \Comment{Advance plan generation monotonically}
\State $\Sigma.\rho \gets \rho_{\mathrm{cand}}$; $\Sigma.\Omega \gets \Omega_{\mathrm{live}}$
\State $\Sigma.B_{\mathrm{resv}} \gets B_{\mathrm{resv}}^{\mathrm{cand}} + \operatorname{Cost}(\Delta \pi)$; $\Sigma.B_{\mathrm{rem}} \gets B_{\mathrm{rem}}^{\mathrm{cand}} - \operatorname{Cost}(\Delta \pi)$
\State $\Sigma.V_\pi \gets V_{\mathrm{cand}} \cup \Delta \pi.V_\pi$; $\Sigma.\pi \gets \operatorname{SpliceDAG}(\Sigma.\pi, \Delta \pi, t_{\mathrm{now}})$
\State $\Sigma.t_{\mathrm{dispatch}} \gets t_{\mathrm{dispatch}}'$; \Return $\Sigma$
\end{algorithmic}
\end{algorithm}

\paragraph{Semantic Receipt Reuse and Atomicity Guarantees.}
\aas permits reusing previously acquired receipts in $E_{\mathrm{store}}$ across plan repair if and only if three conditions hold simultaneously:
\begin{enumerate}
    \item \textbf{Repaired Dispatch Freshness:} The receipt's expiration bound extends beyond the prospective \emph{repaired} dispatch and execution window ($t_{\mathrm{dispatch}}' + \delta_{\mathrm{exec}} \le t_{\mathrm{obs}}(e) + \Delta t_{\mathrm{eff}}(e) - \epsilon_{\mathrm{skew}}$), rather than merely the current timestamp $t_{\mathrm{now}}$.
    \item \textbf{State Invariance:} The underlying state versions guarded by the receipt ($G(e)$) match live state: $G(e) = \Sigma.s|_{G(e)}$.
    \item \textbf{Epistemic Independence:} The receipt's exposure set $\mathcal{X}(e)$ satisfies the EFD diversity cuts required by current live obligations $\Sigma.\Omega$.
\end{enumerate}
Algorithm~\ref{alg:adaptive-repair} specifies a validation guard before generation advancement. The current prototype cancels running jobs and reconciles reservations before solving the replacement problem; a failed repair can therefore leave the prior execution canceled. It enforces budget conservation but does not implement full transactional rollback of external jobs.

The prototype removes receipts with invalid versions or insufficient remaining lifetime, represents potentially reusable receipts as zero-cost observation proxies, and checks their original observation timestamps against the proposed dispatch. The joint solver checks diversity with any new witnesses. In the matched 20-trial selected-operation timeout study, receipt-aware repair and full resynthesis each admit 18 cases, while the static control admits none; reuse lowers mean committed cost (Section~\ref{sec:eval-repair}).

%% file: sections/08-reference-scheduler.tex
\section{Reference Scheduler Architecture and Algorithms}
\label{sec:reference-scheduler}

The \aas reference scheduler separates plan synthesis, evidence execution, admission, and repair into four components.

\subsection{Component Architecture}
\label{sec:scheduler-arch}

\begin{figure}[t]
\centering
\begin{tikzpicture}[
  comp/.style={draw=navy,rounded corners=2pt,fill=softgray,align=center,font=\scriptsize,inner sep=3.5pt,minimum width=3.3cm},
  store/.style={draw=darkslate,cylinder,shape border rotate=90,aspect=0.25,fill=softgray!60,font=\scriptsize,align=center,inner sep=2pt},
  arr/.style={-{Latex[length=2mm]},thick,darkslate},
  darr/.style={-{Latex[length=2mm]},thick,dashed,cobalt}]

  \node[comp] (compiler) {\textbf{Obligation Compiler}\\\tiny Parse $\Omega(q,s)$ \& EFD requirements};
  \node[comp,right=0.8cm of compiler] (synthesizer) {\textbf{Plan Synthesizer}\\\tiny Joint EFD \& Temporal Solver};
  \node[store,below=0.6cm of compiler] (catalog) {\textbf{Operation Catalogue}\\\tiny Signatures, latencies, $\mathcal{X}(a)$};
  \node[comp,below=0.6cm of synthesizer] (dispatcher) {\textbf{Execution Dispatcher}\\\tiny Async worker pool \& timeouts};
  \node[comp,below=0.6cm of dispatcher] (replanner) {\textbf{Replanning Monitor}\\\tiny Risk tracking \& dynamic repair};
  \node[comp,below=0.6cm of catalog] (mintif) {\textbf{Admission Mint Interface}\\\tiny Freshness check \& manifest $\mathcal{W}_q$};

  \draw[arr] (compiler) -- node[above,font=\tiny] {Parsed $\Omega$} (synthesizer);
  \draw[arr] (catalog) -- (synthesizer);
  \draw[arr] (synthesizer) -- node[right,font=\tiny] {Plan $\pi^*$} (dispatcher);
  \draw[arr] (dispatcher) -- node[right,font=\tiny] {Receipts} (replanner);
  \draw[darr] (replanner.east) -- ++(0.25,0) |- (synthesizer.east);
  \draw[arr] (replanner) -- node[above,font=\tiny] {Valid $E$} (mintif);
\end{tikzpicture}
\caption{The \aas reference scheduler. The synthesizer selects and times evidence operations, the dispatcher handles execution and consequential sub-admission, and the monitor triggers repair after modeled state changes. Optimality requires the exact-search conditions of Theorem~\ref{thm:benders-optimality}.}
\label{fig:scheduler-arch}
\end{figure}
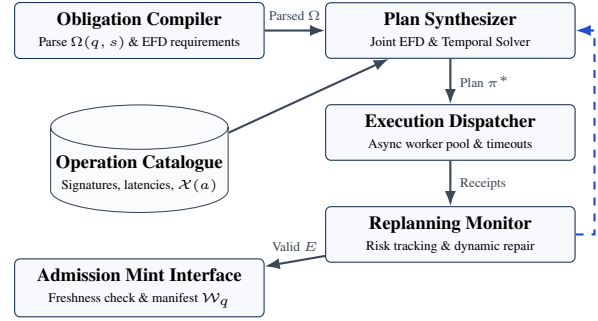

The \aas architecture comprises six decoupled subsystems (Figure~\ref{fig:scheduler-arch}):
\begin{enumerate}
    \item \textbf{Obligation Compiler:} Ingests $\Omega(q,s)$ from the \cac policy resolver, extracts predicate requirements, epistemic diversity thresholds $\kappa_E^{\min}(\omega)$, and temporal lifespans $\Delta t_\omega$.
    \item \textbf{Operation Catalogue \& Exposure Index:} Indexes registered operations $\mathcal{A}$ by covered claims, conservative latencies $\hat{\ell}(a)$, costs $c(a)$, token consumption $\operatorname{tok}(a)$, and exposure sets $\mathcal{X}(a) \subseteq \mathbb{F}$.
    \item \textbf{Plan Synthesizer:} Implements the joint optimization engine. It selects an admissible subset of operations $V_\pi \subseteq \mathcal{A}$ via ILP (Section~\ref{sec:joint-cost-efd}) or progress-based greedy fallback (Section~\ref{sec:greedy-efd}) and synthesizes execution start times $\sigma_\pi$ via Backward JIT scheduling (Algorithm~\ref{alg:backward-scheduler}).
    \item \textbf{Execution Dispatcher:} Manages an asynchronous worker pool enforcing strict per-operation timeouts $\hat{\ell}(a)$. Crucially, any consequential diagnostic operation is intercepted and subjected to recursive \cac admission before dispatch.
    \item \textbf{Replanning Monitor:} Listens to incoming receipts, detects whether new observations alter world state $s$, recalculates risk $\rho(q,s)$, and triggers incremental plan repair (Algorithm~\ref{alg:adaptive-repair}) upon risk escalation, timeout, or negative evidence.
    \item \textbf{Admission Mint Interface:} Verifies composite witness manifests $\mathcal{W}_q = \{(\omega, W_\omega)\}$ and forwards them to the authoritative \cac engine for formal discharge evaluation and certificate minting.
\end{enumerate}

\subsection{Complete Scheduling Loop}
\label{sec:complete-loop}

Algorithm~\ref{alg:complete-aas} formalizes the complete end-to-end scheduling and execution loop operating on authoritative state $\Sigma$.

\begin{algorithm}[!t]
\caption{Complete \aas Scheduling and Execution Loop}
\label{alg:complete-aas}
\begin{algorithmic}[1]
\scriptsize
\linespread{0.82}\selectfont
\Require Proposal $q$, initial state $s_0$, context $\chi$, deadline $T_{\mathrm{dead}}$, budget $B$, constraints $\mathcal{D}, \mathcal{T}$
\Ensure Admission certificate $\mathcal{C}_q$ or $\textsc{Refuse}$
\State Initialize $\Sigma.s \gets s_0$; $\Sigma.\rho \gets \rho(q, s_0)$; $\Sigma.\Omega \gets F_{\Pi}(\Sigma.\rho, q, s_0)$; $\Sigma.g \gets 0$
\State $\Sigma.B_{\mathrm{commit}} \gets \mathbf{0}$; $\Sigma.B_{\mathrm{resv}} \gets \mathbf{0}$; $\Sigma.B_{\mathrm{rem}} \gets B$; $\Sigma.V_{\mathrm{done}}, \Sigma.E_{\mathrm{store}}, \Sigma.E_{\mathrm{seen}}, \Sigma.\mathcal{J}_{\mathrm{run}} \gets \emptyset$
\State $\Sigma.\pi \gets \operatorname{SynthesizePlan}(\Sigma.\Omega, \mathcal{A}, \Sigma.B_{\mathrm{rem}}, \mathcal{D}, \mathcal{T})$
\If{$\Sigma.\pi = \textsc{Infeasible}$} \Return $\textsc{Refuse}(\text{No plan within budget/deadline})$ \EndIf
\State $\Sigma.B_{\mathrm{resv}} \gets \operatorname{Cost}(\Sigma.\pi)$; $\Sigma.B_{\mathrm{rem}} \gets \Sigma.B_{\mathrm{rem}} - \operatorname{Cost}(\Sigma.\pi)$; $\Sigma.t_{\mathrm{dispatch}} \gets \operatorname{TargetDispatch}(\Sigma.\pi)$
\While{$t_{\mathrm{now}} < T_{\mathrm{dead}}$}
  \For{each ready operation $a \in \Sigma.V_\pi$ with $\sigma_\pi(a) \le t_{\mathrm{now}}$ in state \texttt{Ready}}
    \If{$\rho(a) \succeq \tau_\Pi$} \Comment{Consequential sub-admission}
      \State $\mathcal{C}_a \gets \operatorname{RecursiveCAC}(a, \Sigma.s, \Sigma.B_{\mathrm{rem}})$
      \If{$\mathcal{C}_a = \textsc{Refuse}$} \State $\operatorname{CancelAll}(\Sigma.\mathcal{J}_{\mathrm{run}})$; \Return $\textsc{Refuse}(\text{Sub-admission failed})$ \EndIf
    \EndIf
    \State Transition $a \to \texttt{Running}$; move $c(a)$ from $\Sigma.B_{\mathrm{resv}} \to \Sigma.B_{\mathrm{commit}}$ \Comment{Charged once at launch}
    \State Launch $a$ into $\Sigma.\mathcal{J}_{\mathrm{run}}$ tagged with generation $\Sigma.g$ and timeout $\hat{\ell}(a)$
  \EndFor
  \State Await next asynchronous event $\mathcal{E}$ or timeout until earliest scheduled launch
  \If{Wall-clock timeout with $t_{\mathrm{now}} \ge T_{\mathrm{dead}}$}
    \State $\operatorname{CancelAll}(\Sigma.\mathcal{J}_{\mathrm{run}})$; refund unspent $B_{\mathrm{resv}} \to B_{\mathrm{rem}}$; \Return $\textsc{Refuse}(\text{Deadline exceeded})$
  \EndIf
  \If{$\mathcal{E}.\mathrm{id} \in \Sigma.E_{\mathrm{seen}}$} \State \textbf{continue} \Comment{Discard duplicate event idempotently} \EndIf
  \State $\Sigma.E_{\mathrm{seen}} \gets \Sigma.E_{\mathrm{seen}} \cup \{\mathcal{E}.\mathrm{id}\}$
  \If{$\mathcal{E}$ is an operation termination event (job $a$)}
    \State $\Sigma.\mathcal{J}_{\mathrm{run}} \gets \Sigma.\mathcal{J}_{\mathrm{run}} \setminus \{\mathcal{E}.\text{job}\}$; $\Sigma.V_{\mathrm{done}} \gets \Sigma.V_{\mathrm{done}} \cup \{a\}$
    \If{$\mathcal{E}$ carries stale generation $g_{\mathcal{E}} < \Sigma.g$}
      \State Transition $a \to \texttt{Superseded}$; reconcile physical effects via $\operatorname{AdaptiveRepair}(\Sigma, \mathcal{E}, t_{\mathrm{now}})$
      \State \textbf{continue}
    \EndIf
    \State Transition $a \to \texttt{Completed}$ (or \texttt{Failed}/\texttt{TimedOut}); refund unused reservations
  \EndIf
  \If{$\mathcal{E}$ is Affirmative Refutation}
    \State $\operatorname{CancelAll}(\Sigma.\mathcal{J}_{\mathrm{run}})$; \Return $\textsc{Refuse}(\text{Predicate refuted})$
  \ElsIf{$\mathcal{E}$ returns valid receipt $e$ from operation $a$}
    \State Ingest $e$ into $\Sigma.E_{\mathrm{store}}$ only if $a$ was verified (including $\mathcal{C}_a$ if consequential)
  \EndIf
  \If{$\mathcal{E}$ alters state $s$ or is operation fault/timeout}
    \State $\Sigma \gets \operatorname{AdaptiveRepair}(\Sigma, \mathcal{E}, t_{\mathrm{now}})$
    \If{$\Sigma = \textsc{Abort}$} \Return $\textsc{Refuse}(\text{Repair failed})$ \Comment{Retain physical state $\Sigma.s$; refuse safely} \EndIf
  \EndIf
  \State $\Sigma.\mathcal{W}_q \gets \operatorname{ExtractWitnesses}(\Sigma.\Omega, \Sigma.E_{\mathrm{store}})$
  \If{$\Sigma.\mathcal{W}_q$ satisfies quorums and EFD cuts for all $\Sigma.\Omega$}
    \If{Readiness~\eqref{eq:earliest-legal-dispatch} and execution-window freshness~\eqref{eq:execution-window-validity} hold for $\Sigma.\mathcal{W}_q$ at $t_{\mathrm{now}}$}
      \State $s_{\mathrm{snap}} \gets \Sigma.s$; $\text{verdict} \gets \operatorname{CAC}.\operatorname{Discharge}(\Sigma.\Omega, \Sigma.\mathcal{W}_q, s_{\mathrm{snap}}, \chi)$
      \If{$\text{verdict} = \textsc{Satisfied}$}
        \State $\mathcal{C}_q \gets \operatorname{CAC}.\operatorname{MintCert}(q, s_{\mathrm{snap}}, \Sigma.\mathcal{W}_q)$
        \If{$\Sigma.s.\nu \neq s_{\mathrm{snap}}.\nu$} \Comment{CAS check: concurrent state invalidation}
          \State Discard $\mathcal{C}_q$; $\Sigma \gets \operatorname{AdaptiveRepair}(\Sigma, \text{VersionConflict}, t_{\mathrm{now}})$
        \Else
          \State \Return $\mathcal{C}_q$ \Comment{Authoritative admission succeeded}
        \EndIf
      \EndIf
    \EndIf
  \EndIf
\EndWhile
\State $\operatorname{CancelAll}(\Sigma.\mathcal{J}_{\mathrm{run}})$; refund unspent $B_{\mathrm{resv}} \to B_{\mathrm{rem}}$; \Return $\textsc{Refuse}(\text{Deadline exceeded})$
\end{algorithmic}
\end{algorithm}

\subsection{Computational Complexity and Scalability Profile}
\label{sec:complexity}

While arbitrary epistemic cut minimization is NP-hard (Theorem~\ref{thm:np-hard}), our generated instances have at most 16 obligations and 48 operations. On the reference benchmark platform (Apple Silicon, 12-core ARM64, 36\,GB RAM, macOS/Darwin, Python 3.13, HiGHS branch-and-cut solver):
\begin{itemize}
    \item In the corrected feasible scaling sweep, end-to-end solver time averages 0.58\,ms at four operations and 10.66\,ms at 48 operations; neither value isolates Master ILP time.
    \item The greedy fallback (Equation~\eqref{eq:greedy-efd-heuristic}) enumerates fault coalitions up to quorum size for its provisional score. It remains polynomial for fixed small quorums but provides no feasibility certificate when it stops without a plan.
    \item Two-tier Backward JIT scheduling is measured as part of the end-to-end solver time above; the sweep does not separately estimate its latency under bounded concurrency.
    \item In the corrected 20-trial selected-operation timeout study, all 20 cases trigger repair. Receipt-aware repair and full resynthesis each admit 18; local-process repair computation averages 2.05 and 2.32\,ms, respectively, in the rerun (Section~\ref{sec:eval-repair}).
\end{itemize}
These measured times are small relative to simulated operation latencies. The ordinary size sweep contains no refinement cuts; a separate 20-case constructed suite exercises the supported cut certificates but does not bound worst-case online planning time.

%% file: sections/09-evaluation-design.tex
\section{Empirical Evaluation}
\label{sec:evaluation}
\label{sec:evaluation-design}

We evaluate the \aas reference implementation with generated infrastructure workloads and small constructed counterexamples. The primary comparison applies six policies in two evaluation modes to the same 1,200 instances, yielding 14,400 policy-mode records. A policy-mode record is not an independent deployment. Targeted paired studies test freshness, fault-domain diversity, repair, and decomposition cuts.

\subsection{Workloads and Protocol}
\label{sec:scenarios}

The generator creates 400 instances each for PostgreSQL standby promotion, Kubernetes node remediation, and cloud IAM/network reconfiguration (seeds 1000--2199). They contain 4--12 obligations, operation costs and precedence edges, evidence lifetimes of 3--30\,s, and an exposure map with 24 modeled fault domains. Actual operation latency is sampled from a log-normal distribution centered on nominal latency ($\sigma=0.35$); this is a simulator assumption, not a measured cloud model. Domain faults are injected into 15\% of instances. The selected operation catalogue and gateway policy are held fixed across the main policy comparisons.

\label{sec:baselines}
The six primary policies are serial forward acquisition (B1), cheapest-witness greedy selection (B2), constraint-aware parallel forward scheduling (B3/B5), full \aas (B4), and constraint-aware selection followed by scheduling without decomposition refinement (B6). B3 and B5 share an implementation and are shown separately only to preserve the original policy nomenclature. B2 is deliberately simple and is not an implementation of prior portfolio methods such as VP-CONTROL~\cite{vpcontrol2026}. A bounded exact solver (B7) is used only for the oracle study; full resynthesis (B8) is used in the repair study.

\label{sec:eval-protocol}
In Mode A, the admission gateway scores a candidate manifest retrospectively. In Mode B, it enforces the admission boundary and blocks invalid candidates. We distinguish valid admission, gateway rejection, proven safe refusal, and planning indeterminacy. A recorded invalid candidate in Mode A is a simulated proposal, not an unauthorized physical action. Each main-table outcome count uses the full 1,200-instance denominator. Table latency and cost means are conditional on valid admission; rejected attempts and safe refusals are excluded from those means.

\input{tables/tab1_main_comparison.tex}
\input{tables/tab2_mode_b_enforced.tex}

\subsection{Main Comparison}
\label{sec:main-results}

In Mode A, \aas produces 1,075/1,200 valid candidates (89.58\%), 12 stale candidates (1.00\%), and 113 proven safe refusals (9.42\%). Constraint-aware forward scheduling produces 647 valid candidates (53.92\%), 440 stale candidates (36.67\%), and the same 113 refusals. Thus, on these paired generated instances, backward placement changes 428 outcomes from stale candidate to valid admission. B6 matches B4 on this workload because no refinement cut is generated by its ordinary instances; this comparison alone does not establish a decomposition benefit.

Greedy cost selection and serial forward acquisition each produce 738/1,200 correlated-failure classifications (61.50\%) under the declared exposure map. No B4 candidate is classified as structurally correlated. In Mode B the gateway admits 1,078 \aas candidates, blocks nine stale candidates, and records 113 safe refusals. Three of the 12 Mode A stale cases are recovered by the controller in Mode B; no physical proposal executes on a blocked candidate. The gateway blocks 738 serial-forward and 757 greedy-cost candidate dispatches. Mean wasted committed cost over all 1,200 proposals is \$0.013 for B4 and \$0.106 for B2. These are simulator outcomes under the modeled faults and budgets.

\begin{figure}[t]
\centering
\includegraphics[width=\columnwidth]{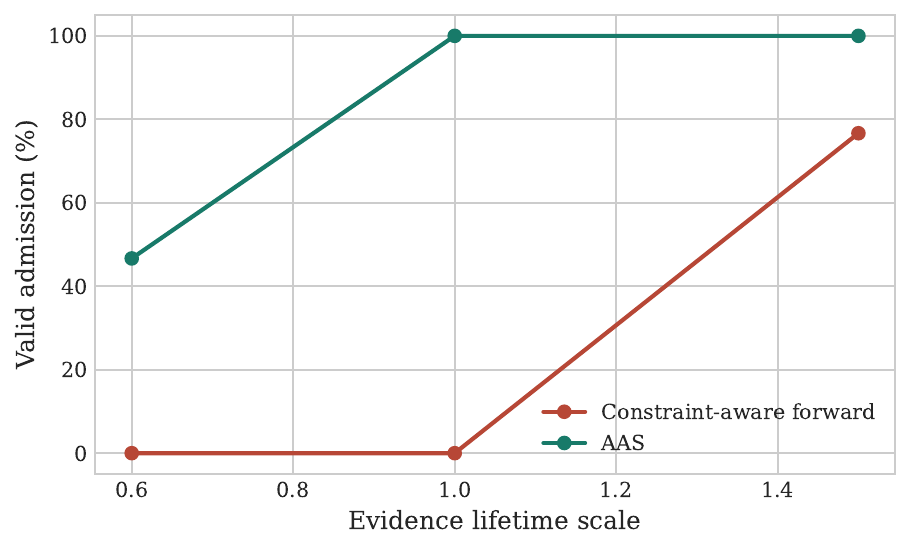}
\caption{RQ1: Valid admission under matched evidence-lifetime scales. Each point uses the same 30 PostgreSQL base instances and operation-keyed latency draws for both policies and all three scales.}
\label{fig:freshness-sweep}
\end{figure}

\subsection{RQ1: Paired Freshness Sensitivity}
\label{sec:eval-freshness}

We vary only evidence lifetimes ($0.6\times$, $1.0\times$, $1.5\times$) on 30 PostgreSQL base instances (seeds 7000--7029). Each operation receives a fixed latency draw keyed by base seed and operation identifier, so the same trace is reused across policies and scales. At $0.6\times$, \aas admits 14/30 while constraint-aware forward scheduling admits 0/30; the other 16 \aas cases are proven infeasible and refused. At nominal lifetime, the counts are 30/30 versus 0/30; at $1.5\times$, they are 30/30 versus 23/30. The paired admission difference at $1.5\times$ is 7/30 (23.3 percentage points; percentile bootstrap 95\% interval 10.0--40.0 points over base instances). The gap narrows as evidence remains valid longer. The separate $K_{\max}=1$ fixture verifies that a two-operation serial schedule is feasible, but does not distinguish policies.

\begin{figure}[t]
\centering
\includegraphics[width=\columnwidth]{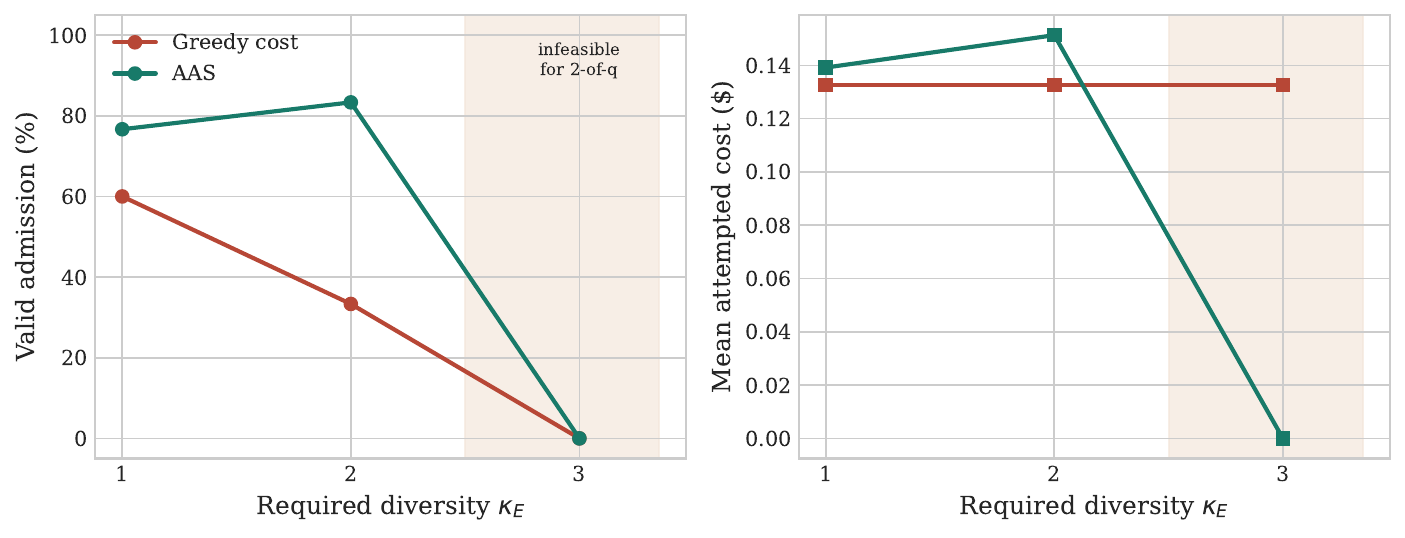}
\caption{RQ2: Valid admission and unconditional attempted cost on 30 matched base instances under fixed 2-of-$q$ quorums. Threshold 3 is structurally infeasible because $\kappa_E\le2$; its zero attempted cost reflects refusal, not a cost improvement.}
\label{fig:efd-frontier}
\end{figure}

\subsection{RQ2: Diversity and Cost}
\label{sec:eval-efd}

We vary required \efd threshold $\kappa_E^{\min}\in\{1,2,3\}$ on the same 30 base instances (ten per scenario; seeds 8000--8029), fixing 2-of-$q$ quorum rules, operation-keyed latency, and injected domain faults. \aas admits 23 and 25 instances at thresholds 1 and 2; the greedy cost baseline admits 18 and 10. The different admission counts reflect both diversity and realized faults, so the threshold-2 arm is not a monotone availability claim. At threshold 3, all 30 \aas instances are proven infeasible because a completed local fault basis implies $\kappa_E\le k_\omega=2$ for the varied obligations. Greedy cost selection produces no valid admission either. The zero attempted cost of \aas in this arm records safe refusal. Cost comparisons among admitted proposals must therefore use the feasible threshold-1 and threshold-2 arms; unconditional cost is \$0.139 and \$0.151 per instance, respectively.

For the declared exposure map and positive authorization rule, $\kappa_E(W,\Gamma_\omega)\ge h$ means that fewer than $h$ modeled roots cannot expose a decisive approving coalition. This is a conditional structural guarantee; unmodeled shared dependencies or an incorrect causal account of exposed approvals can defeat it.

\subsection{RQ3: Bounded Joint Optimality}
\label{sec:eval-optimality}

On 20 bounded instances ($|\Omega|\le4$, $|\mathcal A|\le10$), \aas and the exact oracle both find feasible plans at the same cost. Mean local solve times in the corrected rerun are 2.15 and 5.44\,ms, respectively; this small paired sample does not establish a general speed advantage. The revised solver uses exhaustive temporal-order fallback for portfolios of at most six operations, and the oracle shares that temporal search component. The comparison checks selection and cost on small cases, not independent temporal correctness or general completeness.

\begin{figure}[t]
\centering
\includegraphics[width=\columnwidth]{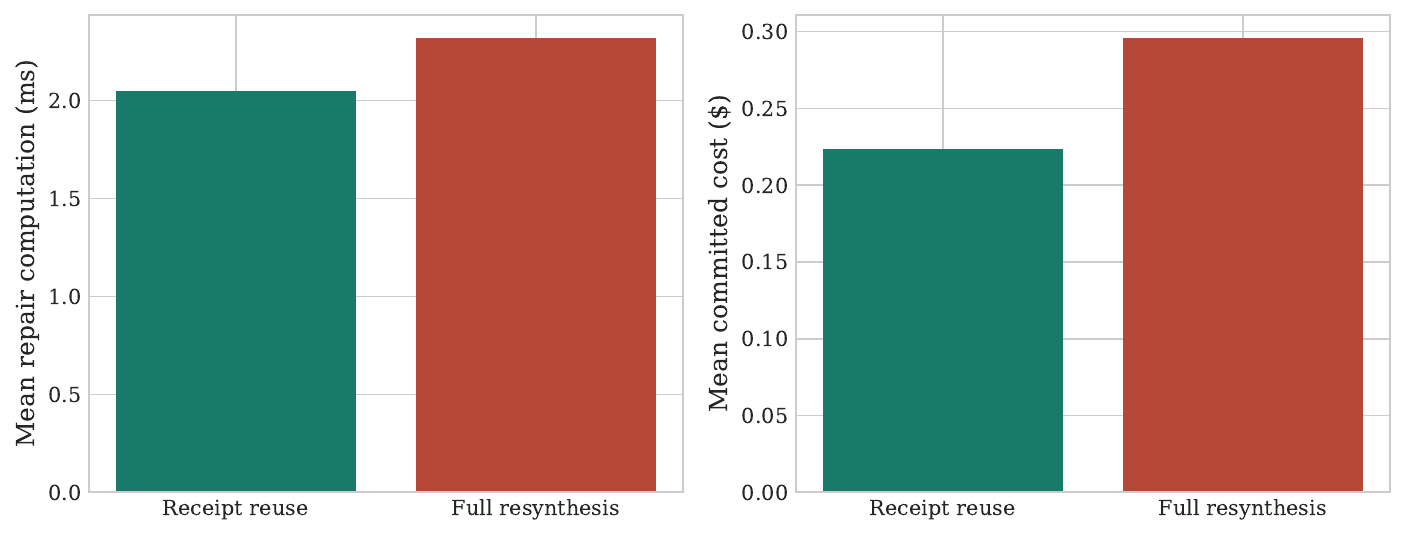}
\caption{RQ4: Matched mean repair computation and committed cost on 20 selected-operation timeout trials. Receipt reuse and resynthesis each admit 18; the static control admits none. Computation times are local-process measurements.}
\label{fig:repair-vs-resynth}
\end{figure}

\subsection{RQ4: Repair After a Selected-Operation Timeout}
\label{sec:eval-repair}

The corrected intervention faults an operation selected by the same initial plan in every arm (20 trials, seeds 5000--5019). It compares incremental repair with receipt reuse, the same controller with reuse disabled, full resynthesis, and a static no-repair control. All 20 cases trigger repair in the three enabled arms. Reuse, no-reuse, and resynthesis each admit 18/20; static planning admits 0/20. The two failures in each repair arm occur on different instances, so equal totals do not imply identical recovery. Receipt reuse records 2.85 reused receipts per case and mean committed cost \$0.223 versus \$0.296 for either no-reuse or resynthesis. The paired cost difference (reuse minus no-reuse) is $-\$0.073$ with a percentile bootstrap 95\% interval of approximately $[-\$0.100,-\$0.047]$ over 20 instances. Mean local repair computation in the corrected rerun is 2.05 versus 2.21\,ms for no-reuse and 2.32\,ms for resynthesis; these small timing differences are process dependent and do not establish a deployment latency advantage. This study injects operation timeouts, not environmental version changes.

\subsection{RQ5: Consequential Diagnostics}
\label{sec:eval-diagnostics}

Twelve deterministic fixtures test prospective recursive admission: ten finite stratified trees are admitted, while a cycle and a non-descending-risk tree are rejected. These fixtures check specific invariants; they do not estimate distributed liveness or resource-leak rates.

\begin{figure}[t]
\centering
\includegraphics[width=\columnwidth]{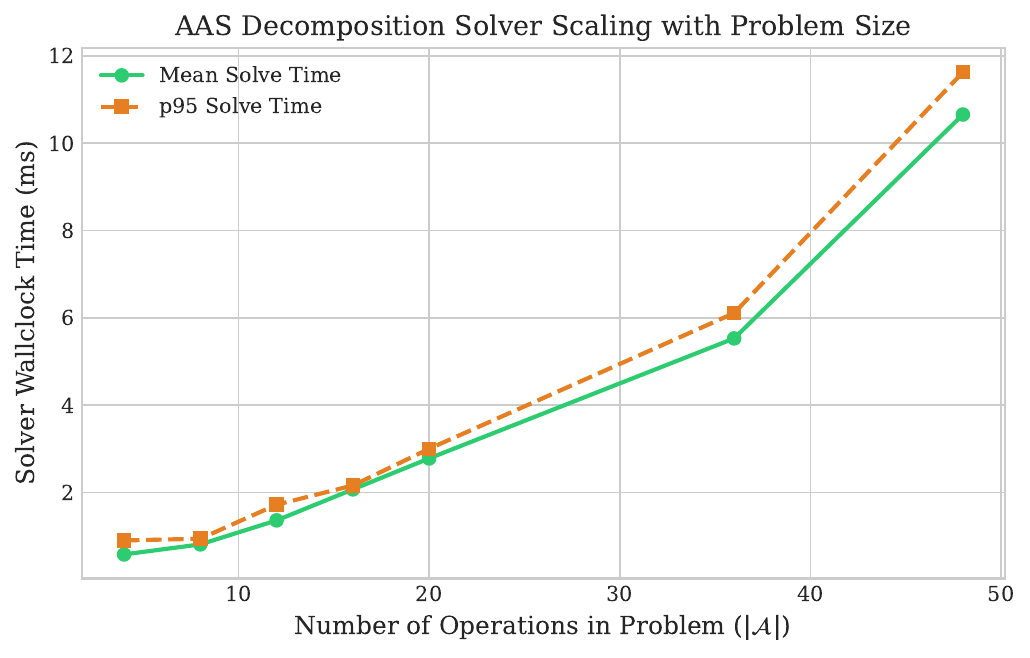}
\caption{RQ6: Solver time on seven generated sizes (4--48 operations). This ordinary scaling sweep generates no refinement cuts; the separate constructed suite below exercises both supported cut certificates.}
\label{fig:solver-scaling}
\end{figure}

\subsection{RQ6: Solver Scaling and Certified Refinement}
\label{sec:eval-scalability}

On seven feasible generated sizes, each with 2-of-$q$ obligations and separate witness roots, mean solver time rises from 0.58\,ms at four operations to 10.66\,ms at 48 operations (p95 11.63\,ms at the largest size). Every fixture admits a plan; this assertion prevents timing structural refusals. The sweep supports neither an asymptotic claim nor a bound for repeated refinement, and generates zero cuts. We therefore add 20 constructed, seeded small cases: ten violate a witness's intrinsic lifetime, and ten contain a cheap precedence chain that misses the deadline. In every case the unrefined master selects a temporally infeasible portfolio; one certified cut leads the decomposed solver to a feasible alternative in the second iteration. Both cut types occur ten times, and the refined cost matches the exact oracle in all 20. These cut-producing cases demonstrate the implemented mechanism, not its frequency in deployed workloads.

\input{tables/tab3_ablations.tex}

\begin{figure}[t]
\centering
\includegraphics[width=\columnwidth]{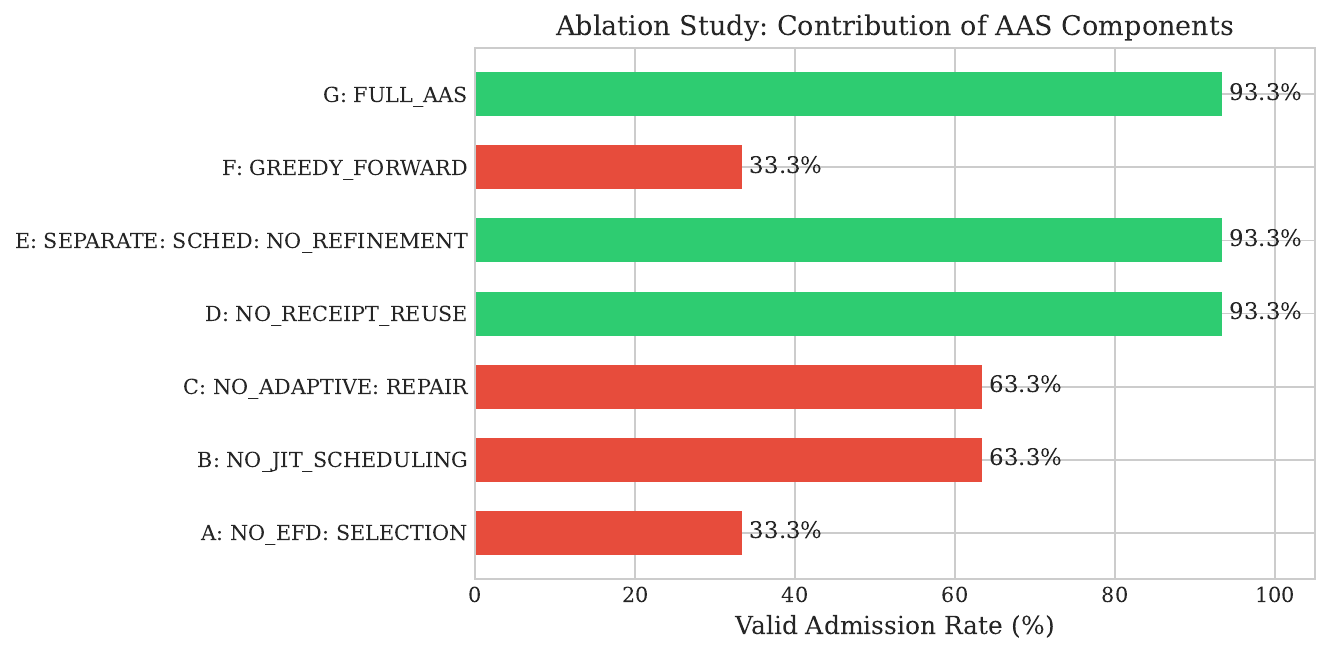}
\caption{General ablation outcomes on 30 generated instances per policy. No-reuse and no-refinement coincide with full \aas here because these trials do not trigger repair or refinement.}
\label{fig:ablation-waterfall}
\end{figure}

\subsection{Ablations and Interpretation}
\label{sec:ablations}

On 30 shared generated instances, full \aas admits 28/30. Greedy selection without \efd constraints admits 10/30 and produces 20 correlated-failure classifications. Forward scheduling admits 19/30 with 11 stale outcomes. The nominal ``No Adaptive Repair'' general ablation shares the forward policy implementation and therefore cannot isolate repair; the dedicated timeout study above supplies that intervention. No-reuse and no-refinement each admit 28/30, matching full \aas because the general workload activates neither mechanism. The paired repair and cut-producing studies identify their effects under controlled activation. These are mechanism tests on synthetic data, not estimates of production reliability.

%% file: tables/tab1_main_comparison.tex
\begin{table*}[t]
\centering
\footnotesize
\begin{tabularx}{\textwidth}{lYYYYYYY}
\toprule
\textbf{Scheduler Policy} & \textbf{Valid Dispatch} & \textbf{Stale Failures} & \textbf{Correlated Failures} & \textbf{Safe Refusal} & \textbf{Admitted Latency (s)} & \textbf{Admitted Cost (\$)} & \textbf{Replan (ms)} \\
\midrule
Serial Forward (ASAP) & 29.75\% (357/1200) & 0.0\% (0/1200) & 61.5\% (738/1200) & 8.75\% (105/1200) & 7.70 (p95: 7.70) & \$0.11 & --- \\
Greedy Cost Portfolio & 28.17\% (338/1200) & 1.58\% (19/1200) & 61.5\% (738/1200) & 8.75\% (105/1200) & 2.40 (p95: 2.40) & \$0.11 & --- \\
Static Parallel & 53.92\% (647/1200) & 36.67\% (440/1200) & 0.0\% (0/1200) & 9.42\% (113/1200) & 4.18 (p95: 6.51) & \$0.14 & --- \\
\textbf{AAS (Proposed)} & 89.58\% (1075/1200) & 1.0\% (12/1200) & 0.0\% (0/1200) & 9.42\% (113/1200) & 5.50 (p95: 12.40) & \$0.15 & --- \\
Constraint-Aware Forward & 53.92\% (647/1200) & 36.67\% (440/1200) & 0.0\% (0/1200) & 9.42\% (113/1200) & 4.18 (p95: 6.51) & \$0.14 & --- \\
Separate Sched (No Refinement) & 89.58\% (1075/1200) & 1.0\% (12/1200) & 0.0\% (0/1200) & 9.42\% (113/1200) & 5.50 (p95: 12.40) & \$0.15 & --- \\
\bottomrule
\end{tabularx}
\caption{Mode A on 1,200 generated instances per policy. Outcome counts use all instances; latency and cost means condition on valid admission. The gateway scores candidate manifests retrospectively.}
\label{tab:evaluation-summary}
\end{table*}

%% file: tables/tab2_mode_b_enforced.tex
\begin{table*}[t]
\centering
\small
\begin{tabularx}{\textwidth}{lYYYYYY}
\toprule
\textbf{Scheduler Policy} & \textbf{Valid Admitted} & \textbf{Stale Refused} & \textbf{Correlated Blocked} & \textbf{Explicit Safe Refusal} & \textbf{Admitted Latency (s)} & \textbf{Admitted Cost (\$)} \\
\midrule
Serial Forward (ASAP) & 29.75\% (357/1200) & 0.0\% (0/1200) & 61.5\% (738/1200) & 8.75\% (105/1200) & 7.70 & \$0.11 \\
Greedy Cost Portfolio & 28.17\% (338/1200) & 1.58\% (19/1200) & 61.5\% (738/1200) & 8.75\% (105/1200) & 2.40 & \$0.11 \\
Static Parallel & 53.92\% (647/1200) & 36.67\% (440/1200) & 0.0\% (0/1200) & 9.42\% (113/1200) & 4.18 & \$0.14 \\
\textbf{AAS (Proposed)} & 89.83\% (1078/1200) & 0.75\% (9/1200) & 0.0\% (0/1200) & 9.42\% (113/1200) & 5.50 & \$0.15 \\
Constraint-Aware Forward & 53.92\% (647/1200) & 36.67\% (440/1200) & 0.0\% (0/1200) & 9.42\% (113/1200) & 4.18 & \$0.14 \\
Separate Sched (No Refinement) & 89.58\% (1075/1200) & 1.0\% (12/1200) & 0.0\% (0/1200) & 9.42\% (113/1200) & 5.50 & \$0.15 \\
\bottomrule
\end{tabularx}
\caption{Mode B on 1,200 generated instances per policy. The gateway blocks invalid candidates before physical proposal execution. Latency and cost means condition on valid admission.}
\label{tab:mode-b-summary}
\end{table*}

%% file: tables/tab3_ablations.tex
\begin{table*}[t]
\centering
\small
\begin{tabularx}{\textwidth}{L{0.23\textwidth}YYYYYY}
\toprule
\textbf{Ablation Variant} & \textbf{Valid (\%)} & \textbf{Stale (\%)} & \textbf{Corr (\%)} & \textbf{Refusal (\%)} & \textbf{Lat (s)} & \textbf{Cost (\$)} \\
\midrule
A: no EFD selection & 33.33\% (10/30) & 0.0\% (0/30) & 66.67\% (20/30) & 0.0\% (0/30) & 2.40 & \$0.11 \\
B: no JIT scheduling & 63.33\% (19/30) & 36.67\% (11/30) & 0.0\% (0/30) & 0.0\% (0/30) & 4.19 & \$0.14 \\
C: no adaptive repair & 63.33\% (19/30) & 36.67\% (11/30) & 0.0\% (0/30) & 0.0\% (0/30) & 4.19 & \$0.14 \\
D: no receipt reuse & 93.33\% (28/30) & 6.67\% (2/30) & 0.0\% (0/30) & 0.0\% (0/30) & 5.15 & \$0.15 \\
E: no refinement & 93.33\% (28/30) & 6.67\% (2/30) & 0.0\% (0/30) & 0.0\% (0/30) & 5.15 & \$0.15 \\
F: greedy forward & 33.33\% (10/30) & 0.0\% (0/30) & 66.67\% (20/30) & 0.0\% (0/30) & 7.70 & \$0.11 \\
G: full AAS & 93.33\% (28/30) & 6.67\% (2/30) & 0.0\% (0/30) & 0.0\% (0/30) & 5.15 & \$0.15 \\
\bottomrule
\end{tabularx}
\caption{General ablation study on the same 30 generated instances per policy. Latency and cost means condition on valid admission.}
\label{tab:ablations}
\end{table*}

%% file: sections/10-discussion-limitations.tex
\section{Discussion and Limitations}
\label{sec:discussion-limitations}

\paragraph{Scope of Empirical Evidence.}
The three main scenario families still use generated operations, costs, exposure maps, latencies, and faults. Policy and mode records share instances. The new sensitivity studies pair base instances, fault injections, and operation-keyed latency draws, but retain the same generator and small sample sizes. The 20-instance oracle comparison shares a bounded temporal-search component; the 20 cut-producing cases are constructed to activate the supported certificates. The corrected timeout study demonstrates recovery under selected-operation faults, but does not test production state shifts, distributed cancellation, or transactional rollback. Independent workloads and live integrations are required to assess transfer and worst-case solve behavior.

\paragraph{Clock Synchronization and Distributed Skew.}
Theorem~\ref{thm:freshness-intersection} defines evidence validity evaluated at the execution gateway clock $C_{\mathrm{gw}}$, deducting maximum gateway-observer skew $\epsilon_{\mathrm{skew}}$ ($|C_i(t) - C_{\mathrm{gw}}(t)| \le \epsilon_{\mathrm{skew}}$). In distributed architectures lacking hardware-synchronized clocks (such as GPS/atomic reference clocks in TrueTime~\cite{corbett2013spanner}), relative clock drift between two independent observers $i$ and $j$ can reach $2\epsilon_{\mathrm{skew}}$. If peer observers compare timestamps directly without mediation through an authoritative gateway, soundness requires deducting $2\epsilon_{\mathrm{skew}}$:
\begin{equation}
\Delta t_{\mathrm{effective}}^{\mathrm{peer}}(e) = \max \left( 0, \; \Delta t(e) - 2\epsilon_{\mathrm{skew}} \right).
\end{equation}
If clock drift exceeds evidence validity, short-lived telemetry cannot be safely scheduled across disparate clock domains without atomic local re-observation at the gateway.

\paragraph{Proactive Pre-fetching vs. Invalidation Risk.}
Pre-fetching long-lived evidence (e.g., reading cluster topology 30\,s prior to failover) amortizes scheduling latency. However, if concurrent background mutations modify the underlying resource versions, the pre-fetched witness becomes obsolete. The \cac admission gateway enforces state version guards $G_q = \{\nu_s = \nu_{\mathrm{live}}\}$. If any guarded version changes between pre-fetching and dispatch, the certificate CAS fails, forcing an abort. Thus, aggressive pre-fetching in high-concurrency environments increases the probability of optimistic concurrency aborts.

\paragraph{Incomplete Epistemic Dependency Maps.}
The structural diversity guarantees of Section~\ref{sec:epistemic-scheduling} depend on the fidelity of the exposure map $\mathcal{X}: \mathcal{A} \to 2^{\mathbb{F}}$. If two ostensibly independent services share an undisclosed transitive dependency---such as a common DNS resolver, power bus, or base model---the calculated cut $\kappa_E(W,\Gamma_\omega)$ can overstate diversity. Establishing complete epistemic independence remains an empirical challenge.

\paragraph{Human-in-the-Loop Latency Asymmetry.}
Human approval can dominate automated evidence-acquisition latency. The current simulator does not model an asynchronous human response or re-fire short-lived telemetry after approval. Such a workflow would need an explicit provisional state, a new dispatch-time freshness check, and another admission decision.

%% file: sections/11-related-work.tex
\section{Related Work}
\label{sec:related-work}

\paragraph{Admission Control and Mediated Agent Execution.}
Access-control models such as ABAC~\cite{hu2014abac} decide whether a principal may invoke an operation. Proof-Carrying Code~\cite{necula1997pcc} made consumer-side proof checks explicit. Recent agent-tool work, including ToolGuardian~\cite{toolguardian2026}, characterizes tool behavior and applies task-aware runtime authorization. \cac~\cite{he2026cac} specifies risk-conditioned assurance obligations. \aas addresses the intervening planning problem: which observations to obtain and when to obtain them before mediated execution.

\paragraph{Verification Portfolios and Model Routing.}
FrugalGPT~\cite{chen2023frugalgpt} and RouteLLM~\cite{ong2024routellm} trade inference cost against model capability. VP-CONTROL~\cite{vpcontrol2026} studies cost-aware verifier portfolios, common-mode evidence failures, and commit-time guards, including a live HTTP/SQLite study. Its portfolio problem motivates source-aware selection. Our formulation additionally schedules observations with finite validity intervals and treats consequential evidence acquisition as a recursively admitted operation. The evaluations address different workloads and should not be read as a direct performance comparison.

\paragraph{Real-Time and Dependency Scheduling.}
Real-time scheduling~\cite{liulayland1973,buttazzo2011hard} and resource-constrained project scheduling~\cite{graham1979optimization} address deadlines, precedence, and capacity. Assurance scheduling adds a policy-defined validity interval to each selected witness. A feasible task schedule can therefore yield an inadmissible dispatch if evidence expires before the protected execution window ends.

\paragraph{Fault Tolerance and Epistemic Diversity.}
N-Version Programming~\cite{avizienis1985nversion} and Byzantine fault-tolerant replication~\cite{schneider1990byzantine} use redundancy against specified failures. EFD theory~\cite{he2026efd} defines a structural cut relative to the authorization rule and an explicit fault basis. \aas imports that cut into witness selection and admission, then adds temporal freshness and budget constraints. Its guarantee remains relative to the declared exposure map; an undisclosed common dependency can invalidate it.

%% file: sections/12-conclusion.tex
\section{Conclusion}
\label{sec:conclusion}

\aas formulates evidence acquisition for consequential agent actions as a joint selection and scheduling problem. Its witness assignment enforces declared fault-domain cuts; its temporal scheduler places expiring observations near dispatch; and its controller admits consequential diagnostics and repairs plans after state changes. The formal guarantees depend on modeled dependencies, bounded latencies, and exact search where optimality is claimed.

On three generated workload families, \aas raises valid candidate admission from 647/1,200 for constraint-aware forward scheduling to 1,075/1,200 and reduces stale candidates from 440 to 12. Paired sensitivity studies hold the base instance and latency draws fixed while varying freshness and diversity. In 20 constructed cut-producing cases, certified refinement recovers an oracle-matching feasible alternative; in 20 selected-operation timeout trials, receipt-aware repair admits 18 and lowers committed cost relative to resynthesis. The bounded oracle shares temporal search with the solver, and neither constructed cuts nor simulator timeouts establish production reliability. Independent workloads and live gateway integrations remain the next tests.

%% file: appendices/a-formal-proofs.tex
\section{Formal Proofs}
\label{app:proofs}

This appendix provides complete formal proofs for the theoretical statements established in the main manuscript.

\subsection{Proof of Theorem~\ref{thm:freshness-intersection} (Dispatch Freshness Intersection)}
\label{app:proof-freshness}

\begin{theorem}[Restatement of Theorem~\ref{thm:freshness-intersection}]
Fix an executed assurance plan $\pi$ with completed operations $V_\pi$, actual completions $\{C_\pi(a)\}$, observation timestamps $\{t_{\mathrm{obs}}(e)\}$, and composite witness manifest $W_{\mathrm{all}}$. Let $V_{\mathrm{req}} \subseteq V_\pi$ denote mandatory operations prior to dispatch ($V_{\mathrm{manifest}} \cup \operatorname{Anc}_{D_{\mathrm{op}}}(V_{\mathrm{manifest}})$), and $\delta_{\mathrm{ready}} = \delta_{\mathrm{verify}} + \delta_{\mathrm{mint}} + \delta_{\mathrm{disp}}$.
A dispatch timestamp satisfying readiness and freshness throughout the protected execution window exists, before imposing the independent deadline and budget guards, if and only if condition~\eqref{eq:freshness-intersection-condition} holds:
\begin{equation}
\begin{aligned}
\max_{a \in V_{\mathrm{req}}} C_\pi(a) &+ \delta_{\mathrm{ready}} + \delta_{\mathrm{exec}} + \epsilon_{\mathrm{skew}} \\
&\le \min_{e \in W_{\mathrm{all}}} \left( t_{\mathrm{obs}}(e) + \Delta t_{\mathrm{eff}}(e) \right).
\end{aligned}
\label{eq:app-freshness-condition}
\end{equation}
which is equivalent to pairwise operational skew bound~\eqref{eq:pairwise-freshness-skew} and observation skew bound~\eqref{eq:pairwise-obs-skew}. Non-witness prerequisites $u \in V_{\mathrm{req}} \setminus V_{\mathrm{manifest}}$ participate directly in readiness bound~\eqref{eq:pairwise-freshness-skew}.
\end{theorem}

\begin{proof}
By Definition~\ref{def:prospective-admissibility} and Equation~\eqref{eq:execution-window-validity}, dispatch requires that for every witness receipt $e \in W_{\mathrm{all}}$, the entire protected execution window $[t_{\mathrm{dispatch}}, \; t_{\mathrm{dispatch}} + \delta_{\mathrm{exec}}]$ is contained within the conservative gateway-evaluated validity interval $I(e) = [t_{\mathrm{obs}}(e) - \epsilon_{\mathrm{skew}}, \; t_{\mathrm{obs}}(e) + \Delta t_{\mathrm{eff}}(e) - \epsilon_{\mathrm{skew}}]$.

Furthermore, proposal dispatch cannot legally occur until all mandatory operations in $V_{\mathrm{req}}$ (both witness producers and precedence prerequisites) have completed execution, receipts have been emitted and verified ($\delta_{\mathrm{verify}}$), the admission certificate has been minted by \cac ($\delta_{\mathrm{mint}}$), and gateway mediation delay ($\delta_{\mathrm{disp}}$) has elapsed:
\begin{equation}
t_{\mathrm{dispatch}} \ge \max_{a \in V_{\mathrm{req}}} C_\pi(a) + \delta_{\mathrm{ready}},
\label{eq:app-dispatch-lower}
\end{equation}
where $\delta_{\mathrm{ready}} = \delta_{\mathrm{verify}} + \delta_{\mathrm{mint}} + \delta_{\mathrm{disp}}$.

Simultaneously, the completion of the protected execution window cannot exceed the gateway-evaluated expiration timestamp of any witness receipt:
\begin{equation}
t_{\mathrm{dispatch}} + \delta_{\mathrm{exec}} \le \min_{e \in W_{\mathrm{all}}} \left( t_{\mathrm{obs}}(e) + \Delta t_{\mathrm{eff}}(e) - \epsilon_{\mathrm{skew}} \right).
\label{eq:app-dispatch-upper}
\end{equation}

Combining~\eqref{eq:app-dispatch-lower} and~\eqref{eq:app-dispatch-upper}, an admissible dispatch timestamp $t_{\mathrm{dispatch}}$ exists if and only if the feasible dispatch interval $[t_{\mathrm{lower}}, t_{\mathrm{upper}}]$ is non-empty, where:
\begin{align}
t_{\mathrm{lower}} &= \max_{a \in V_{\mathrm{req}}} C_\pi(a) + \delta_{\mathrm{ready}}, \\
t_{\mathrm{upper}} &= \min_{e \in W_{\mathrm{all}}} \left( t_{\mathrm{obs}}(e) + \Delta t_{\mathrm{eff}}(e) - \epsilon_{\mathrm{skew}} \right) - \delta_{\mathrm{exec}}.
\end{align}

The interval $[t_{\mathrm{lower}}, t_{\mathrm{upper}}]$ is non-empty if and only if $t_{\mathrm{lower}} \le t_{\mathrm{upper}}$:
\begin{equation}
\begin{aligned}
\max_{a \in V_{\mathrm{req}}} C_\pi(a) &+ \delta_{\mathrm{ready}} \le \min_{e \in W_{\mathrm{all}}} \Big( t_{\mathrm{obs}}(e) \\
&+ \Delta t_{\mathrm{eff}}(e) - \epsilon_{\mathrm{skew}} \Big) - \delta_{\mathrm{exec}}.
\end{aligned}
\end{equation}
Moving constants $\delta_{\mathrm{ready}}$, $\delta_{\mathrm{exec}}$, and $\epsilon_{\mathrm{skew}}$ directly establishes Equation~\eqref{eq:app-freshness-condition}.

To establish equivalence with pairwise condition~\eqref{eq:pairwise-freshness-skew}:
\begin{itemize}
    \item $(\implies)$ Suppose condition~\eqref{eq:app-freshness-condition} holds. For any operation $a \in V_{\mathrm{req}}$ and any witness $e_i \in W_{\mathrm{all}}$, we have $C_\pi(a) \le \max_{u \in V_{\mathrm{req}}} C_\pi(u)$ and $\min_{e \in W_{\mathrm{all}}} (t_{\mathrm{obs}}(e) + \Delta t_{\mathrm{eff}}(e)) \le t_{\mathrm{obs}}(e_i) + \Delta t_{\mathrm{eff}}(e_i)$. Substituting these bounds gives $C_\pi(a) + \delta_{\mathrm{ready}} + \delta_{\mathrm{exec}} + \epsilon_{\mathrm{skew}} \le t_{\mathrm{obs}}(e_i) + \Delta t_{\mathrm{eff}}(e_i)$, yielding $C_\pi(a) - t_{\mathrm{obs}}(e_i) \le \Delta t_{\mathrm{eff}}(e_i) - \delta_{\mathrm{ready}} - \delta_{\mathrm{exec}} - \epsilon_{\mathrm{skew}}$. For $a \in V_{\mathrm{manifest}}$, substituting $C_\pi(a) = t_{\mathrm{obs}}(e_j) + \ell_{\mathrm{post}}(a)$ gives Equation~\eqref{eq:pairwise-obs-skew}.
    \item $(\impliedby)$ Conversely, suppose~\eqref{eq:pairwise-freshness-skew} holds for all operations $a \in V_{\mathrm{req}}$ and witnesses $e_i \in W_{\mathrm{all}}$. Choosing $a^* = \arg\max_{a \in V_{\mathrm{req}}} C_\pi(a)$ and $e_i^* = \arg\min_{e \in W_{\mathrm{all}}} (t_{\mathrm{obs}}(e) + \Delta t_{\mathrm{eff}}(e))$ directly recovers condition~\eqref{eq:app-freshness-condition}.
\end{itemize}
This completes the proof.
\end{proof}

\subsection{Proof of Theorem~\ref{thm:np-hard} (NP-Hardness of ASP-DEC)}
\label{app:proof-nphard}

\begin{theorem}[Restatement of Theorem~\ref{thm:np-hard}]
The decision version of the Assurance Scheduling Problem (\textsc{ASP-DEC}) is NP-hard, even when operational precedence is empty ($D_{\mathrm{op}} = \emptyset$) and operation latencies are uniform.
\end{theorem}

\begin{proof}
We establish NP-hardness by constructing a polynomial-time reduction from the classical NP-complete problem \textsc{Minimum Weight Set Cover} (WSC)~\cite{karp1972reducibility} to \textsc{ASP-DEC}.

\paragraph{The WSC Decision Problem.}
An instance of \textsc{Minimum Weight Set Cover} is defined by a 4-tuple $\langle \mathcal{U}, \mathcal{S}, w, K \rangle$: a finite universe $\mathcal{U} = \{u_1, \ldots, u_n\}$, a collection of subsets $\mathcal{S} = \{S_1, \ldots, S_m\}$ with $S_j \subseteq \mathcal{U}$ and $\bigcup_{j=1}^m S_j = \mathcal{U}$, a non-negative cost function $w: \mathcal{S} \to \mathbb{R}_{\ge 0}$, and a target cost threshold $K \in \mathbb{R}_{\ge 0}$. The problem asks whether there exists a sub-collection $\mathcal{S}' \subseteq \mathcal{S}$ such that $\bigcup_{S_j \in \mathcal{S}'} S_j = \mathcal{U}$ and $\sum_{S_j \in \mathcal{S}'} w(S_j) \le K$.

\paragraph{Polynomial-Time Construction.}
Given an arbitrary instance $\langle \mathcal{U}, \mathcal{S}, w, K \rangle$ of WSC, we construct an instance of \textsc{ASP-DEC} $\mathcal{P} = \langle \Omega, \mathcal{A}, B, \mathcal{D}, \mathcal{T} \rangle$ with cost bound $K$ as follows:
\begin{enumerate}
    \item \textbf{Assurance Obligations $\Omega$:} For each $u_i \in \mathcal{U}$, create obligation $\omega_i = \langle \phi_{\omega_i}, \mathcal{E}_{\omega_i}, \Delta t_{\omega_i}, \kappa_E^{\min}(\omega_i), \mathcal{Q}_{\omega_i} \rangle$ with eligible evidence $\mathcal{E}_{\omega_i} = \{\varepsilon_{\mathrm{std}}\}$, freshness lifespan $\Delta t_{\omega_i} = 10$, diversity threshold $\kappa_E^{\min}(\omega_i) = 1$, and quorum threshold $k_{\omega_i} = 1$. Thus $|\Omega| = |\mathcal{U}| = n$.
    \item \textbf{Assurance Operations $\mathcal{A}$:} For each subset $S_j \in \mathcal{S}$, create an operation $a_j \in \mathcal{A}$ with financial cost $c(a_j) = w(S_j)$, latencies $\ell(a_j) = \hat{\ell}(a_j) = 1$, cognitive tokens $\operatorname{tok}(a_j) = 0$, risk $\rho(a_j) = \mathbf{0} \prec \tau_\Pi$ (non-consequential), potential coverage $\operatorname{Cov}(a_j) = \{ (\omega_i, \varepsilon_{\mathrm{std}}) \mid u_i \in S_j \}$, lifespan $\Delta t(a_j) = 10$, prerequisites $\operatorname{Pre}(a_j) = \emptyset$, and distinct epistemic exposure $\mathcal{X}(a_j) = \{f_j\}$ for distinct $f_j \in \mathbb{F} = \{f_1, \ldots, f_m\}$. Thus $|\mathcal{A}| = |\mathcal{S}| = m$.
    \item \textbf{Dependencies $\mathcal{D}$:} Set $D_{\mathrm{op}} = \emptyset$ and $D_{\mathrm{epi}} = \{ (a_j, f_j) \mid 1 \le j \le m \}$.
    \item \textbf{Budgets $B$ and Temporal Constraints $\mathcal{T}$:} Set $B_{\$} = K$, $B_{\mathrm{lat}} = 10$, $B_{\mathrm{tok}} = 0$, $B_{\mathrm{risk}} = 0$, and $K_{\max} = \infty$. Set $t_0 = 0$, $T_{\mathrm{dead}} = 10$, $\delta_{\mathrm{exec}} = 1$, $\delta_{\mathrm{verify}}=\delta_{\mathrm{mint}}=\delta_{\mathrm{disp}} = 0$, and $\epsilon_{\mathrm{skew}} = 0$.
\end{enumerate}
This construction maps universe elements to obligations and subsets to multi-obligation operations in time $O(|\mathcal{U}| \cdot |\mathcal{S}|)$, which is polynomial in the input size.

\paragraph{Equivalence Proof.}
We now prove that there exists a valid set cover $\mathcal{S}' \subseteq \mathcal{S}$ with $\sum_{S_j \in \mathcal{S}'} w(S_j) \le K$ if and only if there exists a prospectively admissible plan $\pi \models_{\mathrm{pros}} \mathcal{P}$ with $\operatorname{Cost}(\pi) \le K$:

\begin{itemize}
    \item $(\implies)$ Suppose $\mathcal{S}' \subseteq \mathcal{S}$ is a valid set cover with $\sum_{S_j \in \mathcal{S}'} w(S_j) \le K$.
    Construct assurance plan $\pi = \langle V_\pi, E_\pi, \sigma_\pi, \mathcal{M}_\pi \rangle$ as follows:
    \begin{itemize}
        \item Set of operations: $V_\pi = \{ a_j \in \mathcal{A} \mid S_j \in \mathcal{S}' \}$,
        \item Precedence edges: $E_\pi = \emptyset$,
        \item Scheduled start times: $\sigma_\pi(a_j) = 0$ for all $a_j \in V_\pi$,
        \item Witness mapping: For each $\omega_i \in \Omega$, since $\mathcal{S}'$ covers $\mathcal{U}$, there exists at least one $S_j \in \mathcal{S}'$ such that $u_i \in S_j$. Choose one such $a_j \in V_\pi$ and set $\mathcal{M}_\pi(\omega_i) = \{a_j\}$.
    \end{itemize}
    We verify prospective admissibility of $\pi$:
    \begin{enumerate}
        \item \emph{Precedence:} $E_\pi = \emptyset$, trivially satisfied.
        \item \emph{Latency \& Deadlines:} All operations start at 0 and have latency $\hat{\ell}=1$, so $\hat{C}_\pi(a_j) = 1$. With $\delta_{\mathrm{disp}} = 0$, target dispatch is $t_{\mathrm{dispatch}} = 1$. Then $t_{\mathrm{dispatch}} + \delta_{\mathrm{exec}} = 1 + 1 = 2 \le 10 = T_{\mathrm{dead}}$, satisfying deadline feasibility. Latency is $1 - 0 = 1 \le B_{\mathrm{lat}}$.
        \item \emph{Cost Budget:} $\operatorname{Cost}(\pi) = \sum_{a_j \in V_\pi} c(a_j) = \sum_{S_j \in \mathcal{S}'} w(S_j) \le K = B_{\$}$. Token and risk budgets are identically zero, satisfying all budget invariants.
        \item \emph{Quorum Coverage:} For each $\omega_i$, $(\omega_i, \varepsilon_{\mathrm{std}}) \in \operatorname{Cov}(a_j)$ because $u_i \in S_j$. The quorum requirement is $|\mathcal{M}_\pi(\omega_i)| = 1 \ge k_{\omega_i} = 1$.
        \item \emph{Freshness:} Effective lifespan is $\Delta t_{\mathrm{eff}} = \min(10, 10) = 10$. Then $t_{\mathrm{dispatch}} + \delta_{\mathrm{exec}} = 2 \le \sigma_\pi(a_j) + \Delta t_{\mathrm{eff}} = 0 + 10 = 10$.
        \item \emph{Epistemic Diversity:} Each obligation requires $\kappa_E^{\min}(\omega_i) = 1$. The singleton witness $\{a_j\}$ is decisive for its 1-of-1 rule and is exposed by $\{f_j\}$, so $\kappa_E(\mathcal{M}_\pi(\omega_i),\Gamma_{\omega_i}) = 1 \ge \kappa_E^{\min}(\omega_i)$.
        \item \emph{Diagnostics:} All operations have risk $\rho = \mathbf{0} \prec \tau_\Pi$, so no diagnostic sub-plans are required.
    \end{enumerate}
    Thus $\pi \models_{\mathrm{pros}} \mathcal{P}$ with $\operatorname{Cost}(\pi) \le K$.

    \item $(\impliedby)$ Conversely, suppose there exists a prospectively admissible plan $\pi = \langle V_\pi, E_\pi, \sigma_\pi, \mathcal{M}_\pi \rangle$ for $\mathcal{P}$ with $\operatorname{Cost}(\pi) \le K$.
    Define the sub-collection $\mathcal{S}' = \{ S_j \in \mathcal{S} \mid a_j \in V_\pi \}$.
    By the resource budget invariant~\eqref{eq:cost-budget}:
    \begin{equation*}
    \sum_{S_j \in \mathcal{S}'} w(S_j) = \sum_{a_j \in V_\pi} c(a_j) = \operatorname{Cost}(\pi) \le K.
    \end{equation*}
    Now consider any arbitrary universe element $u_i \in \mathcal{U}$. In the constructed ASP instance, $u_i$ corresponds to obligation $\omega_i \in \Omega$.
    By prospective quorum coverage (Definition~\ref{def:prospective-admissibility}), $\pi$ must assign candidate witness operations $\mathcal{M}_\pi(\omega_i) \subseteq V_\pi$ satisfying:
    \begin{align*}
    &|\mathcal{M}_\pi(\omega_i)| \ge k_{\omega_i} = 1 \quad \text{and} \\
    &\forall a \in \mathcal{M}_\pi(\omega_i), \; (\omega_i, \varepsilon_{\mathrm{std}}) \in \operatorname{Cov}(a).
    \end{align*}
    Thus, there exists at least one operation $a_j \in V_\pi$ such that $(\omega_i, \varepsilon_{\mathrm{std}}) \in \operatorname{Cov}(a_j)$.
    By our polynomial construction, $(\omega_i, \varepsilon_{\mathrm{std}}) \in \operatorname{Cov}(a_j)$ holds if and only if $u_i \in S_j$.
    Since $a_j \in V_\pi$, the corresponding subset $S_j$ belongs to $\mathcal{S}'$.
    Because this holds for every $u_i \in \mathcal{U}$, we conclude that:
    \begin{equation*}
    \bigcup_{S_j \in \mathcal{S}'} S_j = \mathcal{U}.
    \end{equation*}
    Hence, $\mathcal{S}'$ is a valid set cover of $\mathcal{U}$ with total weight at most $K$.
\end{itemize}

Since \textsc{Minimum Weight Set Cover} is NP-complete, the decision problem \textsc{ASP-DEC} is NP-hard.

\paragraph{Complexity Analysis: Multi-Obligation Coverage vs. EFD Cut Verification.}
We emphasize that Theorem~\ref{thm:np-hard} formally establishes that \textsc{ASP-DEC} is NP-hard due to multi-obligation witness coverage via reduction from Minimum Weight Set Cover.
Regarding epistemic diversity: (1)~\emph{Fixed thresholds:} To verify $\kappa_E(W,\Gamma_\omega)\ge h$, enumerate each fault coalition $C\subseteq\mathbb F$ with $|C|<h$ and check whether it exposes fewer than $k_\omega$ assigned witnesses. For fixed $h$, this costs $O(|W|\,h\,|\mathbb F|^{h-1})$. (2)~\emph{ILP separation:} At $h=2$, singleton-root cuts can be instantiated statically. At $h=3$, root pairs are also polynomial to enumerate. (3)~\emph{Decision-rule bound:} The completed local-root basis implies $\kappa_E(W,\Gamma_\omega)\le k_\omega$ once a quorum exists. Hence a threshold above the quorum cardinality is infeasible before scheduling. The NP-hardness reduction uses $h=k_\omega=1$ and is unaffected by this correction.
\end{proof}

\begin{table*}[t]
\centering
\footnotesize
\begin{tabularx}{\textwidth}{L{2.5cm}L{3.2cm}cccC{2.0cm}Y}
\toprule
\textbf{Operation Class} & \textbf{Example Implementation} & \textbf{Mean Cost} & \textbf{Nominal Latency} & \textbf{Lifespan $\Delta t$} & \textbf{Consequential? ($\rho \ge \tau_\Pi$)} & \textbf{Primary Epistemic Fault Domains ($\mathcal{X}$)} \\
\midrule
Passive Telemetry & Prometheus scrape / eBPF & \$0.0001 & 40\,ms & 5\,s & No & Local kernel, metrics daemon \\
Log Audit & Loki / OpenSearch query & \$0.002 & 250\,ms & 30\,s & No & Log forwarder, search index \\
Hardware Attestation & TPM quote / AWS Nitro quote & \$0.005 & 350\,ms & 300\,s & No & Hardware root-of-trust, CA \\
Static SMT Proof & Z3 / Dafny verification & \$0.02 & 1.8\,s & $\infty$ & No & Solver binary, CPU architecture \\
Redundant Model & Independent LLM judge & \$0.03 & 1.2\,s & 60\,s & No & Model weights, inference stack \\
Consensus Query & Raft / Paxos quorum read & \$0.001 & 80\,ms & 10\,s & No & Consensus network quorum \\
Filesystem Freeze & LVM / ZFS snapshot probe & \$0.01 & 1.5\,s & 15\,s & \textbf{Yes} & Storage controller, I/O bus \\
Network Probe & Synthetic BGP / ping burst & \$0.005 & 600\,ms & 8\,s & \textbf{Yes} & Network fabric, firewall \\
Failover Dry-Run & Staging replica promotion & \$0.25 & 6.5\,s & 45\,s & \textbf{Yes} & Database engine, replication stream \\
Human Sign-off & Slack / PagerDuty sign-off & \$5.00 & 180\,s & 600\,s & No & Human operator, IdP token \\
\bottomrule
\end{tabularx}
\caption{Taxonomy of heterogeneous assurance operations in \aas, detailing financial expense, observation latency, temporal decay, and operational consequence.}
\label{tab:operation-catalogue}
\end{table*}

\subsection{Proof of Theorem~\ref{thm:benders-optimality} (Convergence and Conditional Global Optimality)}
\label{app:proof-optimality}

\begin{theorem}[Restatement of Theorem~\ref{thm:benders-optimality}]
The decomposed Logic-Based Benders Decomposition (LBBD) algorithm terminates in a finite number of iterations. Furthermore, for any fully modeled or flattened problem instance $\mathcal{P}$, if the master solver executes to completion without timeout ($T_{\mathrm{ctrl}} = \infty$) and the subproblem is evaluated authoritatively, the returned plan $\pi^*$ is \textbf{globally cost-optimal} over all prospectively admissible plans:
$\operatorname{Cost}(\pi^*) = \min_{\pi \in \Pi_{\mathrm{adm}}(\mathcal{P})} \sum_{a \in V_\pi} c(a)$.
If the solver reaches computation timeout $T_{\mathrm{ctrl}} < \infty$ and returns the best incumbent integer solution, or falls back to the greedy heuristic (Section~\ref{sec:greedy-efd}), the resulting plan is \textbf{approximately optimized and prospectively feasible}, with no claim of global optimality. Consequential diagnostic sub-plans admitted dynamically at runtime are coordinated outside this static optimality theorem.
\end{theorem}

\begin{proof}
We divide the proof into two parts: finite termination and global optimality.

\paragraph{Part 1: Finite Termination.}
Let $|\mathcal{A}|$ denote candidate assurance operations and $|\Omega|$ denote mandatory obligations. The decision space of the Master ILP is bounded by binary assignments $\mathbf{x} \in \{0, 1\}^{|\mathcal{A}|}, \mathbf{y} \in \{0, 1\}^{|\mathcal{A}| \times |\Omega|}$. The total number of integer assignment configurations is at most $2^{|\mathcal{A}| \cdot (|\Omega| + 1)}$, which is strictly finite.
At each iteration $k$ where candidate solution $(\mathbf{x}^{(k)}, \mathbf{y}^{(k)})$ is rejected by the Authoritative Temporal Solver:
(1)~If rejection is caused by an unavoidable directed precedence chain whose critical path exceeds available lead time ($t_{\mathrm{CP}} + \delta_{\mathrm{ready}} > T_{\mathrm{dead}} - t_0 - \delta_{\mathrm{exec}}$), structural cut~\eqref{eq:ilp-chain-cut} ($\sum_{a \in V_{\mathrm{chain}}} x_a \le |V_{\mathrm{chain}}| - 1$) is added. This strictly forbids $V_{\mathrm{chain}}$ and all supersets from being generated in future iterations, pruning at least one point in $\{0, 1\}^{|\mathcal{A}|}$.
(2)~Otherwise, the failure is caused by freshness expiration or concurrency collision among assigned witnesses. When certified by an exact subproblem evaluation, the solver identifies active conflicting assignment pairs $A_{\mathrm{conflict}} = \{(a, \omega) \mid y_{a, \omega}^{(k)} = 1\}$, and appends assignment cut~\eqref{eq:ilp-assign-cut} ($\sum_{(a, \omega) \in A_{\mathrm{conflict}}} y_{a, \omega} \le |A_{\mathrm{conflict}}| - 1$). If heuristic evaluation exhausts without certified proof ($\textsc{NotFound}$), a safe full-candidate combinatorial no-good cut ($\sum_{a: x_a^{(k)}=1} (1 - x_a) + \sum_{a: x_a^{(k)}=0} x_a \ge 1$) is appended. Each cut strictly eliminates candidate assignment $(\mathbf{x}^{(k)}, \mathbf{y}^{(k)})$ without pruning unvisited alternatives.
Because each iteration strictly eliminates at least one previously unvisited integer assignment point and the total number of integer assignments is finite, the LBBD loop must terminate in a finite number of iterations (at most $2^{|\mathcal{A}| \cdot (|\Omega| + 1)}$).

\paragraph{Part 2: Soundness of Cuts and Global Optimality.}
To establish global cost optimality upon non-timeout termination:
\begin{enumerate}
    \item \emph{Outer Relaxation Property:} The initial Master ILP~\eqref{eq:ilp-obj}--\eqref{eq:ilp-binary} without conflict cuts enforces all necessary conditions of prospective admissibility: operational prerequisite closure ($x_a \le x_u$), multi-obligation quorum coverage, \efd diversity cuts, and resource budgets. Thus, the feasible region of the Master ILP is an outer relaxation of $\Pi_{\mathrm{adm}}(\mathcal{P})$.
    \item \emph{Soundness of Pruning:} A cut is valid if it prunes no prospectively admissible plan $\pi \in \Pi_{\mathrm{adm}}(\mathcal{P})$. Structural cuts are sound because if a subset of operations $V_{\mathrm{chain}}$ has a topological critical path exceeding $T_{\mathrm{dead}} - t_0 - \delta_{\mathrm{ready}} - \delta_{\mathrm{exec}}$, latency monotonicity guarantees that no schedule containing $V_{\mathrm{chain}}$ can complete before the hard deadline under conservative latencies. Assignment cuts are sound when subproblem infeasibility is certified authoritatively by an exact temporal oracle (evaluating all discrete topological serializations and continuous dispatch intervals via STN enumeration). Tier~1 and Tier~2 backward heuristics operate on a discrete grid and provide polynomial-time sound schedule admission, but heuristic exhaustion ($\textsc{NotFound}$) cannot mathematically prove continuous infeasibility; hence, heuristic failures use safe full-candidate combinatorial cuts to avoid unsound over-pruning.
    \item \emph{Optimality of the First Feasible Candidate:} The Master ILP minimizes the exact non-decreasing linear objective function $Z(\mathbf{x}) = \sum_{a \in \mathcal{A}} c(a) x_a$. When subproblems are evaluated with exact completeness, cuts prune only provably infeasible points, so the sequence of optimal objective values returned by the Master ILP is monotonically non-decreasing: $Z(\mathbf{x}^{(1)}) \le Z(\mathbf{x}^{(2)}) \le \cdots \le Z(\mathbf{x}^*)$. When candidate $\mathbf{x}^*$ is first verified as temporally feasible, it satisfies all constraints of $\Pi_{\mathrm{adm}}(\mathcal{P})$ and its cost is a valid lower bound on all unvisited feasible candidates. Hence, $\operatorname{Cost}(\pi^*) = \min_{\pi \in \Pi_{\mathrm{adm}}(\mathcal{P})} \sum_a c(a)$, proving conditional global optimality on the modeled instance $\mathcal{P}$.
\end{enumerate}
This completes the proof.
\end{proof}

%% file: appendices/b-operation-catalogue.tex
\section{Catalogue of Assurance Operations}
\label{app:operation-catalogue}

Table~\ref{tab:operation-catalogue} provides a representative taxonomy of assurance operations supported in the \aas reference implementation, categorized across operational risk, latency, cost, and typical validity windows.

%% file: appendices/c-open-problems.tex
\section{Open Problems and Roadmap}
\label{app:open-problems}

\noindent\textbf{Negative Knowledge Induction (The Hardknock Bridge).}
When an assurance plan fails to reach admission---or an admitted action triggers an unexpected incident---the execution trace can reveal omitted epistemic dependencies and correlated failures. A remaining research question is how to use such traces to revise future admission obligations $\Omega$ without overfitting to individual incidents.

\noindent\textbf{Human-Referential Fidelity and Intent Drift.}
While \aas models physical state freshness and epistemic diversity, it treats intent predicate $\phi_\omega$ as fixed. In long-running autonomous workflows, operator intent or business context may drift while evidence is gathered. Formulating formal fidelity metrics connecting physical evidence directly to human-referential intent remains an active research direction.

\noindent\textbf{Multi-Tenant Cross-Sovereignty Scheduling.}
When orchestrating operations across multiple cloud domains, assurance must navigate heterogeneous trust roots, non-comparable clocks, and data-residency boundaries. Extending \aas to federated epistemic markets is a compelling direction for future architectures.